\documentclass[12pt]{amsart}
\usepackage{amscd,amssymb,cite,amsmath,latexsym,comment,epic,eepic,euscript}
\usepackage{graphicx}
\usepackage{enumerate}
\usepackage[initials]{amsrefs}
\usepackage[all]{xy}
\usepackage{tikz}
\usetikzlibrary{decorations.markings,arrows, decorations.pathreplacing}

\theoremstyle{plain}
\newtheorem{thm}{Theorem}[section]
\newtheorem{cor}[thm]{Corollary} %%Delete [thm] to re-start numbering
\newtheorem{lemma}[thm]{Lemma} %%Delete [thm] to re-start numbering
\newtheorem{prop}[thm]{Proposition}

\theoremstyle{remark}

\newtheorem{remark}[thm]{Remark}

\theoremstyle{definition}
\newtheorem{defi}[thm]{Definition}
\newtheorem{example}[thm]{Example}

\newtheorem{notation}[thm]{Notation}
\newtheorem{obs}[thm]{Observation}

\usepackage{amscd,amssymb,comment,epic,eepic,euscript}
\usepackage{graphicx}
\usepackage{enumerate}
\usepackage[initials]{amsrefs}
\usepackage[all]{xy}
\newcount\theTime
\newcount\theHour
\newcount\theMinute
\newcount\theMinuteTens
\newcount\theScratch
\theTime=\number\time
\theHour=\theTime
\divide\theHour by 60
\theScratch=\theHour
\multiply\theScratch by 60
\theMinute=\theTime
\advance\theMinute by -\theScratch
\theMinuteTens=\theMinute
\divide\theMinuteTens by 10
\theScratch=\theMinuteTens
\multiply\theScratch by 10
\advance\theMinute by -\theScratch

\def\today{{\number\day\space
 \ifcase\month\or
  January\or February\or March\or April\or May\or June\or
  July\or August\or September\or October\or November\or December\fi
 \space\number\year}}

\newcommand{\beq}{\begin{equation}}
\newcommand{\eeq}{\end{equation}}
\newcommand{\bre}{\begin{remark}}
\newcommand{\ere}{\end{remark}}
\newcommand{\beqno}[1]{\begin{equation}\label{#1}}

\newcommand{\mref}[1]{(\ref{#1})}

\newcommand\Z{\mathbb Z}
\DeclareMathOperator{\supp}{supp}
\DeclareMathOperator{\Hess}{Hess}
\newenvironment{prmainthm}{\paragraph{\textit{Proof of Theorem 2.11}}}{\hfill$\square$}
\newenvironment{prmainprop}{\paragraph{\textit{Proof of Proposition 4.1}}}{\hfill$\square$}
\newenvironment{prmainlem}{\paragraph{\textit{Proof of Lemma 5.3}}}{\hfill$\square$}
\begin{document}

\title[Spectral edge interior case] %\timeanddate]
{Green's function asymptotics near the internal edges of spectra of periodic elliptic operators. Spectral gap interior.}

\author[Kha]{Minh Kha$^{*}$}
\address{M.K., Department of Mathematics, Texas A\&M University,
College Station, TX 77843-3368, USA}
\email{kha@math.tamu.edu}

\author[Kuchment]{Peter Kuchment$^{*}{}^\natural$}
\address{P.K., Department of Mathematics, Texas A\&M University,
College Station, TX 77843-3368, USA}
\email{kuchment@math.tamu.edu}

\author[Raich]%, \timeanddate]
{Andrew Raich$^{**}$}
\address{A.R., Department of Mathematical Sciences, University of Arkansas, Fayetteville, AR 72701,
USA}
\email{araich@uark.edu}
\thanks{\footnotesize ${}^{*}$ The first two authors were partially supported by the NSF grant DMS-1517938. \\${}^{**}$ The third author is partially funded by NSF grant DMS-1405100.\\
 ${}^{\natural}$ The second author would like to thank the Isaac Newton Institute for Mathematical Sciences, Cambridge, for support and hospitality during the programme Periodic and Ergodic Problems, where work on this paper was undertaken.}

%\date{\timeanddate}

\begin{abstract}
Precise asymptotics known for the Green function of the Laplacian have found their analogs for bounded below periodic elliptic operators of the second-order below and at the bottom of the spectrum. Due to the band-gap structure of the spectra of such operators, the question arises whether similar results can be obtained near or at the edges of spectral gaps. In a previous work, two of the authors considered the case of a spectral edge. The main result of this article is finding such asymptotics near a gap edge, for ``generic'' periodic elliptic operators of second-order with real coefficients in dimension $d \geq 2$, when the gap edge occurs at a symmetry point of the Brillouin zone.
\end{abstract}

\maketitle
%%%%%%%%%%%%%%%%%%%%%%%%%%
\section{Introduction}
%%%%%%%%%%%%%%%%%%%%%%%%%%

The behavior at infinity of the Green function of the Laplacian in $\mathbb{R}^n$ outside and at the boundary of its spectrum is well known. Analogous results below and at the lower boundary of the spectrum have been established for bounded below periodic elliptic operators of the second order in \cite{Bab,MT} (see also \cite{Woess} for the discrete version). Due to the band-gap structure of the spectra of such periodic operators, the question arises whether similar results can be obtained at or near the edges of spectral gaps. The corresponding result at the internal edges of the spectrum was established in \cite{KR}. The main result of this article, Theorem \ref{main}, is the description of such asymptotics near the spectral edge for generic periodic elliptic operators of second-order with real coefficients in dimension $d \geq 2$, if the spectral edge is attained at a symmetry point of the Brillouin zone.

It is well known that outside of the spectrum the Green function decays exponentially at infinity, with the rate of decay controlled by the distance to the spectrum. See, e.g., well known Combes-Thomas estimates \cite{CombThom,BarComb}. However, comparison with the formulas for the case of the Laplacian shows that an additional algebraically decaying factor (depending on the dimension) is lost in this approach. Moreover, the exponential decay in general is expected to be anisotropic, while the operator theory approach can provide only isotropic estimates. The result of this paper provides the exact principal term of asymptotics, thus resolving these issues.

%%%%%%%%%%%%%%%%%%%%%%%%%%%
\section{Assumptions, notation and the main result}
%%%%%%%%%%%%%%%%%%%%%%%%%%%

Consider a linear second order elliptic operator in $\mathbb{R}^d$ with periodic coefficients
\begin{equation}
\label{eqn:L definition}
L(x,D)=\sum_{k,l=1}^{d} D_{k}(a_{kl}(x)D_{l})+V(x)=D^{*}A(x)D+V(x).
\end{equation}

Here $A=(a_{kl})_{1 \leq k,l \leq d}$, $D=(D_{1},\dots,D_{d})$, and $\displaystyle D_k:=-i\partial_{k}=-i\frac{\partial}{\partial x_k}$. All coefficients $a_{kl}, V$ are smooth real-valued functions on $\mathbb{R}^d$, periodic with respect to the integer
lattice $\mathbb{Z}^d$ in $\mathbb{R}^d$, i.e., $a_{kl}(x+n)=a_{kl}(x)$ and $V(x+n)=V(x)$, $\forall x \in \mathbb{R}^{d}, n \in \mathbb{Z}^d$.
The operator $L$ is assumed to be elliptic, i.e., the matrix $A$ is symmetric and
\begin{equation}
\label{E:ellipticity}
\sum_{k,l=1}^{d} a_{kl}(x)\xi_{k}\xi_{l} \geq \theta |\xi|^2,
\end{equation}
for some $\theta >0$ and any $x \in \mathbb{R}^d$, $\xi=(\xi_1,\dots,\xi_{d}) \in \mathbb{R}^d$.
The operator $L$, with the Sobolev space $H^{2}(\mathbb{R}^d)$ as the domain, is an unbounded, self-adjoint operator in $L^2(\mathbb{R}^d)$ (see e.g., \cite{Shubin}).

The spectrum of the above operator $L$ in $L^{2}(\mathbb{R}^d)$ has a \textbf{band-gap structure} \cite{Ea, K, RS4}, i.e., it is the union of a sequence of closed bounded intervals (\textbf{bands} or \textbf{stability zones} of the operator $L$) $[\alpha_j, \beta_j] \subset \mathbb{R}\,(j=1,2,\dots)$:
\begin{equation}
\label{band-gap}
\sigma(L)=\bigcup_{j=1}^{\infty}[\alpha_j, \beta_j],
\end{equation}
such that $\alpha_{j} \leq \alpha_{j+1}$, $\beta_{j} \leq \beta_{j+1}$ and $\lim_{j \rightarrow \infty}\alpha_j=\infty$. The bands can (and do) overlap when $d>1$, but they may leave open intervals in between, called \textbf{spectral gaps}.
Thus, a spectral gap is an interval of the form $(\beta_{j}, \alpha_{j+1})$ for some $j \in \mathbb{N}$ for which $\alpha_{j+1}>\beta_{j}$.
We make a convention that the open interval $(-\infty, \alpha_{1})$, which contains all real numbers below the bottom of the spectrum of $L$, is also a spectral gap. However, we will be mostly interested in finite spectral gaps.

In this text, we study Green's function asymptotics for the operator $L$ in a spectral gap, near to a spectral gap edge. More precisely, consider a finite spectral gap $(\beta_{j},\alpha_{j+1})$ for some $j \in \mathbb{N}$ and a value $\lambda \in (\beta_j, \alpha_{j+1})$ which is close either to the spectral edge $\beta_j$ or to the spectral edge $\alpha_{j+1}$. We would like to study the asymptotic behavior when $|x-y|\to\infty$ of the Green's kernel $G_{\lambda}(x,y)$ of the resolvent operator $R_{\lambda,L}:=(L-\lambda)^{-1}$. The case of the spectral edges (i.e., $\lambda=\alpha_{j+1}$ or $\lambda=\beta_{j}$) was studied for the similar purpose in \cite{KR}. All asymptotics here and also in \cite{KR} are deduced from an assumed ``generic" spectral edge behavior of the \textbf{dispersion relation} of the operator $L$, which we will briefly review below.

Let $W=[0,1]^{d} \subset \mathbb{R}^d$ be the unit cube, which is a \textbf{fundamental domain} of $\mathbb{R}^d$ with respect to the lattice $\mathbb{Z}^d$ (\textbf{Wigner-Seitz cell}). The \textbf{dual} (or \textbf{reciprocal}) \textbf{lattice} is $2\pi \mathbb{Z}^d$ and its fundamental domain is $[-\pi,\pi]^{d}$ (\textbf{Brillouin zone}).

The $d$-dimensional tori with respect to the lattices $\mathbb{Z}^d$ and $2\pi \mathbb{Z}^d$ are denoted by $\mathbb{T}:=\mathbb{R}^d/\mathbb{Z}^d$ and $\mathbb{T}^{*}:=\mathbb{R}^d/2\pi \mathbb{Z}^d$, respectively.

\begin{defi}
\label{H sk}
For any $k \in \mathbb{C}^d$, the subspace $H^{s}_{k}(W) \subset H^{s}(W)$ consists of restrictions to $W$ of functions $f \in H^{s}_{loc}(\mathbb{R}^d)$ that satisfy for any $\gamma \in \mathbb{Z}^d$ the \textbf{Floquet-Bloch condition} (also known as automorphicity condition or cyclic condition)
\begin{equation}
\label{floquet_cond}
f(x+\gamma)=e^{ik\cdot \gamma}f(x) \quad \mbox{for a.e} \quad x\in W.
\end{equation}
Here $H^{s}$ denotes the standard Sobolev space of order $s$. Note that when $s=0$, the above space coincides with $L^{2}(W)$ for any $k$.
\end{defi}

Due to periodicity, the operator $L(x,D)$ preserves condition \mref{floquet_cond} and thus, it defines an operator $L(k)$ in $L^{2}(W)$ with the domain $H^{2}_{k}(W)$. In this model, $L(k)$ is realized as a $k$-independent differential expression $L(x,D)$ acting on functions in $W$ with boundary conditions depending on $k$ (which can be identified with sections of a linear bundle over the torus $\mathbb{T}$). An alternative definition of $L(k)$ is as the operator $L(x,D+k)$ in $L^{2}(\mathbb{T})$ with the domain $H^{2}(\mathbb{T})$. In the latter model, $L(k)$ acts on the $k$-independent domain of periodic functions on $W$ as follows:
\begin{equation}
\label{conjugatingLk}
e^{-ik\cdot x}L(x,D)e^{ik\cdot x}=(D+k)^{*}A(x)(D+k)+V(x).
\end{equation}
We use the latter model of $L(k)$ throughout this paper, unless specified differently.

Note that the condition \mref{floquet_cond} is invariant under translations of $k$ by elements of the dual lattice $2\pi \mathbb{Z}^d$. Moreover, the operator $L(k)$ is unitarily equivalent to $L(k+2\pi \gamma)$, for any $\gamma \in \mathbb{Z}^d$. In particular, when dealing with real values of $k$, it suffices to restrict $k$ to the Brillouin zone $[-\pi, \pi]^d$ (or any its shifted copy). It is well-known (see \cite{Ea, K, RS4}) that the spectrum of $L$ is the union of all the spectra of $L(k)$ when $k$ runs over the Brillouin zone, i.e.
\begin{equation}
\label{fl_spectrum}
\sigma(L)=\bigcup_{k \in [-\pi, \pi]^d}\sigma(L(k)).
\end{equation}

Hence, the spectrum of $L$ is the range of the multivalued function
\begin{equation}
\label{sp_function}
k \mapsto \lambda(k):=\sigma(L(k)), \quad  k\in  [-\pi, \pi]^d.
\end{equation}
%The graph of this multi-valued function in $[-\pi, \pi]^d \times \mathbb{R}$, i.e.,
%\begin{equation}
%\label{disp_relation}
%\{(k,\lambda): \lambda \in \sigma(L(k))\},
%\end{equation}
%is called the \textbf{dispersion relation/curve} of the operator $L$.
%or (real) \textbf{Bloch variety} $B(L)$ for the operator $L$:

By \mref{conjugatingLk}, $L(k)$ is self-adjoint in $L^{2}(\mathbb{T})$ and has domain $H^{2}(\mathbb{T})$ for each $k \in \mathbb{R}^d$.
By the ellipticity of $L$, each $L(k)$ is bounded from below and has compact resolvent. This forces each of the operators $L(k)$, $k \in \mathbb{R}^d$, to have discrete spectrum in $\mathbb{R}$. Therefore, we can label its eigenvalues in non-decreasing order:
\begin{equation}
\label{eigenv}
\lambda_{1}(k) \leq \lambda_{2}(k) \leq \dots \quad .
\end{equation}
Hence, we can single out continuous and piecewise-analytic \textbf{band functions} $\lambda_{j}(k)$ for each $j \in \mathbb{N}$ \cite{Wilcox}. The range of the band function $\lambda_j$ constitutes exactly the band $[\alpha_j, \beta_j]$ of the spectrum of $L$  shown in \mref{band-gap}.

\begin{defi}
A \textbf{Bloch solution} of the equation $L(x,D)u=0$ is a solution of the form
\begin{equation*}
u(x)=e^{ik\cdot x}\phi (x),
\end{equation*}
where the function $\phi$ is 1-periodic in each variable $x_{j}$ for $j=1,\dots,d$. The vector $k$ is the \textbf{quasimomentum} and $z=e^{ik}=(e^{ik_{1}},\dots,e^{ik_{d}})$ is the \textbf{Floquet exponent} (or \textbf{Floquet multiplier}) of the solution.
In our formulation, allowing quasimomenta $k$ to be complex is essential.
\end{defi}

\begin{defi}
The (complex) \textbf{Bloch variety} $B_{L}$ of the operator $L$ consists of all pairs $(k,\lambda) \in \mathbb{C}^{d+1}$ such that the equation $Lu=\lambda u$ in $\mathbb{R}^d$ has a non-zero Bloch solution $u$ with a quasimomentum $k$. Similarly, the \textbf{real Bloch variety} $B_{L,\mathbb{R}}$ is $B_{L} \cap \mathbb{R}^{d+1}$.
\end{defi}
The Bloch variety $B_{L}$ can be treated as the \textbf{dispersion relation/curve}, i.e., the graph of the multivalued function $\lambda(k)$:
\begin{equation*}
\label{disp_relation}
B_L=\{(k,\lambda): \lambda \in \sigma(L(k))\}.
\end{equation*}
Note that $L(k)$ is non-self-adjoint if $k \notin \mathbb{R}^d$. However, $L(k)-L(0)$ is an operator of lower order for each $k \in \mathbb{C}^d$. Therefore, the spectra of all operators $L(k)$ on the torus $\mathbb{T}$ are discrete (see pp.188-190 in \cite{Agmon}).
%since a perturbation of a self-adjoint elliptic operator in $L^2(\mathbb{T})$  by some lower order term has discrete spectrum
\begin{remark}
In fact, the main techniques of Floquet theory (e.g., \mref{fl_spectrum}) apply to non-self-adjoint operators. It is required only that the operators $L(k)=L(x,D+k): H^{2}(\mathbb{T})\rightarrow L^2(\mathbb{T})$ are Fredholm for $k \in \mathbb{C}^d$. The latter condition is always satisfied due to ellipticity and embedding theorems (see Theorem 2.1 in \cite{K}). Unlike the self-adjoint case though, we do not have the band-gap structure as in \mref{band-gap}.
\end{remark}
\begin{defi}
The (complex) \textbf{Fermi surface} $F_{L,\lambda}$ of the operator $L$ at the energy level $\lambda \in \mathbb{C}$ consists of all quasimomenta $k \in \mathbb{C}^{d}$ such that the equation $Lu=\lambda u$ in $\mathbb{R}^d$ has a non-zero Bloch solution $u$ with a quasimomentum $k$.  For $\lambda=0$, we simply write $F_{L}$ instead of $F_{L,0}$.
The \textbf{real Fermi surface} $F_{L,\mathbb{R}}$ is $F_{L} \cap \mathbb{R}^d$.
\end{defi}
Equivalently, $k \in F_{L,\lambda}$ means the existence of a nonzero periodic solution $u$ of the equation $L(k)u=\lambda u$.
In other words, Fermi surfaces are level sets of the dispersion relation.

The following result can be found in Theorem 3.1.7 in \cite{K}:
\begin{lemma}
\label{L:Bloch_variety}
There exist entire ($2\pi \mathbb{Z}^d$-periodic in $k$) functions of finite orders on $\mathbb{C}^{d}$ and on $\mathbb{C}^{d+1}$ such that the Fermi and Bloch varieties are the sets of all zeros of these functions respectively.
\end{lemma}

From this lemma and the proof of Lemma 4.5.1 in \cite{K} (see also \cite{Wilcox}), the band functions $\lambda(k)$ are piecewise analytic on $\mathbb{C}^d$.

From now on, we fix $L$ as a self-adjoint elliptic operator of the form \mref{eqn:L definition}, whose band-gap structure is as \mref{band-gap}.
By adding a constant to the operator $L$ if necessary, we can assume that the spectral edge of interest is $0$. It is also enough to suppose that the adjacent spectral band is of the form $[0,a]$ for some $a>0$ since the case when the spectral edge $0$ is the maximum of its adjacent spectral band is treated similarly.

Suppose there is no spectrum for small negative values of $\lambda$ and hence there is a spectral gap below 0. Thus, there exists at least one band function $\lambda_{j}(k)$ for some $j \in \mathbb{N}$ such that $0$ is the minimal value of this function on the Brillouin zone.

To establish our main result, we need to impose the following analytic assumption on the dispersion curve $\lambda_j$ as in \cite{KR}:\\

\textbf{Assumption A}\\

\emph{There exists $k_0 \in [-\pi, \pi]^d$ and a band function $\lambda_{j}(k)$ such that:}\\

\textbf{A1} $\lambda_{j}(k_0)=0$.\\

\textbf{A2} $\min_{k \in \mathbb{R}^d, i \neq j}|\lambda_{i}(k)|>0$.\\

\textbf{A3} \emph{$k_0$ is the only\footnote{Finitely many such points can be also easily handled.} (modulo $2\pi \mathbb{Z}^d$) minimum of $\lambda_{j}$}.\\

\textbf{A4} \emph{$\lambda_{j}(k)$ is a Morse function near $k_0$, i.e., its Hessian matrix $\displaystyle H:=\Hess{(\lambda_j)}(k_0)$ at $k_0$ is  positive definite. In particular, the Taylor expansion of $\lambda_j$ at $k_0$ is:
\begin{equation*}
\lambda_{j}(k)=\frac{1}{2}(k-k_{0})^{T}H(k-k_{0})+O(|k-k_0|^3).
\end{equation*}}

It is known \cite{KloppRalston} that the conditions \textbf{A1} and \textbf{A2} `generically' hold (i.e., they can be achieved by small perturbation of coefficients of the operator) for Schr\"odinger operators. Although this has not been proven, the conditions \textbf{A3} and \textbf{A4} are widely believed (both in the mathematics and physics literature) to hold `generically'.
In other words, it is conjectured that for a `generic' selfadjoint second-order elliptic operator with periodic coefficients on $\mathbb{R}^d$ each of the spectral gap's endpoints is a unique (modulo the dual lattice $2\pi\mathbb{Z}^d$), nondegenerate extremum of a single band function $\lambda_{j}(k)$ (see e.g., \textbf{Conjecture 5.1} in \cite{KP2}).
It is known that for a non-magnetic periodic Schr\"odinger operator, the bottom of the spectrum always corresponds to a non-degenerate minimum of $\lambda_1$ \cite{KS}. A similar statement is correct for a wider class of  `factorable' operators \cite{BS_2001, BS_2003}. The following condition on $k_0$ will also be needed:\\

\textbf{A5} \emph{The quasimomentum $k_0$ is a \textbf{high symmetry point of the Brillouin zone}, i.e., all components of $k_0$ must be either equal to $0$ or to $\pi$}.\\

We denote by $X$ the set of such high symmetry points in the Brillouin zone.

It is known \cite{HKSW} that the condition \textbf{A5} is not always satisfied and spectral edges could occur deeply inside the Brillouin zone. However, as it is discussed in \cite{HKSW}, in many practical cases (e.g., in the media close to homogeneous) this condition holds.

We would like to introduce a suitable fundamental domain with respect to the dual lattice $2\pi \mathbb{Z}^d$ to work with.
\begin{defi}
\label{fd}
Consider the quasimomentum $k_0$ in our assumptions. By \textbf{A5}, $k_0=(\delta_1 \pi, \delta_2 \pi,\dots, \delta_d \pi)$, where $\delta_{j} \in \{0,1\}$ for $j \in \{1,\dots,d\}$. We denote by $\mathcal{O}$ the fundamental domain so that $k_0$ is its center of symmetry, i.e.,
$$\mathcal{O}=\prod_{j=1}^d [(\delta_{j}-1)\pi, (\delta_j+1)\pi].$$

When $k_0=0$, $\mathcal{O}$ is just the Brillouin zone.
\end{defi}

We now introduce notation that will be used throughout the paper.
\begin{notation}
\label{notation}
$ $

\begin{enumerate}[(a)]
\item Let $z_1 \in \mathbb{C}$, $z_2 \in \mathbb{C}^{d-1}$, $z_3 \in \mathbb{C}^d$ and $r_i$ be positive numbers for $i=1,2,3$. Then we denote by $B(z_1,r_1)$, $D'(z_2,r_2)$ and $D(z_3,r_3)$ the open balls (or discs) centered at $z_1$, $z_2$ and $z_3$ whose radii are $r_1$, $r_2$ and $r_3$ in $\mathbb{C}$, $\mathbb{C}^{d-1}$ and $\mathbb{C}^d$ respectively.

\item The real parts of a complex vector $z$, or of a complex matrix $A$ are denoted by $\Re(z)$ and $\Re(A)$ respectively.

\item The standard notation $O(|x-y|^{-n})$ for a function $f$ defined on $\mathbb{R}^{2d}$ means there exist constants $C>0$ and $R>0$ such that $|f(x,y)|\leq C|x-y|^{-n}$ whenever $|x-y|>R$. Also, $f(x,y)=o(|x-y|^{-n})$ means that
$$\displaystyle \lim_{|x-y| \rightarrow \infty}|f(x,y)|/|x-y|^{n}=0.$$

\item We often use the notation $A \lesssim B$ to mean that the quantity $A$ is less or equal than the quantity $B$ up to some multiplicative constant factor, which does not affect the arguments.
\end{enumerate}
\end{notation}

As we discussed, for each $z \in \mathbb{C}^d$, the operator $L(z)$ has discrete spectrum and is therefore a closed operator with non-empty resolvent set.
These operators have the same domain $H^2{(\mathbb{T})}$ and for each $\phi \in H^2(\mathbb{T})$, $L(z)\phi$ is a $L^2(\mathbb{T})$-valued analytic function of $z$, due to \mref{conjugatingLk}. Consequently, $\{L(z)\}_{z \in \mathbb{C}^d}$ is an analytic family of type $\mathcal{A}$ in the sense of Kato \cite{Ka}\footnote{It is also an analytic family in the Banach space of bounded linear operators acting from $H^2(\mathbb{T})$ to $L^2(\mathbb{T})$.}. Due to $\textbf{A1-A2}$, $\lambda_j(k_0)$ is a simple eigenvalue of $L(k_0)$. By using analytic perturbation theory for the family $\{L(z)\}_{z \in \mathbb{C}^d}$ (see e.g., Theorem XII.8 in \cite{RS4}), there is an open neighborhood $V$ of $k_0$  in $\mathbb{C}^d$ and some $\epsilon_0>0$ such that

\textbf{(P1)} $\lambda_j$ is analytic in a neighborhood of the closure of $V$.

\textbf{(P2)} $\lambda_{j}(z)$ has algebraic multiplicity one, i.e., it is a simple eigenvalue of $L(z)$ for any $z \in \overline{V}$.

\textbf{(P3)}  The only eigenvalue of $L(z)$ contained in the closed disc $\overline{B}(0,\epsilon_0)$ is $\lambda_j(z)$. Moreover, we may also assume that $|\lambda_j(z)|<\epsilon_0$ for each $z \in V$.

\textbf{(P4)}  For each $z \in \overline{V}$, let $\phi(z,x)$ be a nonzero $\mathbb{Z}^d$-periodic function of $x$ such that it is the unique (up to a constant factor) eigenfunction of $L(z)$ with the eigenvalue $\lambda_{j}(z)$, i.e., $L(z)\phi(z,\cdot)=\lambda_{j}(z)\phi(z,\cdot)$. We will also use sometimes the notation $\phi_{z}$ for the eigenfunction $\phi(z,\cdot)$.

By elliptic regularity, $\phi(z,x)$ is smooth in $x$. On a neighborhood of $\overline{V}$, $\phi(z,\cdot)$ is a $H^2(\mathbb{T})$-valued holomorphic function.

\textbf{(P5)} By condition \textbf{A4} and the continuity of $\Hess{(\lambda_j)}$,\footnote{The Hessian matrix of $\lambda_j$.} we can assume that for all $z\in V$,
$$2\Re(\Hess{(\lambda_{j})}(z))>\min \sigma(\Hess{(\lambda_{j})}(k_0))I_{d \times d}.$$

\textbf{(P6)} %For $\overline{z}, z \in V$ if $z \in V$, 
$V$ is invariant under complex conjugation. Furthermore,
the smooth function
\begin{equation}
\label{eigen_inner_prod}
F(z):=(\phi(z,\cdot),\phi(\overline{z},\cdot))_{L^{2}(\mathbb{T})}
\end{equation}
is non-zero on $V$, due to analyticity of the mapping $z \mapsto \phi(z,\cdot)$ and the inequality $F(k_0)=\|\phi(k_{0})\|^{2}_{L^{2}(\mathbb{T})}>0$.

%This below remark is not needed since the proof of analytic perturbation theorem has already contained that requirement in (P3).
%\begin{remark}
%\label{R:P3}
%Analytic perturbation theory guarantees that for some $\epsilon_0>0$ small enough, one can find a neighborhood $\tilde{V}$ of $k_0$ such that $\overline{B}(0, 2\epsilon_0) \cap \sigma(L(z))=\{\lambda_j(z)\}$ for each $z \in \tilde{V}$. Then by defining $V:=\tilde{V} \cap \lambda_{j}^{-1}(B(0,\epsilon_0))$, the circle of radius $\epsilon_0$ never meets any of spectra of $L(z)$, $z \in V$. This shows the property (\textbf{P3}).
%\end{remark}

The following lemma will be useful when dealing with operators having real and smooth coefficients:
\begin{lemma}
\label{evenness_band_functions}
(i) For $k$ in $\mathbb{R}^d$ and $i \in \mathbb{N}$,
\begin{equation}
\label{evenness}
\lambda_{i}(k)=\lambda_{i}(-k).
\end{equation}
In other words, each band $\lambda_{i}$ of $L$ is an even function on $\mathbb{R}^d$.

(ii) If $k_0 \in X$, we have $\lambda_{i}(k+k_{0})=\lambda_{i}(-k+k_0)$ for all $k$ in $\mathbb{R}^d$ and $i \in \mathbb{N}$.
\end{lemma}
\begin{proof}
Let $\phi_{k}$ be an eigenfunction of $L(k)$ corresponding to $\lambda_{j}(k)$. This means that $\phi_{k}$ is a periodic solution to the equation
\begin{equation}
\label{eigen_eqn}
L(x,\partial+ik)\phi_{k}(x)=\lambda_{j}(k)\phi_{k}(x).
\end{equation}
Taking the complex conjugate of \mref{eigen_eqn}, we get
\begin{equation*}
L(x,\partial-ik)\overline{\phi_{k}}(x)=\lambda_{j}(k)\overline{\phi_{k}}(x).
\end{equation*}
Therefore, $\overline{\phi_k}$ is an eigenfunction of $L(-k)$ with eigenvalue $\lambda_{j}(k)$. This implies the identity \mref{evenness}.

(ii) By (i), $\lambda_{i}(k+k_{0})=\lambda_{i}(-k-k_{0})=\lambda_{i}(-k+k_{0})$ since $2k_{0} \in 2\pi\mathbb{Z}^d$.
\end{proof}

\begin{cor}
\label{evenness_symmetry}
If $\beta \in \mathbb{R}^d$ such that $k_0+i\beta \in \overline{V}$ then $\lambda_j(k_0+i\beta) \in \mathbb{R}$.
\end{cor}
\begin{proof}
Indeed, the statement (ii) of Lemma \ref{evenness_band_functions} implies that the Taylor series of $\lambda(k)$ at $k_0$ has only even degree terms and real coefficients.
\end{proof}

Corollary \ref{evenness_symmetry} allows us to define near $\beta=0$ the real analytic function $E(\beta):=\lambda_{j}(k_0+i\beta)$ near $0$. Since its Hessian at $0$ is negative-definite (by \textbf{A4}), there exists a connected and bounded neighborhood $V_0$ of $0$ in $\mathbb{R}^d$ such that $k_0+iV_0 \subseteq V$ and $\Hess{(E)}(\beta)$ is negative-definite whenever $\beta$ belongs to $V_0$.
Thus, $E$ is strictly concave on $V_0$ and $\sup_{\beta \in V_0}E(\beta)=E(0)=0$, $\nabla E(\beta)=0$ iff $\beta=0$. Note that at the bottom of the spectrum (i.e., $j=1$),
we could take $V_0$ as the whole Euclidean space $\mathbb{R}^d$.

By the Morse lemma and the fact that $0$ is a nondegenerate
critical point of $E$, there is a smooth change of coordinates $\phi: U_0 \rightarrow \mathbb{R}^d$ so that $0 \in U_0 \subset \subset V_0$, $U_0$ is connected, $\phi(0)=0$ and
$E(\phi^{-1} (a))=-|a|^2, \forall a \in \phi(U_0)$. Set $K_{\lambda}:=\{\beta \in U_0: E(\beta)\geq \lambda \}$ and
$\Gamma_{\lambda}:=\{\beta \in U_0: E(\beta)=\lambda \}$ for each $\lambda \in \mathbb{R}$. Now, we consider $\lambda$ to be in the set $\{-|a|^2: a \in \phi(U_0), a \neq 0\}$. Then
$K_{\lambda}$ is a strictly convex $d$-dimensional compact body in $\mathbb{R}^d$, and $\Gamma_{\lambda}=\partial K_{\lambda}$ is a compact hypersurface in $\mathbb{R}^d$.
The compactness of $K_{\lambda}$ follows from the equation $-|\phi(\beta)|^2=E(\beta) \geq \lambda$ which yields that $|\beta|=|\phi^{-1}(\phi(\beta))| \leq \max \{|\phi^{-1}(a)|: a \in \phi(U_0), |a|^2 \leq -\lambda\}$.
Additionally, $\lim_{\lambda \rightarrow 0^-}\max_{\beta \in K_{\lambda}}|\beta|=0$.

Let $\mathcal{K}_{\lambda}$ be the Gauss-Kronecker curvature of $\Gamma_{\lambda}$. Since the Hessian of $E$ is negative-definite on $\Gamma_{\lambda}$, $\mathcal{K}_{\lambda}$ is nowhere-zero.
For the value of $\lambda$ described in the previous paragraph and each $s \in \mathbb{S}^{d-1}$, there is a unique vector $\beta_{s} \in \Gamma_{\lambda}$ such that the value of
the Gauss map of the hypersurface $\Gamma_{\lambda}$ at this point coincides with $s$, i.e.
\begin{equation}
\label{E:gradient_E_s}
\nabla E(\beta_{s})=-|\nabla E(\beta_s)|s.
\end{equation}
This is due to the fact that the Gauss map of a compact, connected oriented hypersurface in $\mathbb{R}^d$, whose Gauss-Kronecker curvature is nowhere zero, is a diffeomorphism onto the sphere $\mathbb{S}^{d-1}$ (see e.g., Theorem 5, p.104 in \cite{Thorpe} or Corollary 3.1 in \cite{Ghomi}). Thus, $\beta_s$ depends smoothly on $s$.
We also see that $$\lim_{|\lambda| \rightarrow 0}\max_{s \in \mathbb{S}^{d-1}}|\beta_s|=0.$$

Note that $\beta_s$ could be defined equivalently by using the support functional $h$ of the strictly convex set $K_{\lambda}$. Recall that for each direction $s \in \mathbb{S}^{d-1}$,
$$\displaystyle h(s)=\max_{\xi \in K_{\lambda}}\langle s, \xi\rangle.$$
Then $\beta_s$ is the unique point in $\Gamma_{\lambda}$ such that $\langle s, \beta_s \rangle=h(s)$.

By letting $|\lambda|$ close enough to $0$, we can make sure that $(-\lambda)^{1/2} \in \phi(U_0)$. Then
\begin{equation}
\label{E:beta_s_in_V}
\{k_0+it\beta_{s}, (t,s) \in [0,1] \times \mathbb{S}^{d-1}\} \subset V.
\end{equation}

We can now state the main result of the paper.
\begin{thm}
\label{main}
Suppose conditions \textbf{A1-A5} are satisfied.
For $\lambda<0$ sufficiently close to $0$ (depending on the dispersion branch $\lambda_j$ and the operator $L$), the Green's function $G_{\lambda}$ of $L$ at $\lambda$ admits the following asymptotics as $|x-y| \rightarrow \infty$:
\begin{equation}
\label{main_asymp}
\begin{split}
G_{\lambda}(x,y)&=\frac{e^{(x-y)(ik_{0}-\beta_{s})}}{(2\pi|x-y|)^{(d-1)/2}}\frac{|\nabla E(\beta_s)|^{(d-3)/2}}{\det{(-\mathcal{P}_s \Hess{(E)}(\beta_{s})\mathcal{P}_s)}^{1/2}}\frac{\phi_{k_{0}+i\beta_{s}}(x)\overline{\phi_{k_{0}-i\beta_{s}}(y)}}{(\phi_{k_{0}+i\beta_{s}},\phi_{k_{0}-i\beta_{s}})_{L^{2}(\mathbb{T})}}
\\&+e^{(y-x)\cdot \beta_{s}}r(x,y).
\end{split}
\end{equation}
Here $\displaystyle s=(x-y)/|x-y|$, $\mathcal{P}_s$ is the projection from $\mathbb{R}^{d}$ onto the tangent space of the unit sphere $\mathbb{S}^{d-1}$ at the point $s$, and for any $\varepsilon>0$, there exists a constant $C_{\varepsilon}>0$ (independent of $s$) such that the remainder term $r$ satisfies $|r(x,y)| \leq C_{\varepsilon}|x-y|^{-d/2+\varepsilon}$ when $|x-y|$ is large enough.
\end{thm}
This result achieves our stated goal of showing the precise (anisotropic) rates of the exponential decay of the Green's function and capturing the additional algebraic decay factor.

%%%%%%%%%%%%%%
\section{Proof of the main theorem \ref{main} and some remarks}
%%%%%%%%%%%

Theorem \ref{main} is a direct consequence of its local (with respect to the direction of $(x-y)$) version:
\begin{thm}
\label{main_local}
Under the hypotheses of Theorem \ref{main} and when $\lambda \approx 0$, for each $\omega \in \mathbb{S}^{d-1}$, there exists a neighborhood $\mathcal{V}_{\omega}$ in $\mathbb{S}^{d-1}$ containing $\omega$ and a smooth function
$e(s)=(e_{s,2},\dots, e_{s,d}):\mathcal{V}_{\omega} \to (T_s \mathbb{S}^{d-1})^{d-1}$ such that the asymptotics
\begin{equation}
\label{main_asymp_local}
%\begin{split}
%G_{\lambda}(x,y)&=\frac{e^{(x-y)(ik_{0}-\beta_{s})}}{(2\pi|x-y|)^{(d-1)/2}}\frac{|\nabla E(\beta_s)|^{(d-3)/2}}{\det{(-\mathcal{P}_s \Hess{(E)}(\beta_{s})\mathcal{P}_s)}^{1/2}}\frac{\phi_{k_{0}+i\beta_{s}}(x)\overline{\phi_{k_{0}-i\beta_{s}}(y)}}{(\phi_{k_{0}+i\beta_{s}},\phi_{k_{0}-i\beta_{s}})_{L^{2}(\mathbb{T})}}
%\\&+e^{(y-x)\cdot \beta_{s}}r(x,y).
%\end{split}
\begin{split}
G_{\lambda}(x,y)&=\frac{e^{(x-y)(ik_{0}-\beta_{s})}}{(2\pi|x-y|)^{(d-1)/2}}\left(\frac{|\nabla E(\beta_s)|^{(d-3)}}{\det{(-e_{s,p}\cdot \Hess{(E)}(\beta_{s})e_{s,q})_{2 \leq p,q \leq d}}}\right)^{1/2}
\\&\times \frac{\phi_{k_{0}+i\beta_{s}}(x)\overline{\phi_{k_{0}-i\beta_{s}}(y)}}{(\phi_{k_{0}+i\beta_{s}},\phi_{k_{0}-i\beta_{s}})_{L^{2}(\mathbb{T})}}+e^{(y-x)\cdot \beta_{s}}r(x,y),
\end{split}
\end{equation}
hold for all $(x,y)$ such that $s=(x-y)/|x-y| \in \mathcal{V}_{\omega}$. Furthermore, $|r(x,y)| \leq C(\varepsilon, \omega)|x-y|^{-d/2+\varepsilon}$ for some positive constant $C(\varepsilon, \omega)$ depending on $\varepsilon$ and $\omega$.
\end{thm}

\begin{prmainthm}
Observe that for any basis $\{e_{s,l}\}_{2 \leq l \leq d}$ of the tangent space $T_s \mathbb{S}^{d-1}$, $$\displaystyle \det{(-\mathcal{P}_s \Hess{(E)}(\beta_{s})\mathcal{P}_s)}=\det{(-e_{s,p}\cdot \Hess{(E)}(\beta_{s})e_{s,q})_{2 \leq p,q \leq d}}.$$
Now, using of a finite cover of the unit sphere by neighborhoods $\mathcal{V}_{\omega_j}$ in Theorem \ref{main_local}, one obtains Theorem \ref{main}.
\end{prmainthm}

\begin{remark}
$ $

\begin{itemize}
\item The asymptotics \mref{main_asymp}  (or \mref{main_asymp_local}) resemble the formula (1.1) in Theorem 1.1 \cite{MT} when $\lambda$ is below the bottom of the spectrum of the operator. Moreover, as in Theorem 1.1 in \cite{MT06},  by using the Gauss-Kronecker curvature $\mathcal{K}_{\lambda}$, the main result \mref{main_asymp} could be restated as follows:
\begin{equation*}
\begin{split}
G_{\lambda}(x,y)&=\frac{e^{(x-y)(ik_{0}-\beta_{s})}}{(2\pi|x-y|)^{(d-1)/2}}\frac{1}{|\nabla E(\beta_s)| \mathcal{K}_{\lambda}(\beta_s)^{1/2}}
\frac{\phi_{k_{0}+i\beta_{s}}(x)\overline{\phi_{k_{0}-i\beta_{s}}(y)}}{(\phi_{k_{0}+i\beta_{s}},\phi_{k_{0}-i\beta_{s}})_{L^{2}(\mathbb{T})}}\\&+e^{(y-x)\cdot \beta_{s}}O(|x-y|^{-d/2+\varepsilon}).
\end{split}
\end{equation*}

\item Although \mref{main_asymp} is an anisotropic formula, it is not hard to obtain from \mref{main_asymp} an isotropic upper estimate for the Green's function $G_{\lambda}$ based on the \textbf{distance from $\lambda$ to the spectrum of the operator $L$},\footnote{Recall that the spectral edge is assumed to be zero.} e.g., there are some positive constants $C_1, C_2$ (depending only on $L$ and $\lambda_j$) and $C_3$ (which may depend on $\lambda$) such that whenever $|x-y|>C_3$, the following inequality holds:
\begin{equation*}
|G_{\lambda}(x,y)| \leq C_{1}|\lambda|^{(d-3)/4}\frac{e^{-C_{2}|\lambda|^{1/2}|x-y|}}{|x-y|^{(d-1)/2}}\cdot
\end{equation*}

\item If the band edge occurs at finitely many points, rather than a single $k_0$, one just needs to combine the asymptotics coming from all these isolated minima.
\end{itemize}
\end{remark}

Now we outline the proof of Theorem \ref{main_local}. In Section 5, we introduce the tools of Floquet-Bloch theory to reduce the problem to that of finding the asymptotics of a scalar integral, similarly to \cite{KR}.
The purpose of Section 4 is to prepare for Section 5, i.e., to shift the integral from the fundamental domain $\mathcal{O}$ along some purely imaginary directions in $\mathbb{C}^d$.
Section 6 is devoted to estimating the leading term integral by adapting the method similar to the one used in the discrete case \cite{Woess}.
The decay of the remainder $r(x,y)$ comes from some elementary estimates of the difference between the scalar integral and the leading term of its asymptotics. In order to not overload the main text with technicalities, the proofs of some auxiliary statements are postponed till Sections 7-9.

%%%%%%%%%%%%%%%%%%%%%%%
\section{On local geometry of the resolvent set}
%%%%%%%%%%%%%%%%
The following proposition shows that for any $s \in \mathbb{S}^{d-1}$, $k_0+i\beta_s$ is the only complex quasimomentum having the form of $k+it\beta_{s}$ where $k \in \mathcal{O}, t \in [0,1]$
such that $\lambda$ is in the spectrum of the corresponding fiber operator $L(k+it\beta_s)$. In other words,  by moving from $k \in \mathcal{O}$ in the direction $i\beta_s$, the first time we \textbf{hit the Fermi surface $F_{L, \lambda}$}
(i.e., the spectrum of $L(k)$ meets $\lambda$) is at the value of the quasimomentum $k=k_0+i\beta_s$.
This step is crucial for setting up the scalar integral in the next section, which is solely responsible for the main term asymptotics of our Green's function.

\begin{prop}
\label{singularity}
If $|\lambda|$ is small enough (depending on the dispersion branch $\lambda_j$ and $L$), then $\lambda \in \rho(L(k+it\beta_s))$ if and only if $(k,t) \neq (k_0,1)$.
\end{prop}
The proof of this proposition is presented in Subsection \ref{subsec:Proposition proofs}.
%%%%%%%%%%%%%%%%%%%%%%%%%%%%%%%

%%%%%%%%%%%%%%%%%%%%%%%
\section{A Floquet reduction of the problem}\label{sec:floquet}
%%%%%%%%%%%%%%%%%%%%

We recall here some basic properties of the Floquet transform and then apply this transform to reduce our problem to finding asymptotics of a scalar integral expression, which is close to the one arising when dealing with the Green's function of the Laplacian at a small negative level $\lambda$. As in \cite{KR}, the idea is to show that only the branch of the dispersion relation $\lambda_j$ appearing in the Assumption\textbf{ A}  dominates the asymptotics.

%%%%%%%%%%%%%%%%%%%%%%%%%%%%%%%%
\subsection{The Floquet transform}
%%%%%%%%%%%%%%%%%%%%%%

Let us consider a sufficiently fast decaying function $f(x)$ (to begin with, compactly supported functions) on $\mathbb{R}^d$.
%Let $\gamma \in \mathbb{Z}^d$ and denote the $\gamma$-shifted version of $f$ by $\tau_{\gamma}(f)$: $\tau_{\gamma}(f)(x)=f(x+\gamma)$.
We need the following transform that plays the role of the Fourier transform for the periodic case \cite{K,RS4}. In fact, it is a version of the Fourier transform on the group $\mathbb{Z}^d$ of periods. We use the following version, which is slightly different from the one used in \cite{KR}.
\begin{defi}
\label{Floquet transforms}
The \textbf{Floquet transform} $\mathcal{F}$
\begin{equation*}
%\label{E:transform}
f(x) \rightarrow \widehat{f}(k,x)
\end{equation*}
maps a function $f$ on $\mathbb{R}^d$ into a function $\widehat{f}$ defined on $\mathbb{R}^d \times \mathbb{R}^d$ in the following way:

\begin{equation*}
%\label{eqn:hat f(k)}
\widehat{f}(k,x):=\sum_{\gamma \in \mathbb{Z}^d}f(x+\gamma)e^{-ik\cdot(x+\gamma)}.
\end{equation*}
\end{defi}
From the above definition, one can see that $\widehat{f}$ is $\mathbb{Z}^d$-periodic in the $x$-variable and satisfies a cyclic condition with respect to $k$:
\[   \left\{
\begin{array}
{ll}
\label{E:periodic}
      \widehat{f}(k,x+\gamma)=\widehat{f}(k,x), \quad \forall \gamma \in \mathbb{Z}^d \\
      \widehat{f}(k+2\pi \gamma,x)=e^{-2\pi i \gamma \cdot x}\widehat{f}(k,x), \quad \forall \gamma \in \mathbb{Z}^d \\
\end{array}
.\right. \]

Thus, it suffices to consider the Floquet transform $\widehat{f}$ as a function defined on $\mathcal{O} \times \mathbb{T}$. Usually, we will regard $\widehat{f}$ as a function $\widehat{f}(k,\cdot)$ in $k$-variable in $\mathcal{O}$ with values in the function space $L^{2}(\mathbb{T})$.

For our purpose, we need to list some well-known results of the Floquet transform (see e.g., \cite{K}):
\begin{lemma}
\label{L:floquet}

1. The transform $\mathcal{F}$ is an isometry of $L^{2}(\mathbb{R}^d)$ onto
\begin{equation*}
%\label{E:direct_integral1}
\int_{\mathcal{O}}^{\oplus}L^{2}(\mathbb{T})=L^{2}(\mathcal{O},L^{2}(\mathbb{T}))
\end{equation*}
and of $H^{2}(\mathbb{R}^d)$ into
\begin{equation*}
%\label{E:direct_integral2}
\int_{\mathcal{O}}^{\oplus}H^{2}(\mathbb{T})=L^{2}(\mathcal{O},H^{2}(\mathbb{T})).
\end{equation*}

2. The inversion $\mathcal{F}^{-1}$ is given by the formula
\begin{equation}
\label{E:inversion1}
f(x)=(2\pi)^{-d}\int_{\mathcal{O}}e^{ik \cdot x}\widehat{f}(k,x)dk, \quad x \in \mathbb{R}^d.
\end{equation}

By using cyclic conditions of $\widehat{f}$, we obtain an alternative inversion formula
\begin{equation}
\label{E:inversion2}
f(x)=(2\pi)^{-d}\int_{\mathcal{O}}e^{ik \cdot x}\widehat{f}(k,x-\gamma)dk, \quad x \in W+\gamma.
\end{equation}

3. The action of any $\Z^d$-periodic elliptic operator $L$ (not necessarily self-adjoint) in $L^{2}(\mathbb{R}^d)$ under the Floquet transform $\mathcal{F}$ is given by
\begin{equation*}
%\label{E:conjugate_floquet}
\mathcal{F}L(x,D)\mathcal{F}^{-1}=\int_{\mathcal{O}}^{\oplus}L(x,D+k)dk=\int_{\mathcal{O}}^{\oplus}L(k)dk,
\end{equation*}
where $L(k)$ is defined in \mref{conjugatingLk}.

Equivalently,
\begin{equation*}
%\label{E:conjugate_floquet2}
\widehat{Lf}(k)=L(k)\widehat{f}(k), \quad \forall f \in H^{2}(\mathbb{R}^d).
\end{equation*}

4. (\textbf{A Paley-Wiener theorem for $\mathcal{F}$}.) Let $\phi(k,x)$ be a function defined on $\mathbb{R}^d \times \mathbb{R}^d$ such that for each $k$, it belongs to the Sobolev space $H^{s}(\mathbb{T})$ for $s\in \mathbb{R}^{+}$ and \emph{satisfies the cyclic condition} in $k$-variable. Then
\begin{enumerate}
\item Suppose the mapping $k \rightarrow \phi(k,\cdot)$ is a $C^\infty$-map from $\mathbb{R}^d$ into the Hilbert space $H^{s}(\mathbb{T})$. Then $\phi(k,x)$ is the Floquet transform of a function $f \in H^{s}(\mathbb{R}^d)$ such that for any compact set $K$ in $\mathbb{R}^d$ and any $N>0$, the norm $\|f\|_{H^{s}(K+\gamma)} \leq C_{N}|\gamma|^{-N}$.
In particular, by Sobolev's embedding theorem, if $s>d/2$, then the pointwise estimation holds:
\begin{equation*}
%\label{E:pointwise_floquet}
|f(x)|\leq C_{N}(1+|x|)^{-N}, \quad \forall N>0.
\end{equation*}
\item Suppose the mapping $k \rightarrow \phi(k,\cdot)$ is an analytic map from $\mathbb{R}^d$ into the Hilbert space $H^{s}(\mathbb{T})$. Then $\phi(k,x)$ is the Floquet transform of a function $f \in H^{s}(\mathbb{R}^d)$ such that for any compact set $K$ in $\mathbb{R}^d$, one has $\|f\|_{H^{s}(K+\gamma)} \leq Ce^{-C|\gamma|}$.
In particular, by Sobolev's embedding theorem, if $s>d/2$, then the pointwise estimation holds:
\begin{equation*}
%\label{E:pointwise_floquet}
|f(x)|\leq C e^{-C|x|}.
\end{equation*}
\end{enumerate}
\end{lemma}
%%%%%%%%%%%%%%%%%%%%%%%%%%%
\subsection{The Floquet reduction}
$ $
\vspace{0.3cm}

The Green's function $G_{\lambda}$ of $L$ at $\lambda$ is the Schwartz kernel of the resolvent operator $R_{\lambda}=(L-\lambda)^{-1}$. Fix a $\lambda<0$ such that the statement of Proposition \ref{singularity} holds. For any $s \in \mathbb{S}^{d-1}$ and $t \in [0,1]$, we consider the following operator with real coefficients on $\mathbb{R}^d$:
\begin{equation}
\label{E:L_t_s}
L_{t,s}:=e^{t\beta_s\cdot x}Le^{-t\beta_s\cdot x}.
\end{equation}

For simplicity, we write $L_s:=L_{1,s}$ and note that $L_{0,s}=L$.
Due to self-adjointness of $L$, the adjoint of $L_{t,s}$ is
\begin{equation}
\label{E:adjoint_L_t_s}
L_{t,s}^{*}=L_{-t,s}.
\end{equation}

By definition, $L_{t,s}(k)=L(k+it\beta_s)$ for any $k$ in $\mathbb{C}^d$ and therefore, \mref{fl_spectrum} yields
\begin{equation}
\label{E:band_functions_L_t_s}
\sigma(L_{t,s})=\bigcup_{k \in \mathcal{O}} \sigma(L(k+it\beta_s)) \supseteq \{\lambda_{j}(k+it\beta_s)\}_{k\in \mathcal{O}}.
\end{equation}

The Schwartz kernel $G_{s,\lambda}$ of the resolvent operator $R_{s,\lambda}:=(L_s-\lambda)^{-1}$ is
\begin{equation}
\label{E:Green_Ls}
G_{s,\lambda}(x,y)=e^{\beta_s\cdot x}G_{\lambda}(x,y)e^{-\beta_s \cdot y}=e^{\beta_s \cdot (x-y)}G_{\lambda}(x,y).
\end{equation}
Thus, instead of finding asymptotics of $G_{\lambda}$, we can focus on the asymptotics of $G_{s,\lambda}$.

By \mref{E:band_functions_L_t_s} and Proposition \ref{singularity}, $\lambda$ is not in the spectrum of $L_{t,s}$ for any $s \in \mathbb{S}^{d-1}$ and $t \in [0,1)$.
Let us consider
\begin{equation*}
R_{t,s,\lambda}f:=(L_{t,s}-\lambda)^{-1}f, \quad f \in L^{2}_{comp}(\mathbb{R}^d),
\end{equation*}
where $L^{2}_{comp}$ stands for compactly supported functions in $L^2$.

Applying Lemma \ref{L:floquet}, we have
\begin{equation*}
\widehat{R_{t,s,\lambda}f}(k)=(L_{t,s}(k)-\lambda)^{-1}\widehat{f}(k), \quad (t,k) \in [0,1) \times \mathcal{O}.
\end{equation*}
We consider the sesquilinear form
\begin{equation*}
%\label{E:bilinform}
(R_{t,s,\lambda}f,\varphi)=(2\pi)^{-d}\int_{\mathcal{O}}\left( (L_{t,s}(k)-\lambda)^{-1}\widehat{f}(k), \widehat{\varphi}(k) \right)dk,
\end{equation*}
where $\varphi \in L^2_{comp}(\mathbb{R}^d)$.

In the next lemma
(see Subsection 9.2 for its proof), we show the weak convergence of $R_{t,s,\lambda}$ in $L^{2}_{comp}$ as $t \nearrow 1$  and  introduce the limit operator $\displaystyle R_{s,\lambda}=\lim_{t \rightarrow 1^{-}}R_{t,s,\lambda}$. The limit operator
$\displaystyle R_{s,\lambda}$ is central in our study of the asymptotics of the Green's function.

\begin{lemma}
\label{L:bilin}
Let $d\geq 2$. Under Assumption A, the following equality holds:
\begin{equation}
\label{E:limit}
\lim_{t \rightarrow 1^-}(R_{t,s,\lambda}f,\varphi)=(2\pi)^{-d}\int_{\mathcal{O}}\left( L_{s}(k)-\lambda)^{-1}\widehat{f}(k),\widehat{\varphi}(k)\right)dk.
\end{equation}
The integral in the right hand side of \mref{E:limit} is absolutely convergent for $f,\varphi$ in $L^{2}_{comp}(\mathbb{R}^d)$. Thus, the Green's function $G_{s,\lambda}$ is the integral kernel of the operator $R_{s,\lambda}$ defined as follows
\begin{equation}
\label{E:R_s}
\widehat{R_{s,\lambda}f}(k)=(L_{s}(k)-\lambda)^{-1}\widehat{f}(k).
\end{equation}
\end{lemma}

%%%%%%%%%%%%%%%%%%%%%%%
\subsection{Singling out the principal term in $R_{s,\lambda}$}
%%%%%%%%%%%%%

By \mref{E:R_s}, the Green's function $G_{s,\lambda}$ is the integral kernel of the operator $R_{s,\lambda}$ with the domain $L^{2}_{comp}(\mathbb{R}^d)$. The inversion formula \mref{E:inversion1} gives
\begin{equation*}
\label{E:R_s_exp1}
R_{s,\lambda}f(x)=(2\pi)^{-d}\int_{\mathcal{O}}e^{ik\cdot x}(L_{s}(k)-\lambda)^{-1}\widehat{f}(k,x)dk, \quad x \in \mathbb{R}^d.
\end{equation*}
%and
%\begin{equation*}
%\label{E:R_s_exp2}
%R_{s,\lambda}f(x+\gamma)=(2\pi)^{-d}\int_{\mathbb{T}^*}e^{ik.(x+\gamma)}(L_{s}(k)-\lambda)^{-1}\widehat{f}(k,x)dk, \quad x \in \mathbb{T}, \gamma \in \mathbb{Z}^d.
%\end{equation*}
The purpose of this part is to single out the part of the above integral that is responsible for the leading term of the Green's function asymptotics.

To find the Schwartz kernel of $R_{s,\lambda}$, it suffices to consider functions $f \in C^\infty_c(\mathbb R^d)$. Our first step is to localize the integral around the point $k_0$.
Let us consider a connected neighborhood $V$ of $k_0$ on which there exist nonzero $\mathbb{Z}^d$-periodic (in $x$) functions $\phi_{z}(x), z \in V$
satisfying 1) $L(z)\phi_z=\lambda_j(z)\phi_z$ and 2) each $\phi_z$ spans the eigenspace corresponding to the eigenvalue $\lambda_j(z)$ of the operator $L(z)$. According to \textbf{(P3)}, $\lambda_{j}(V) \subseteq B(0,\epsilon_0)$ and $\partial B(0,\epsilon_0) \subseteq \rho(L(z))$ when $z \in V$. For such $z$, let $P(z)$ be the Riesz projection of $L(z)$ that projects $L^{2}(\mathbb{T})$ onto the eigenspace spanned by $\phi_z$, i.e.,
\begin{equation*}
\label{E:Riesz_projector}
P(z)=-\frac{1}{2\pi i}\oint_{|\alpha|=\epsilon_0}(L(z)-\alpha)^{-1}d\alpha.
\end{equation*}
Taking the adjoint, we get
\begin{equation*}
\label{E:dual_Riesz_projector}
P(z)^{*}=-\frac{1}{2\pi i}\oint_{|\alpha|=\epsilon_0}(L(\overline{z})-\alpha)^{-1}d\alpha=P(\overline{z}),
\end{equation*}
which is the Riesz projection from $L^{2}(\mathbb{T})$ onto the eigenspace spanned by $\phi_{\overline{z}}$. Recall that due to \mref{E:beta_s_in_V}, by choosing $|\lambda|$ small enough, there exists $r_0>0$ (independent of $s$) such that $k\pm i\beta_s \in V$ for $k \in \overline{D}(k_0,r_0) \cap \mathbb{R}^d$.
We denote $ P_s(k):=P(k+i\beta_s)$ for such real $k$. Then $P_s(k)$ is the projector onto the eigenspace spanned by $\phi(k+i\beta_s)$ and $P_{s}(k)^*=P(k-i\beta_s)$.
Additionally, due to \textbf{(P6)},
\begin{equation}
\label{E:form_P_s}
P_{s}(k)g=\frac{(g, \phi(k-i\beta_s))_{L^{2}(\mathbb{T})}}{(\phi(k+i\beta_s), \phi(k-i\beta_s))_{L^{2}(\mathbb{T})}}\phi_{s}(k),
\quad \forall g \in L^{2}(\mathbb{T}).
\end{equation}

%Let $\nu, \tilde{\eta}$ be two cut-off smooth functions on $\mathcal{O}$ such that $\supp(\nu) \Subset  D(k_0,r_0)$, $\supp(\tilde{\eta}) \Subset  D(k_0,r_0)$, $\nu=1$ around $k_0$ and $\tilde{\eta}=1$ on $\supp(\nu)$. Then we define $\eta:=\nu \tilde{\eta}$.

Let $\eta$ be a cut-off smooth function on $\mathcal{O}$ such that $\supp(\eta) \Subset  D(k_0,r_0)$ and $\eta=1$ around $k_0$.

We decompose $\widehat{f}=\eta \widehat{f}+ (1-\eta)\widehat{f}$. When $k \neq k_0$, the operator $L_{s}(k)-\lambda$ is invertible by Proposition \ref{singularity}. Hence, the following function is well-defined and smooth with respect to $(k,x)$ on $\mathbb{R}^{d} \times \mathbb{R}^d$:
\begin{equation*}
\label{E:good_ug}
\widehat{u_g}(k,x)=(L_{s}(k)-\lambda)^{-1}(1-\eta(k))\widehat{f}(k,x).
\end{equation*}
Using Lemma \ref{L:floquet}, smoothness of $\widehat{u_g}$ implies that $u_g$ has rapid decay in $x$.
Now we want to solve
\begin{equation}
\label{E:solve_eqn}
(L_{s}(k)-\lambda)\widehat{u}(k)=\eta(k)\widehat{f}(k).
\end{equation}
Let $Q_{s}(k)=I-P_{s}(k)$ and we denote the ranges of projectors $P_s(k)$, $Q_s(k)$ by $R(P_s(k)), R(Q_s(k))$ respectively.
We are interested in decomposing the solution $\widehat{u}$ into a sum of the form $\widehat{u_1}+\widehat{u_2}$ where $\widehat{u_1}=P_{s}(k)\widehat{u_1}$ and $\widehat{u_2}=Q_{s}(k)\widehat{u_2}$. Let $\widehat{f_1}=P_{s}(k)\eta(k)\widehat{f}$ and $\widehat{f_2}=Q_{s}(k)\eta(k)\widehat{f}$. Observe that since the Riesz projection $P_s(k)$ commutes with the operator $L_s(k)$ and $R(P_s(k))$ is invariant under the action of $L_s(k)$, we have $Q_s(k)L_s(k)P_s(k)=P_s(k)L_s(k)Q_s(k)=0$ and $Q_s(k)L_s(k)Q_s(k)=L_s(k)Q_s(k)$. Thus, the problem of solving \mref{E:solve_eqn} can be reduced to the following block-matrix structure form
\begin{equation*}
\label{E:block_matrix}
 \left(
    \begin{array}{r@{}c|c@{}l}
  &    \begin{matrix}
         (L_{s}(k)-\lambda)P_{s}(k)
      \end{matrix} & \mbox{0} & \\\hline
  &    \mbox{0} &
       \begin{matrix}
       (L_{s}(k)-\lambda)Q_{s}(k)
      \end{matrix}
    \end{array}
\right)
\left(
\begin{array}{c}
\widehat{u_1}\\ \widehat{u_2}
\end{array}
\right)=\left(
\begin{array}{c}
\widehat{f_1}\\ \widehat{f_2}
\end{array}
\right).
\end{equation*}
When $k$ is close to $k_0$, $$B(0,\epsilon_0) \cap \sigma(L_{s}(k)_{|R(Q_{s}(k))})=B(0,\epsilon_0)\cap \sigma(L(k+i\beta_s))\setminus \{\lambda_{j}(k+i\beta_s)\}=\emptyset.$$
Since $\lambda=\lambda_{j}(k_0+i\beta_s)\in B(0,\epsilon_0)$, $\lambda$ must belong to $\rho(L_{s}(k)|_{R(Q_{s}(k))})$. Hence, the operator function $\widehat{u_2}(k)=(L_s(k)-\lambda)^{-1}Q_s(k)\widehat{f_2}(k)$ is well-defined and smooth in $k$ and hence by Lemma \ref{L:floquet} again, $u_2$ has rapid decay when $|x| \rightarrow \infty$. Indeed, we claim that the Schwartz kernel coming from the operator-valued function $(1-\eta(k))(L_s(k)-\lambda)^{-1}+\eta(k)((L_s(k)-\lambda)|_{R(Q_s(k))})^{-1}Q_s(k)$ decays fast enough to be included in the error term $r(x,y)$ in \mref{main_asymp}. We shall give a microlocal proof of this claim in Section 7.

The $u_1$ term contributes the leading asymptotics for the Schwartz kernel $G_{s,\lambda}$. Therefore, we only need to solve the equation $(L_s(k)-\lambda)P_s(k)\widehat{u_1}=\widehat{f_1}$ on the one-dimensional range of $P_s(k)$.

Applying \mref{E:form_P_s}, we can rewrite
$$\widehat{f_1}(k)=\frac{\eta(k)(\widehat{f},\phi(k-i\beta_s))_{L^{2}(\mathbb{T})}}{(\phi(k+i\beta_s),\phi(k-i\beta_s))_{L^2(\mathbb{T})}}\phi(k+i\beta_s),$$
so that equation becomes
\begin{equation*}
(L_s(k)-\lambda)\frac{(\widehat{u_1},\phi(k-i\beta_s))_{L^{2}(\mathbb{T})}}{(\phi(k+i\beta_s),\phi(k-i\beta_s))_{L^2(\mathbb{T})}}\phi(k+i\beta_s)=\frac{\eta(k)(\widehat{f},\phi(k-i\beta_s))_{L^{2}(\mathbb{T})}}{(\phi(k+i\beta_s),\phi(k-i\beta_s))_{L^2(\mathbb{T})}}\phi(k+i\beta_s).
\end{equation*}
So,
\begin{equation*}
\frac{(\lambda_{j}(k+i\beta_s)-\lambda)(\widehat{u_1},\phi(k-i\beta_s))_{L^{2}(\mathbb{T})}}{(\phi(k+i\beta_s),\phi(k-i\beta_s))_{L^2(\mathbb{T})}}\phi(k+i\beta_s)=\frac{\eta(k)(\widehat{f},\phi(k-i\beta_s))_{L^{2}(\mathbb{T})}}{(\phi(k+i\beta_s),\phi(k-i\beta_s))_{L^2(\mathbb{T})}}\phi(k+i\beta_s).
\end{equation*}
In addition to the equation $\widehat{u_1}=P_s(k)\widehat{u_1}$, $\widehat{u_1}$ must also satisfy
\begin{equation*}
(\lambda_{j}(k+i\beta_s)-\lambda)(\widehat{u_1},\phi(k-i\beta_s))_{L^{2}(\mathbb{T})}=\eta(k)(\widehat{f},\phi(k-i\beta_s))_{L^{2}(\mathbb{T})}.
\end{equation*}
Thus, we define
\begin{equation*}
\label{E:def_u1}
\widehat{u_1}(k,\cdot):=\frac{\eta(k)\phi(k+i\beta_s,\cdot)(\widehat{f},\phi(k-i\beta_s))_{L^{2}(\mathbb{T})}}{(\phi(k+i\beta_s),\phi(k-i\beta_s))_{L^2(\mathbb{T})}(\lambda_j(k+i\beta_s)-\lambda)}.
\end{equation*}
By the inverse Floquet transform \mref{E:inversion1},
\begin{equation*}
u_{1}(x)=(2\pi)^{-d}\int_{\mathcal{O}}e^{ik\cdot x}\frac{\eta(k)\phi(k+i\beta_s,x)(\widehat{f},\phi(k-i\beta_s))_{L^{2}(\mathbb{T})}}{(\phi(k+i\beta_s),\phi(k-i\beta_s))_{L^2(\mathbb{T})}(\lambda_j(k+i\beta_s)-\lambda)}dk,
\end{equation*}
for any $x \in \mathbb{R}^d$.
%%%%%%%%%%%%%%%%%%%%%%
\subsection{A reduced Green's function.}
%%%%%%%%%%%%%%%%%%%%

We are now ready for setting up \textbf{a reduced Green's function $G_0$}, whose asymptotic behavior reflects exactly the leading term of the asymptotics of the Green's function $G_{s,\lambda}$. We introduce $G_0(x,y)$ (roughly speaking) as the Schwartz kernel of the restriction of the operator $R_{s,\lambda}$ onto the one-dimensional range of $P_s$ (which is the direct integral of idempotents $P_s(k)$) as follows:
\begin{equation*}
u_{1}(x)=\int_{\mathbb{R}^d}G_{0}(x,y)f(y)dy, \quad x \in \mathbb{R}^d,
\end{equation*}
where $f$ is in $L^2_{comp}(\mathbb{R}^d)$.

We recall from \mref{eigen_inner_prod} that $F(k+i\beta_s)$ is the inner product $(\phi(k+i\beta_s),\phi(k-i\beta_s))_{L^2(\mathbb{T})}$.
As in \cite{KR}, we notice that
\begin{equation*}
\begin{split}
u_1(x)&=(2\pi)^{-d}\int_{\mathcal{O}}\int_{\mathbb{T}}e^{ik\cdot x}\eta(k)\widehat{f}(k,y)\frac{\overline{\phi(k-i\beta_s,y)}\phi(k+i\beta_s,x)}{F(k+i\beta_s)(\lambda_j(k+i\beta_s)-\lambda)}dydk\\
&=(2\pi)^{-d}\int_{\mathcal{O}}\eta(k)\int_{[0,1]^d}\sum_{\gamma \in \mathbb{Z}^d}f(y-\gamma)e^{ik\cdot (x+\gamma-y)}\frac{\overline{\phi(k-i\beta_s,y)}\phi(k+i\beta_s,x)}{F(k+i\beta_s)(\lambda_j(k+i\beta_s)-\lambda)}dydk\\
&=(2\pi)^{-d}\int_{\mathcal{O}}\eta(k)\sum_{\gamma \in \mathbb{Z}^d}\int_{[0,1]^d+\gamma}f(y)e^{ik\cdot (x-y)}\frac{\overline{\phi(k-i\beta_s,y+\gamma)}\phi(k+i\beta_s,x)}{F(k+i\beta_s)(\lambda_j(k+i\beta_s)-\lambda)}dydk\\
&=(2\pi)^{-d}\int_{\mathcal{O}}\eta(k)\sum_{\gamma \in \mathbb{Z}^d}\int_{[0,1]^d+\gamma}f(y)e^{ik\cdot (x-y)}\frac{\overline{\phi(k-i\beta_s,y)}\phi(k+i\beta_s,x)}{F(k+i\beta_s)(\lambda_j(k+i\beta_s)-\lambda)}dydk\\
&=(2\pi)^{-d}\int_{\mathbb{R}^d}f(y)\left(\int_{\mathcal{O}}\eta(k)e^{ik\cdot (x-y)}\frac{\overline{\phi(k-i\beta_s,y)}\phi(k+i\beta_s,x)}{F(k+i\beta_s)(\lambda_j(k+i\beta_s)-\lambda)}dk\right)dy.
\end{split}
\end{equation*}
Therefore, our \textbf{reduced Green's function is}
\begin{equation}
\label{E:formula_G0}
G_0(x,y)=(2\pi)^{-d}\int_{\mathcal{O}}\eta(k)e^{ik\cdot (x-y)}\frac{\phi(k+i\beta_s,x)\overline{\phi(k-i\beta_s,y)}}{F(k+i\beta_s)(\lambda_j(k+i\beta_s)-\lambda)}dk.
\end{equation}

%%%%%%%%%%%%%%%%%%%%%%
\section{Asymptotics of the Green's function}
%%%%%%%%%%%%%%%

Let $(e_1,\dots,e_d)$ be the standard orthonormal basis in $\mathbb{R}^d$. Fixing $\omega \in \mathbb{S}^{d-1}$, we would like to show that the asymptotics \mref{main_asymp_local} will hold for all $(x,y)$ such that $x-y$ belongs to a conic neighborhood containing $\omega$.
Without loss of generality, suppose that $\omega \neq e_1$.

Now let $\mathcal{R}_s$ be the rotation in $\mathbb{R}^d$ such that $\mathcal{R}_{s}(s)=e_{1}$ and $\mathcal{R}_s$ leaves the orthogonal complement of the subspace spanned by $\{s,e_{1}\}$ invariant. We define $e_{s,j}:=\mathcal{R}_{s}^{-1}(e_j)$, for all $j=2,\dots,d$.
Then, $\langle s,e_{s,p} \rangle=\langle e_1,e_p\rangle=0$ and $\langle e_{s,p},e_{s,q} \rangle=\langle e_p,e_q\rangle=\delta_{p,q}$ for $p,q>1$. In other words,
\begin{center}
$\{s,e_{s,2},\dots,e_{s,d}\}$ is an orthonormal basis of $\mathbb{R}^d$.
\end{center}

Then around $\omega$, we pick a compact coordinate patch $\mathcal{V}_{\omega}$, so that the $\mathbb{R}^{d(d-1)}$-valued function $e(s)=(e_{s,l})_{2 \leq l \leq d}$ is smooth in a neighborhood of $\mathcal{V}_{\omega}$.

We use the same notation for $\mathcal{R}_s$ and its $\mathbb{C}$-linear extension to $\mathbb{C}^d$. %$\mathcal{R}_s(a+ib)=\mathcal{R}_s(a)+i\mathcal{R}_s(b)$ for any $a,b \in \mathbb{R}^d$.

%%%%%%%%%%%%%%
\subsection{The asymptotics of the leading term of the Green's function}
%%%%%%%%%%%

We introduce the function $\rho(k,x,y)$ on $D(k_0,r_0)\times \mathbb{R}^d \times \mathbb{R}^d$ as follows:
\begin{equation*}
\rho(k,x,y)=\frac{\phi(k+i\beta_s,x)\overline{\phi(k-i\beta_s,y)}}{F(k+i\beta_s)}.
\end{equation*}
where $F$ is defined in \mref{eigen_inner_prod} and $D(k_0, r_0)$ is described in Section 5.3.

Due to Proposition \ref{P:joint_continuity}, the function $\rho$ is in $C^{\infty}(\overline{D(k_0,r)}\times \mathbb{R}^{d}\times \mathbb{R}^d)$.
For each $(x,y)$, the Taylor expansion around $k_0$ of $\rho(k)$ gives
\begin{equation}
\label{E:taylor_rho}
\rho(k,x,y)=\rho(k_0,x,y)+\rho'(k,x,y)(k-k_0),
\end{equation}
where $\rho' \in C^{\infty}(\overline{D(k_0,r_0)}\times \mathbb{R}^d \times \mathbb{R}^d, \mathbb{C}^d)$. Note that for  $z \in V$, $\phi(z,x)$ is $\mathbb{Z}^d$-periodic in $x$ and thus, $\rho$ and $\rho'$ are $\mathbb{Z}^d \times \mathbb{Z}^d$-periodic in $(x,y)$.
Since our integrals are taken with respect to $k$, it is safe to write $\rho(k_0)$ instead of $\rho(k_0,x,y)$. We often omit the variables $x,y$ in $\rho$  if no confusion can arise.

Let $\mu(k):=\eta(k+k_0)$ be a cut-off function supported near 0, where $\eta$ is introduced in Subsection 5.3.
We define
\begin{equation}
\label{E:main_integrals}
\begin{split}
I&:=(2\pi)^{-d}\int_{\mathcal{O}}e^{i(k-k_0)\cdot (x-y)}\frac{\mu(k-k_0)}{\lambda_j(k+i\beta_s)-\lambda}dk,\\
J&:=(2\pi)^{-d}\int_{\mathcal{O}}e^{i(k-k_0)\cdot (x-y)}\frac{\mu(k-k_0)(k-k_0)\rho'(k,x,y)}{\lambda_j(k+i\beta_s)-\lambda}dk.
\end{split}
\end{equation}
Hence, we can represent the reduced Green's function as
$$G_0(x,y)=e^{ik_0\cdot (x-y)}(\rho(k_0)I+J).$$

The rest of this subsection is devoted to computing the asymptotics of the main integral $I$, which gives the leading term in asymptotic expansion of the reduced Green's function $G_0(x,y)$ as $|x-y| \rightarrow \infty$.

By making the change of variables $\xi=(\xi_1,\xi')=\mathcal{R}_s(k-k_0)$, we have
\begin{equation}
\label{E:integral_I}
I=(2\pi)^{-d}\int_{\mathbb{R}^d}e^{i|x-y|\xi_1}\frac{\mu(\xi_1,\xi')}{(\lambda_j\circ \mathcal{R}^{-1}_s)(\xi+\mathcal{R}_s(k_0+i\beta_s))-\lambda}d\xi.
\end{equation}

We introduce the following function defined on some neighborhood of $0$ in $\mathbb{C}^d$:
$$W_{s}(z):=(\lambda_j\circ \mathcal{R}^{-1}_s)(-iz+\mathcal{R}_s(k_0+i\beta_s))-\lambda.$$
It is holomorphic near 0 (on $i\mathcal{R}_s(V)$) and $W_s(0)=0$. Then $W_s(iz)$ is the analytic continuation to the domain $\mathcal{R}_s(V)$ of the denominator of the integrand in \mref{E:integral_I}. For a complex vector $z=(z_1,\dots,z_d) \in \mathbb{C}^d$, we write $z=(z_1, z')$, where $z'=(z_2,\dots,z_d)$.

The following proposition provides a factorization of $W_s$ that is crucial for computing the asymptotics of the integral $I$.
\begin{prop}
\label{P:factorization_Ws}
There exist $r>0$ and $\epsilon>0$ (independent of $s \in \mathcal{V}_{\omega}$), such that $W_s$ has the decomposition
\begin{equation}
\label{E:factorization_Ws}
W_s(z)=(z_1-A_{s}(z'))B_s(z), \quad \forall z=(z_1, z') \in B(0,r) \times D'(0,\epsilon). \footnote{See Notation 2.8 (a) in Section 2 for the definitions of $B(0,r)$ and $D'(0, \epsilon)$.}
\end{equation}
Here the functions $A_s$, $B_s$ are holomorphic in $D'(0,\epsilon)$ and $B(0,r) \times D'(0,\epsilon)$ respectively such that $A_s(0)=0$ and $B_s$ is non-vanishing on $B(0,r) \times D'(0,\epsilon)$. Also, these functions and their derivatives depend continuously on $s$. Moreover for $z' \in D'(0,\epsilon)$,
\begin{equation}
\label{E:A_s}
A_s(z')=\frac{1}{2}z'\cdot Q_s z'+O(|z'|^3),
\end{equation}
where $O(|z'|^3)$ is uniform in $s$ when $z' \rightarrow 0$ and $Q_s$ is the positive definite $(d-1) \times (d-1)$ matrix
\begin{equation}
\label{E:Q_s}
Q_s=-\frac{1}{|\nabla E(\beta_s)|}\Big(e_{s,p}\cdot \Hess{(E)}(\beta_s)e_{s,q} \Big)_{2 \leq p,q \leq d}.
\end{equation}
\end{prop}
\begin{proof}
By Cauchy-Riemann equations for $W_s$ and \mref{E:gradient_E_s},
\begin{equation}
\label{E:gradient_Ws}
\frac{\partial W_s}{\partial z_1}(0)=\frac{\partial W_s}{\partial \xi_1}(0)=-i\nabla\lambda_j(k_0+i\beta_s)\cdot\mathcal{R}^{-1}_s e_1=-\nabla E(\beta_s)\cdot s=|\nabla E(\beta_s)|>0.
\end{equation}
Thus $0$ is a simple zero of $W_s$. Due to smoothness in $s$ of $W_s$ and $\beta_s$, we have
\begin{equation}
\label{E:min_gradient_Ws}
c:=\min_{s \in \mathcal{V}_{\omega}}\frac{\partial W_s}{\partial z_1}(0)\geq \min_{s \in \mathbb{S}^{d-1}}\left|\nabla E(\beta_s)\right|>0.
\end{equation}
Applying the Weierstrass preparation theorem (see Theorem \ref{Weierstrass}), we obtain the decomposition \mref{E:factorization_Ws} on a neighborhood of $0$.

To show that this neighborhood can be chosen such that it does not depend on $s$, we have to chase down how the neighborhood is constructed in the proof of theorem 7.5.1 in \cite{Hormander} (only the first three lines of the proof there matter) and then show that all steps in this construction can be done independently of $s$.

In the first step of the construction, we need $r>0$ such that $W_s(z_1,0')\neq 0$ when $0<|z_1|<2r$. The mapping $\displaystyle (s,z) \mapsto \frac{\partial W_s}{\partial z_1}(z)=-i\nabla\lambda_j(-i\mathcal{R}_s^{-1}z+k_0+i\beta_s)\cdot s$ is jointly continuous on $\mathcal{V}_{\omega}\times \mathcal{R}_s(V)$ and the value of this mapping at $z=0$ is greater or equal than $c$ due to \mref{E:gradient_Ws} and \mref{E:min_gradient_Ws}. Therefore, $\displaystyle \left|\frac{\partial W_s}{\partial z_1}(z)\right|>c/2$ in some open neighborhood $X_s \times Y_s$ of $(s,0)$ in $\mathcal{V}_{\omega} \times \mathbb{C}^d$.
By compactness, $\displaystyle \mathcal{V}_{\omega} \subseteq \bigcup_{k=1}^N X_{s(k)}$ for a finite collection of points $s_1,\dots,s_N$ on the sphere. Let $Y$ be the intersection of all $Y_{s_k}$ and let $r>0$ such that $D(0,2r) \subseteq Y$. Note that $r$ is independent of $s$. We claim $r$ has the desired property.  Observe that for $|z|<2r$, we have $\displaystyle \left|\frac{\partial W_s}{\partial z_1}(z)\right|>\frac{c}{2}$ for any $s$ in $\mathcal{V}_{\omega}$.
For a proof by contraction, suppose that there is some $z_1$ such that $0<|z_1|<2r$ and $W_s(z_1,0')=0=W_s(0,0')$ for some $s$.
Applying Rolle's theorem to the function $t \in [0,1] \mapsto W_s(tz_1,0')$ yields $\displaystyle \frac{\partial W_s}{\partial z_1}(tz_1,0')=0$ for some $t \in (0,1)$. Consequently, $(tz_1,0') \notin D(0,2r)$ while $|tz_1|<|z_1|<2r$ (contradiction!).

For the second step of the construction, we want some $\delta>0$ (independent of $s$) such that $W_s(z) \neq 0$ when $|z_1|=r, |z'|<\delta$.
This can be done in a similar manner. Let $S(0,r) \subset \mathbb{C}$ be the circle with radius $r$. Now we consider the smooth mapping $\displaystyle W:(s,z_1,z') \mapsto W_s(z_1,z')$ where $z_1 \in S(0,r)$. Its value at each point $(s,z_1,0')$ is equal to $\displaystyle W_s(z_1,0')$, which is non-zero due to the choice of $r$ in the first step of the construction. Thus, it is also non-zero in some open neighborhood $\tilde{X}_{s,z_1} \times \tilde{Y}_{s,z_1} \times \tilde{Z}_{s,z_1}$ of $(s,z_1,0')$ in $\mathcal{V}_{\omega} \times S(0,r) \times \mathbb{C}^{d-1}$. We select points $s_1,\dots,s_M \in \mathcal{V}_{\omega}$ and $\gamma_1,\dots,\gamma_M \in S(0,r)$ so that the union of all $\tilde{X}_{s_k,\gamma_k}\times \tilde{Y}_{s_k,\gamma_k}, 1\leq k \leq M$ covers the compact set $\mathcal{V}_{\omega} \times S(0,r)$. Next we choose $\delta>0$ so that $D'(0,\delta)$ is contained in the intersection of these $\tilde{Z}_{s_k,z_k}$. Note that $\delta$ is independent of $s$ and also $z_1$. Of course $W_s(z_1,z')\neq 0$ for all $s$ and $z \in \{|z_1|=r, |z'|<\delta\}$. According to \cite{Hormander}, the decomposition \mref{E:factorization_Ws} holds in the polydisc $\{|z_1|<r, |z'|<\delta\}$.

Also, from the proof of Theorem 7.5.1 in \cite{Hormander}, the function $A_s$ is defined via the following formula
\begin{equation}
\label{E:formula_As}
z_1-A_s(z')=\exp\left(\frac{1}{2\pi i}\int_{|\omega|=r}\left(\frac{\partial W_s(\omega,z')}{\partial \omega}/W_s(\omega,z')\right)\log(z_1-\omega)d\omega\right). %, |z|<\epsilon
\end{equation}
The mappings $(s,z') \mapsto A_s(z')$ and $(s,z) \mapsto B_s(z)$ are jointly continuous due to \mref{E:factorization_Ws} and \mref{E:formula_As}. There exists $\displaystyle 0<\epsilon\leq \delta$ such that $\displaystyle \max_{s \in \mathcal{V}_{\omega}}|A_s(z')|<r$ whenever $|z'|<\epsilon$. We have the identity \mref{E:factorization_Ws} on $B(0,r)\times D'(0,\epsilon)$. Now, we show that this is indeed the neighborhood that has the desired properties.
Since $|z'|<\epsilon$ implies that the points $z=(A_s(z'),z') \in B(0,r) \times D'(0,\epsilon)$, we can evaluate \mref{E:factorization_Ws} at these points to obtain
\begin{equation}
\label{E:eqn_A_s}
W_s(A_s(z'),z')=0, \quad z' \in D'(0,\epsilon).
\end{equation}
By differentiating \mref{E:eqn_A_s}, we have
\begin{equation}
\label{E:diff_A_s}
\frac{\partial W_s}{\partial z_p}(A_s(z'),z')+\frac{\partial W_s}{\partial z_1}(A_s(z'),z')\frac{\partial A_s}{\partial z_p}(z')=0, \quad \text{for } p = 2,\dots,d.
\end{equation}
Observe that from the above construction, the term $\displaystyle \frac{\partial W_s}{\partial z_1}(A_s(z'),z')$ is always non-zero whenever $|z'|<\epsilon$. Consequently, all first-order derivatives of $A_s$ are jointly continuous in $(s,z)$. Similarly, we deduce by induction on $n \in \mathbb{N}^d$ that all derivatives of the function $A_s$ depend continuously on $s$ since after taking differentiation of the equation \mref{E:eqn_A_s} up to order $n$, the  $n$-order derivative term always goes with the nonzero term $\displaystyle \frac{\partial W_s}{\partial z_1}(A_s(z'),z')$ and the remaining terms in the sum are just lower order derivatives. Hence the same conclusion holds for all derivatives of $B_s$ by differentiating \mref{E:factorization_Ws}.

In particular, set $z'=0$ in \mref{E:diff_A_s} to obtain
\begin{equation}
\label{E:diff_A_s_at_0}
\frac{\partial W_s}{\partial z_p}(0)+\frac{\partial W_s}{\partial z_1}(0)\frac{\partial A_s}{\partial z_p}(0)=0, \quad \text{for } p = 2,\dots,d.
\end{equation}
Note that for $p>1$,
\begin{equation}
\label{E:gradient_A_s}
\frac{\partial W_s}{\partial z_p}(z)=-i\nabla \lambda_j(-i\mathcal{R}_s^{-1}z+k_0+i\beta_s)\cdot\mathcal{R}^{-1}_se_p.
\end{equation}
By substituting $z=0$,
\begin{equation}
\label{E:gradient_W_s}
\begin{split}
\frac{\partial W_s}{\partial z_p}(0)&=-i\nabla \lambda_j(k_0+i\beta_s)\cdot\mathcal{R}^{-1}_s e_p\\&=-\nabla E(\beta_s)\cdot e_{s,p}=-|\nabla E(\beta_s)|s\cdot e_{s,p}=0.
\end{split}
\end{equation}
\mref{E:gradient_Ws}, \mref{E:diff_A_s_at_0} and \mref{E:gradient_W_s} imply
\begin{equation}
\label{E:1st_gradient_A_s}
\frac{\partial A_s}{\partial z_p}(0)=0, \quad  \text{for } p = 2,\dots,d.
\end{equation}
Taking a partial derivative with respect to $z_q$ $(q>1)$ of \mref{E:gradient_A_s} at $z=0$, we see that
\begin{equation}
\label{E:Hessian_Ws}
\begin{split}
\frac{\partial^2 W_s}{\partial z_p\partial z_q}(0)&=\sum_{m=1}^d -\nabla(\frac{\partial \lambda_j}{\partial z_m}(k_0+i\beta_s))\cdot \mathcal{R}^{-1}_s e_p (\mathcal{R}^{-1}_s e_q)_{m}\\
&=-\sum_{m,n=1}^d \frac{\partial^2 \lambda_j}{\partial z_m \partial z_n}(k_0+i\beta_s)(e_{s,p})_{m}(e_{s,q})_{n}\\
&=e_{s,q}\cdot \Hess{(E)}(\beta_s)e_{s,p}.
\end{split}
\end{equation}
A second differentiation of \mref{E:diff_A_s} at $z=(A_s(z'),z')$ gives
\begin{equation}
\label{E:2nd_diff_A_s}
\begin{split}
0&=\left(\frac{\partial^2 W_s}{\partial z_p \partial z_q}(z)+
\frac{\partial W_s}{\partial z_1}(z)
\frac{\partial^2 A_s}{\partial z_p \partial z_q}(z')\right)\\
&+\left(\frac{\partial^2 W_s}{\partial z_1 \partial z_q}(z)\frac{\partial A_s}{\partial z_p}(z')+
\frac{\partial^2 W_s}{\partial z_p \partial z_1}(z)
\frac{\partial A_s}{\partial z_q}(z')+
\frac{\partial^2 W_s}{\partial z_1^2}(z)
\frac{\partial A_s}{\partial z_p}(z')
\frac{\partial A_s}{\partial z_q}(z')\right).
\end{split}
\end{equation}
At $z=0$, the sum in the second bracket of \mref{E:2nd_diff_A_s} is zero due to \mref{E:1st_gradient_A_s}. Thus,
\begin{equation}
\frac{\partial^2 A_s}{\partial z_p \partial z_q}(0)=-\left(\frac{\partial W_s}{\partial z_1}(0)\right)^{-1}\frac{\partial^2 W_s}{\partial z_p \partial z_q}(0) \quad (2 \leq p,q\leq d).
\end{equation}
Together with \mref{E:gradient_Ws} and \mref{E:Hessian_Ws}, the above equality becomes
\begin{equation}
\label{E:Hessian_A_s}
\frac{\partial^2 A_s}{\partial z_p \partial z_q}(0)=-\frac{1}{|\nabla E(\beta_s)|}\Big(e_{s,p}\cdot \Hess{(E)}(\beta_s)e_{s,q} \Big)_{2 \leq p,q \leq d}=Q_s.
\end{equation}
Consequently, by \mref{E:1st_gradient_A_s} and \mref{E:Hessian_A_s}, the Taylor expansion of $A_s$ at 0 implies \mref{E:A_s}.

Finally, the remainder term $O(|z'|^3)$ in the Taylor expansion \mref{E:A_s}, denoted by $R_{s,3}(z')$, can be estimated as follows:
\begin{equation*}
\begin{split}
|R_{s,3}(z')| &\lesssim |z'|^3 \max_{|\alpha|=3, 0 \leq t \leq 1}\left|\frac{\partial^{\alpha}A_s}{\partial z^{\alpha}}(tz')\right| \\
&\lesssim |z'|^3 \max_{|\alpha|=3, |y| \leq |z'|}\left|\frac{\partial^{\alpha}A_s}{\partial z^{\alpha}}(y)\right|.
\end{split}
\end{equation*}
Due to the continuity of third-order derivatives of $A_s$ on $\mathcal{V}_{\omega} \times D'(0,\epsilon)$,
\begin{equation}
\label{E:remainder_independent_s}
\lim_{|z'| \rightarrow 0}\max_{s \in \mathcal{V}_{\omega}}\frac{|R_{s,3}(z')|}{|z'|^3}<\infty.
\end{equation}
This proves the last claim of this proposition.
\end{proof}

We can now let the size of the support of $\eta$ $(\Subset \mathcal{O})$ be small enough such that the decomposition \mref{E:factorization_Ws} in Proposition \ref{P:factorization_Ws} holds on the support of $\mu$, i.e., $\supp(\mu) \Subset B(0,r) \times D'(0,\epsilon)$. Therefore, from \mref{E:integral_I}, we can represent the integral $I$ as follows:
\begin{equation}
\label{E:integral_I_rewrite}
I=(2\pi)^{-d}\int_{\mathbb{R}^d}e^{i|x-y|\xi_1}\frac{\mu(\xi_1,\xi')}{W_s(i\xi)}d\xi_1d\xi'=(2\pi)^{-d}\int_{|\xi'|<\epsilon}\int_{\mathbb{R}}\frac{e^{i|x-y|\xi_1}\tilde{\mu}_s(\xi_1,\xi')}{i\xi_1-A_s(i\xi')}d\xi_1d\xi',
\end{equation}
where $\tilde{\mu}_s(\xi)=\mu(\xi)(B_s(i\xi))^{-1}$. We extend $\tilde{\mu}_s$ to a smooth compactly supported function on $\mathbb{R}^d$ by setting $\tilde{\mu}_s=0$ outside its support. Since all derivatives of $\tilde{\mu}_s$ depend continuously on $s$, they are uniformly bounded in $s$. Let $\nu_s(t,\xi')$ be the Fourier transform in the variable $\xi_1$ of the function $\tilde{\mu}_s(-\xi_1,\xi')$ for each $\xi'\in \mathbb{R}^{d-1}$, i.e.,
\begin{equation*}
\nu_s(t,\xi')=\int_{-\infty}^{+\infty} e^{it\xi_1}\tilde{\mu}_s(\xi_1,\xi')d\xi_1.
\end{equation*}
By applying the Lebesgue Dominated Convergence Theorem, the function $\nu_s$ is continuous in $(s,t,\xi')$ on $\mathcal{V}_{\omega} \times \mathbb{R}^d$.
For such $\xi'$, $\nu_s(\cdot,\xi')$ is a Schwartz function in $t$ on $\mathbb{R}$. If we restrict $\xi'$ to a compact set $K$ containing $D'(0,\epsilon)$, then due to Lemma \ref{L:uniform}, for any $N>0$, $\nu_{s}(t,\xi')=O(|t|^{-N})$ uniformly in $s$ and $\xi'$ as $t \rightarrow \infty$. We also choose $\epsilon$ small enough such that whenever $|\xi'|<\epsilon$, the absolute value of the remainder term $O(|\xi'|^3)$ in \mref{E:A_s} is bounded from above by $\frac{1}{4}\xi'\cdot Q_s\xi'$. Note that $\epsilon$ is still independent of $s$, because the term $O(|\xi'|^3)/|\xi'|^3$ is uniformly bounded by the quantity in \mref{E:remainder_independent_s}.
Meanwhile, each positive definite matrix $Q_s$ dominates the positive matrix $\gamma_{\omega} I_{(d-1) \times (d-1)}$, where $\gamma_{\omega}>0$ is the smallest among all the eigenvalues of all matrices $Q_s$ $(s \in \mathcal{V}_{\omega})$. This implies that if $0<|\xi'|<\epsilon$, then
\begin{equation*}
\begin{split}
\Re(i\xi_1-A_s(i\xi'))&=-\Re(A_s(i\xi'))=-\Re(-\frac{1}{2}\xi'\cdot Q_s\xi'+O(|\xi'|^3))
\\ &=\frac{1}{2}\xi'\cdot Q_s\xi'-\Re(O(|\xi'|^3))>\frac{1}{4}\gamma_{\omega}|\xi'|^2>0.
\end{split}
\end{equation*}
We thus can obtain the following integral representation for a factor in the integrand of $I$ (see \mref{E:integral_I_rewrite}):
\begin{equation}
\label{E:representation_trick}
\frac{1}{i\xi_1-A_s(i\xi')}=\int_{-\infty}^0 e^{(i\xi_1 -A_s(i\xi'))w}dw,
\quad (\xi_1,\xi') \in \mathbb{R}\times (D'(0,\epsilon)\setminus\{0\}).
\end{equation}

Hence
%It is fine to rewrite the integral as the following way even its integrand has a singularity at the point 0.
% Since the integrand is always integrable over the disc D'(0,\epsilon), we can simply ignore the point \xi'=0 and replace by the above equation in our integrand. All below calculations involved here are satisfied since each integral in these equalities are integrable on their corresponding domains. Fubini's theorem here is also valid to apply. Here we use Re(A_s(i\xi'))<0, w<0.
\begin{equation*}
\begin{split}
I&=\frac{1}{(2\pi)^{d}}\int_{|\xi'|<\epsilon}\int_{-\infty}^0 e^{-wA_s(i\xi')}\int_{-r}^r e^{i(w+|x-y|)\xi_1}\tilde{\mu}_s(\xi_1,\xi')d\xi_1 dw d\xi'\\
&=\frac{1}{(2\pi)^{d}}\int_{|\xi'|<\epsilon}\int_{-\infty}^{|x-y|} e^{(-t+|x-y|)A_s(i\xi')}\nu_s(t,\xi')dt d\xi'=I_1+I_2,
\end{split}
\end{equation*}
where
\begin{equation*}
\begin{split}
&I_1=(2\pi)^{-d}\int_{|\xi'|<\epsilon}\int_{-\infty}^{|x-y|/2}e^{(-t+|x-y|)A_s(i\xi')}\nu_s(t,\xi')dtd\xi',\\
&I_2=(2\pi)^{-d}\int_{|\xi'|<\epsilon}\int_{|x-y|/2}^{|x-y|}e^{(-t+|x-y|)A_s(i\xi')}\nu_s(t,\xi')dtd\xi'.
\end{split}
\end{equation*}

Note that $|e^{(|x-y|-t)A_s(i\xi')}|\leq 1$ when $t \leq |x-y|$. As $|x-y| \rightarrow \infty$, the term $I_2$ decays rapidly, since for any $n\in\mathbb{N}$,
\begin{equation}
\label{E:I2}
\begin{split}
|I_2| &\leq \int_{|\xi'|<\epsilon}\int_{|x-y|/2}^{|x-y|}|\nu_s(t,\xi')|dtd\xi' \\
&\lesssim \int_{|\xi'|<\epsilon}\int_{|x-y|/2}^{|x-y|}o(|t|^{-(n+1)}) dtd\xi'=o(|x-y|^{-n}).
\end{split}
\end{equation}

We substitute $\displaystyle x'=\sqrt{(|x-y|-t)}\xi'$ and use the identity
$$(|x-y|-t)A_s(i\xi')=-\frac{1}{2}x'\cdot Q_s x'+O\left(|x'|^3(|x-y|-t)^{-1/2}\right)$$
to obtain
\begin{equation}
\label{E:I1_1}
\begin{split}
I_1&=|x-y|^{-(d-1)/2}\frac{1}{(2\pi)^d}\int_{-\infty}^{|x-y|/2}\frac{1}{(1-\frac{t}{|x-y|})^{(d-1)/2}}\\
&\times\int_{|x'|<\epsilon\sqrt{(|x-y|-t)}}\exp{\left(-\frac{1}{2}x'\cdot Q_s x'+O\left(\frac{|x'|^3}{\sqrt{|x-y|-t}}\right)\right)}\nu_s\left(t,\frac{x'}{\sqrt{|x-y|-t}}\right)dx' dt.
\end{split}
\end{equation}

For clarity, we introduce the notation
$x_0:=|x-y|,$ and
$$I(x_0,t):=\int_{|x'|<\epsilon\sqrt{x_0-t}}\exp{\left(-\frac{1}{2}x'\cdot Q_s x'+O\left(\frac{|x'|^3}{\sqrt{x_0-t}}\right)\right)}\nu_s\left(t,\frac{x'}{\sqrt{x_0-t}}\right)dx'.$$

Due to our choice of $\epsilon$ and the definition of $\gamma_{\omega}$, we have
\begin{equation}
\label{E:exponential_est}
\sup_{s \in \mathcal{V}_{\omega}}\exp{\left(-\frac{1}{2}x'\cdot Q_s x'+O\left(\frac{|x'|^3}{\sqrt{x_0-t}}\right)\right)}\leq \exp{\left(-\frac{1}{4}\gamma_{\omega} |x'|^2\right)}.
\end{equation}

%Thus, as $t<x_0/2$, we derive
%\begin{equation*}
%\begin{split}
%\label{E:I(x,t)_est}
%&\sup_{s \in \mathcal{V}_{\omega}}\left|\int_{\epsilon (x_0-t)^{1/12}\leq |x'|<\epsilon\sqrt{x_0-t}}\exp{\left(-\frac{1}{2}x'\cdot Q_s x'+O\left(\frac{|x'|^3}{\sqrt{x_0-t}}\right)\right)}\nu_s\left(t,\frac{x'}{\sqrt{x_0-t}}\right)dx'\right|\\
%\leq &\sup_{s\in \mathcal{V}_{\omega}, |\xi'|<\epsilon}\left|\nu_s(t,\xi')\right|\cdot\int_{\epsilon (x_0-t)^{1/12}\leq |x'|}\exp{\left(-\frac{1}{4}\gamma_{\omega}|x'|^2\right)}dx'\\
%\lesssim &\int_{\epsilon (x_0/2)^{1/12}\leq |x'|}|x'|^{-(d+3)}dx'=O(x_0^{-1}).
%\end{split}
%\end{equation*}
%Consequently,
%$$\lim_{x_0 \rightarrow \infty}\int_{\epsilon (x_0-t)^{1/12}\leq|x'|<\epsilon\sqrt{x_0-t}}\exp{\left(-\frac{1}{2}x'\cdot Q_s x'+O\left(\frac{|x'|^3}{\sqrt{x_0-t}}\right)\right)}\nu_s\left(t,\frac{x'}{\sqrt{x_0-t}}\right)dx'=0$$
%uniformly in $s$.

Hence, for any $t<x_0/2$, the functions
$$\left|\chi_{D'(0, \epsilon \sqrt{x_0-t}}(x')\cdot\exp{\left(-\frac{1}{2}x'\cdot Q_s x'+O\left(\frac{|x'|^3}{\sqrt{x_0-t}}\right)\right)}\cdot\nu_s\left(t,\frac{x'}{\sqrt{x_0-t}}\right)\right|$$
are bounded from above by the integrable function (on $\mathbb{R}^{d-1}$)
$$\exp{\left(-\frac{1}{4}\gamma_{\omega}|x'|^2\right)}\cdot \sup_{|\xi'|<\epsilon}\left|\nu_s(t,\xi')\right|.$$
Moreover, for any $x' \in \mathbb{R}^{d-1}$, we get
$$\lim_{x_0 \rightarrow \infty}\chi_{D'(0, \epsilon \sqrt{x_0-t})}(x')\cdot \exp{\left(O\left(\frac{|x'|^3}{\sqrt{x_0-t}}\right)\right)}=1,$$
and
$$\lim_{x_0 \rightarrow \infty}\chi_{D'(0,\epsilon \sqrt{x_0-t})}(x')\cdot\nu_s\left(t,\frac{x'}{\sqrt{x_0-t}}\right)=\nu_s(t,0)$$
uniformly in $t\in (-\infty,x_0/2]$ and in $s \in \mathcal{V}_{\omega}$.

We now can apply the Lebesgue Dominated Convergence Theorem to obtain
\begin{equation*}
\begin{split}
\lim_{x_0 \rightarrow \infty}I(x_0,t)
=&\lim_{x_0 \rightarrow \infty}\int_{|x'|<\epsilon \sqrt{x_0-t}}\exp{\left(-\frac{1}{2}x'\cdot Q_s x'+O\left(\frac{|x'|^3}{\sqrt{x_0-t}}\right)\right)}\nu_s\left(t,\frac{x'}{\sqrt{x_0-t}}\right)dx'\\
=&\int_{\mathbb{R}^{d-1}}\exp{\left(-\frac{1}{2}x'\cdot Q_s x'\right)}\nu_s\left(t,0\right)dx'
=\frac{\nu_s(t,0)}{(\det{Q_s})^{1/2}}\int_{\mathbb{R}^{d-1}}\exp\left(-\frac{1}{2}|u'|^2\right)du'\\
=&\frac{(2\pi)^{(d-1)/2}}{(\det{Q_s})^{1/2}}\nu_s(t,0).
\end{split}
\end{equation*}
In the third equality above, we use the change of variables  $u':=Q_s^{1/2}x'$.

Likewise,
the functions $\displaystyle \chi_{(-\infty, x_0/2)}(t)(1-t/x_0)^{-(d-1)/2}\left|I(x_0,t)\right|$ are dominated (up to a constant factor)
by the function $\sup_{|\xi'|<\epsilon}|\nu_s(t,\xi')|$, which is integrable on $\mathbb{R}$.

Then a second application of the Lebesgue Dominated Convergence Theorem establishes
\begin{equation*}
\begin{split}
\lim_{x_0 \rightarrow \infty}x_0^{(d-1)/2}I_1&=\lim_{x_0 \rightarrow \infty}(2\pi)^{-d}\int_{-\infty}^{x_0/2}\left(1-\frac{t}{x_0}\right)^{-(d-1)/2} I(x_0,t)dt\\
&=(2\pi)^{-d}\int_{-\infty}^{\infty}\lim_{x_0 \rightarrow \infty}I(x_0,t)dt=\frac{(2\pi)^{-(d-1)/2}}{(\det{Q_s})^{1/2}}\frac{1}{2\pi}\int_{-\infty}^{\infty}\nu_s(t,0)dt \\&
=\frac{\tilde{\mu}_s(0)}{(2\pi)^{(d-1)/2}(\det{Q_s})^{1/2}}.
\end{split}
\end{equation*}
Here we use the Fourier inversion formula for $\nu_s$ in the last equality. Also, $\displaystyle \tilde{\mu}_s(0)=(B_s(0))^{-1}=(\frac{\partial W_s}{\partial z_1}(0))^{-1}=\frac{1}{|\nabla E(\beta_s)|}$ by \mref{E:gradient_Ws}. Consequently, the asymptotic of the integral $I$ as $|x-y| \rightarrow \infty$  is
\begin{equation}
\label{E:leadingterm}
\begin{split}
I&=I_1+I_2=\frac{|x-y|^{-(d-1)/2}}{(2\pi)^{(d-1)/2}(\det{Q_s})^{1/2}|\nabla E(\beta_s)|}+o(|x-y|^{-(d-1)/2})\\&
=\frac{|\nabla E(\beta_s)|^{(d-3)/2}|x-y|^{-(d-1)/2}}{(2\pi)^{(d-1)/2}\det{\left(-e_{s,p}\cdot \Hess{(E)}(\beta_s)e_{s,q} \right)^{1/2}_{2 \leq p,q \leq d}}}+o(|x-y|^{-(d-1)/2}).
\end{split}
\end{equation}
%%%%%%%%%%%%%%%%%%%%%%%%%%%%%
\subsection{Estimates of the integral $J$}
$ $
\vspace{0.3cm}

In this part, we want to show that the expression $J$ decays as $O\left(|x-y|^{-d/2}\right)$. Thus, taking into account \mref{E:leadingterm}, we conclude that $J$ does not contribute to the leading term of the reduced Green's function.

In \mref{E:taylor_rho}, we set the coordinate functions of $\rho'$ as $(\rho_1,\dots,\rho_d)$. Let us introduce the smooth function $\mu^{(l)}(k,x,y)=\rho_l(k+k_0,x,y)\mu(k)$ for any $k \in \mathbb{R}^d$. The support of $\mu^{(l)}$ (as a function of $k$ for each pair $(x,y)$) is contained in the support of $\mu$ and $\mu^{(l)}(k,\cdot,\cdot)$ is $\mathbb{Z}^d \times \mathbb{Z}^d$-periodic. We denote the components of a vector $k$ in $\mathbb{R}^d$ as $(k_1,\dots,k_d)$.
Observe that $J$ is the sum of integrals $J_l$ $(1 \leq l \leq d)$ if we define
\begin{equation}
\begin{split}
J_l:=(2\pi)^{-d}\int_{\mathcal{O}}e^{i(k-k_0)\cdot(x-y)}\frac{\mu^{(l)}(k-k_0,x,y)(k-k_0)_l}{\lambda_{j}(k+i\beta_s)-\lambda}dk.
\end{split}
\end{equation}
\begin{prop}
\label{remainderJ}
 As $|x-y| \rightarrow \infty$, we have $J_1=O\left(|x-y|^{-(d+1)/2}\right)$ and $J_l=O\left(|x-y|^{-d/2}\right)$ if $l>1$. In particular, $J=O\left(|x-y|^{-d/2}\right)$.
\end{prop}
\begin{proof}
Indeed, to treat these integrals, we need to re-examine the calculation in the previous subsection done for the integral $I$.
After applying the orthogonal transformation $\mathcal{R}_s$ on each integral $J_l$, we rewrite them under the form of \mref{E:integral_I_rewrite} as
\begin{equation}
J_l=(2\pi)^{-d}\int_{|\xi'|<\epsilon}\int_{\mathbb{R}}e^{i|x-y
|\xi_1}\frac{\tilde{\mu}_s^{(l)}(\xi_1,\xi',x,y)\xi_l}{i\xi_1-A_s(i\xi')}d\xi_1d\xi',
\end{equation}
where $\tilde{\mu}_s^{(l)}(\xi,x,y)$ is $\mu^{(l)}(\xi,x,y)(B_s(i\xi))^{-1}$ on the support of $\mu^{(l)}$  and vanishes elsewhere. Let $\nu_s^{(l)}(t,\xi',x,y)$ be the Fourier transform in $\xi_1$ of $\tilde{\mu}_s^{(l)}(\xi_1,\xi',x,y)$. If the parameter $s$ is viewed as another argument of our functions here, then $\nu_s^{(l)}(\cdot,\xi',x,y)$ is a Schwartz function for each quadruple $(s,\xi',x,y)$. It is elementary to check that the Fourier transform $\nu_s^{(l)}(t,\xi',x,y)$ is jointly continuous on $\mathcal{V}_{\omega} \times \mathbb{R} \times \mathbb{R}^{d-1} \times \mathbb{R}^d \times \mathbb{R}^d$ due to the corresponding property of $\tilde{\mu}_s^{(l)}(\xi,x,y)$.  Periodicity in $(x,y)$ of $\nu_s^{(l)}$ and Lemma \ref{L:uniform} imply the following decay:
\begin{equation}
\label{E:uniform_decay}
\lim_{t \rightarrow \infty}|t|^{N}\sup_{(s,\xi',x,y) \in \mathcal{V}_{\omega} \times \overline{D'(0,\epsilon)}\times \mathbb{R}^d \times \mathbb{R}^d}|\nu_s^{(l)}(t,\xi',x,y)|=0, \quad N\geq 0.
\end{equation}
In particular,
\begin{equation}
\label{E:uniform_bound}
\max_{1 \leq l \leq d}\sup_{(s,t,\xi',x,y) \in \mathcal{V}_{\omega} \times \mathbb{R} \times \overline{D'(0,\epsilon)}\times \mathbb{R}^d \times \mathbb{R}^d}|\nu_s^{(l)}(t,\xi',x,y)|<\infty
\end{equation}
and
\begin{equation}
\label{E:uniform_integrable}
S:=\max_{1 \leq l \leq d}\int_{\mathbb{R}}\sup_{(s,\xi',x,y) \in \mathcal{V}_{\omega} \times \overline{D'(0,\epsilon)}\times \mathbb{R}^d \times \mathbb{R}^d}|\nu_s^{(l)}(t,\xi',x,y)|dt<\infty.
\end{equation}
Recall that when $0<|\xi'|<\epsilon$, $\Re(A_s(i\xi'))<0$ and thus from \mref{E:uniform_bound},
\begin{equation}
\label{E:int_by_part_J1}
\lim_{t \rightarrow -\infty}e^{(-t+|x-y|)A_s(i\xi')}\nu_s^{(1)}(t,\xi',x,y)=0.
\end{equation}

\textbf{\underline{Case 1}:} $l=1$.\\
Using \mref{E:representation_trick}, \mref{E:int_by_part_J1} and integration by parts, we obtain
\begin{equation}
\label{E:J1_1}
\begin{split}
J_1&=\frac{1}{(2\pi)^{d}}\int_{|\xi'|<\epsilon}\int_{-\infty}^0 e^{-wA_s(i\xi')}\int_{-r}^r \xi_1 e^{i(w+|x-y|)\xi_1}\tilde{\mu}_s^{(1)}(\xi_1,\xi',x,y)d\xi_1 dw d\xi\\
&=-\frac{i}{(2\pi)^{d}}\int_{|\xi'|<\epsilon}\int_{-\infty}^{|x-y|} e^{(-t+|x-y|)A_s(i\xi')}\frac{d}{dt}\nu_s^{(1)}(t,\xi',x,y)dt d\xi'\\
&=-\frac{i}{(2\pi)^{d}}\int_{|\xi'|<\epsilon}\bigg(\nu_s^{(1)}(|x-y|,\xi',x,y)
+\int_{-\infty}^{|x-y|}A_s(i\xi')e^{(-t+|x-y|)A_s(i\xi')}\\
& \hspace{8.5cm} \times \nu_s^{(1)}(t,\xi',x,y)dt\bigg)d\xi'.
\end{split}
\end{equation}
Recall the notation $x_0=|x-y|$. The term
$$\int_{|\xi'|<\epsilon}\nu_s^{(1)}(x_0,\xi',x,y)d\xi'$$ decays rapidly in $x_0$, due to \mref{E:uniform_decay}. We decompose the other term $$\int_{-\infty}^{x_0}A_s(i\xi')e^{(x_0-t)A_s(i\xi')}\nu_s^{(1)}(t,\xi',x,y)dt$$ into two parts, where the first integral is taking over $\displaystyle (x_0/2,x_0]$ and the second one over $\displaystyle (-\infty,x_0/2]$. The first part decays rapidly, as $I_2$ in \mref{E:I2}.
Now we need to prove that the second part decays as $\displaystyle O(x_0^{(d+1)/2})$.
To do this, we use the change of variables $\displaystyle x'=\xi'\sqrt{x_0-t}$ to rewrite the remaining integral as
\begin{equation}
\label{E:J1_2}
\begin{split}
&x_0^{(d+1)/2}\int_{|\xi'|<\epsilon}\int_{-\infty}^{x_0/2}A_s(i\xi')e^{(-t+x_0)A_s(i\xi')}\nu_s^{(1)}(t,\xi',x,y)dtd\xi'\\
=&\int_{-\infty}^{x_0/2}\left(1-\frac{t}{x_0}\right)^{-(d+1)/2}\int_{|x'|<\epsilon \sqrt{x_0-t}}\left(-\frac{1}{2}x'\cdot Q_s x'+O\left(\frac{|x'|^3}{\sqrt{x_0-t}}\right)\right)\\
&\times \exp{\left(-\frac{1}{2}x'\cdot Q_s x'+O\left(\frac{|x'|^3}{\sqrt{x_0-t}}\right)\right)}\nu_s^{(1)}\left(t,\left(\frac{x'}{\sqrt{x_0-t}}\right),x,y\right)dx'dt.
\end{split}
\end{equation}
From \mref{E:uniform_integrable}, we derive
\begin{equation}
\label{E:t_integral_J1}
\begin{split}
\int_{-\infty}^{x_0/2}\left(1-\frac{t}{x_0}\right)^{-(d+1)/2}\sup_{(s,\xi',x,y) \in \mathcal{V}_{\omega} \times \overline{D'(0,\epsilon)}\times \mathbb{R}^d \times \mathbb{R}^d}|\nu_s^{(1)}(t,\xi',x,y)|dt\leq 2^{\frac{(d+1)}{2}}S.
\end{split}
\end{equation}
On the other hand, we recall that
\begin{equation*}
\Re\left(-\frac{1}{2}x'\cdot Q_s x'+O\left(\frac{|x'|^3}{\sqrt{x_0-t}}\right)\right) \leq -\frac{1}{4}\gamma_{\omega} |x'|^2.
\end{equation*}
The exponential term is majorized as follows:
\begin{equation*}
\begin{split}
\left|\left(-\frac{1}{2}x'\cdot Q_s x'+O\left(\frac{|x'|^3}{\sqrt{x_0-t}}\right)\right)\exp{\left(-\frac{1}{2}x'\cdot Q_s x'+O\left(\frac{|x'|^3}{\sqrt{x_0-t}}\right)\right)}\right| \\
\leq \left(\frac{1}{2}x'\cdot Q_sx'+O(\epsilon|x'|^2)\right)\exp{\left(-\frac{1}{4}\gamma_{\omega}|x'|^2\right)}.
\end{split}
\end{equation*}
Consequently,
\begin{equation}
\label{E:x'_integral_J1}
\begin{split}
\int_{|x'|<\epsilon\sqrt{x_0-t}}\bigg|\bigg(-\frac{1}{2}x'\cdot Q_s x'+O\bigg(\frac{|x'|^3}{\sqrt{x_0-t}}&\bigg)\bigg)\exp{\left(-\frac{1}{2}x'\cdot Q_s x'+O\left(\frac{|x'|^3}{\sqrt{x_0-t}}\right)\right)}\bigg|dx'\\
& \lesssim \int_{\mathbb{R}^{d-1}}|x'|^2\exp{\left(-\frac{1}{4}\gamma_{\omega} |x'|^2\right)}dx'<\infty.
\end{split}
\end{equation}
Combining  \mref{E:J1_1} through \mref{E:x'_integral_J1}, we deduce $J_1=O(x_0^{-(d+1)/2})$.\\

\textbf{\underline{Case 2}:} $l>1$.\\
Using \mref{E:representation_trick} and decomposing $J_l$ into two parts as in \textbf{Case 1}, we get
\begin{equation}
\label{E:J2_1}
\begin{split}
J_l&=\frac{1}{(2\pi)^{d}}\int_{|\xi'|<\epsilon}\int_{-\infty}^0 \xi_l e^{-wA_s(i\xi')}\int_{-r}^r e^{i(w+|x-y|)\xi_1}\tilde{\mu}_s^{(l)}(\xi_1,\xi',x,y)d\xi_1 dw d\xi\\
&=\frac{1}{(2\pi)^{d}}\int_{|\xi'|<\epsilon}\int_{-\infty}^{|x-y|} \xi_l e^{(-t+|x-y|)A_s(i\xi')}\nu_s^{(l)}(t,\xi',x,y)dt d\xi'\\
&=\frac{1}{(2\pi)^{d}}\int_{|\xi'|<\epsilon}\int_{-\infty}^{|x-y|/2} \xi_l e^{(-t+|x-y|)A_s(i\xi')}\nu_s^{(l)}(t,\xi',x,y)dt d\xi'+o(|x-y|^{-d/2}).
\end{split}
\end{equation}
By changing the variables as before,
\begin{equation}
\label{E:J2_2}
\begin{split}
&x_0^{d/2}\int_{|\xi'|<\epsilon}\int_{-\infty}^{x_0/2}\xi_l e^{(-t+x_0)A_s(i\xi')}\nu_s^{(l)}(t,\xi',x,y)dtd\xi'\\
=&\int_{-\infty}^{x_0/2}\left(1-\frac{t}{x_0}\right)^{-d/2}\int_{|x'|<\epsilon \sqrt{x_0-t}}x'_l\exp{\left(-\frac{1}{2}x'\cdot Q_s x'+O\left(\frac{|x'|^3}{\sqrt{x_0-t}}\right)\right)}\\
&\hspace{7cm} \times \nu_s^{(l)}\left(t,\left(\frac{x'}{\sqrt{x_0-t}}\right),x,y\right)dx'dt.
\end{split}
\end{equation}
In a similar manner, we obtain
\begin{equation*}
\begin{split}
\int_{-\infty}^{x_0/2}\left(1-\frac{t}{x_0}\right)^{-d/2}\int_{|x'|<\epsilon\sqrt{x_0-t}}\bigg|x'_l&\exp{\left(-\frac{1}{2}x'\cdot Q_s x'+O\left(\frac{|x'|^3}{\sqrt{x_0-t}}\right)\right)}\bigg|\\
& \quad \times \left|\nu_s^{(l)}\left(t,\left(\frac{x'}{\sqrt{x_0-t}}\right),x,y\right)\right|dx'dt\\
& \leq 2^{d/2}S\int_{\mathbb{R}^{d-1}}|x'|\exp{\left(-\frac{1}{4}\gamma_{\omega} |x'|^2\right)}dx'<\infty.
\end{split}
\end{equation*}
This final estimate and \mref{E:J2_1}-\mref{E:J2_2} imply $J_l=O(x_0^{-d/2})$.
\end{proof}
%%%%%%%%%%%%%%%%%%%%%%%%%%%%%
\subsection{Estimates of the remainder term $r(x,y)$}
$ $
\vspace{0.3cm}

To complete the proof of Theorem \ref{main_local}, we only need to prove that the remainder term $r(x,y)$ is $o(|x-y|^{-d/2+\varepsilon})$ for any $\varepsilon>0$ when $|x-y| \rightarrow \infty$. The rapidly decaying terms $I_2$ and $J$ can be ignored in the asymptotics of the remainder $r$, according to Proposition \ref{remainderJ}. Now let $c$ be the coefficient in front of $|x-y|^{-(d-1)/2}$ in $I_1$ (see \mref{E:leadingterm}). Namely, $\displaystyle c:=\lim_{x_0 \rightarrow \infty}x_0^{(d-1)/2}I_1$, where $x_0=|x-y|$. Hence, it is sufficient to show
\begin{equation}
\label{E:I1_2}
\lim_{x_0 \rightarrow \infty}x_0^{d/2-\varepsilon}(I_1-cx_0^{-(d-1)/2})=0, \quad 0<\varepsilon<1/2.
\end{equation}
The idea of the following lemma is to truncate some unnecessary (rapidly decreasing) parts of the main integrals we are interested in.
\begin{lemma}
\label{L:truncation_I1}

(i) For any $\alpha \in (0,1)$ and $n>0$, one has
\begin{equation*}
\begin{split}
\int_{x_0^{\alpha}}^{x_0/2}\frac{1}{(1-\frac{t}{x_0})^{(d-1)/2}}
\int_{|x'|<\epsilon\sqrt{(x_0-t)}}&\exp{\left(-\frac{1}{2}x'\cdot Q_s x'+O\left(\frac{|x'|^3}{\sqrt{x_0-t}}\right)\right)}\\&
\times\left|\nu_s\left(t,\frac{x'}{\sqrt{x_0-t}}\right)\right|dx' dt=O(x_0^{-n}).
\end{split}
\end{equation*}

(ii) Let $\beta \in (0,1/2)$, then
\begin{equation*}
\begin{split}
\int_{-\infty}^{x_0^{\alpha}}
\int_{x_0^{\beta}\leq |x'|<\epsilon\sqrt{(x_0-t)}}&\exp{\left(-\frac{1}{2}x'\cdot Q_s x'+O\left(\frac{|x'|^3}{\sqrt{x_0-t}}\right)\right)}\\&
\times\left|\nu_s\left(t,\frac{x'}{\sqrt{x_0-t}}\right)\right|dx' dt=O(x_0^{-n}),
\end{split}
\end{equation*}
for any $n>0$.
\end{lemma}

\begin{proof}
Both parts are treated similarly to the estimates for $I_2$.

(i) The integral is bounded from above (up to some constant factors) by
\begin{equation*}
\begin{split}
&\int_{x_0^{\alpha}}^{x_0/2}
\int_{|x'|<\epsilon\sqrt{(x_0-t)}}\exp{\left(-\frac{1}{4}x'\cdot Q_s x'\right)}O(t^{-n/\alpha-1})dx' dt \\
&\lesssim \int_{x_0^{\alpha}}^{x_0/2}t^{-n/ \alpha-1}dt \lesssim x_0^{-n}.
\end{split}
\end{equation*}
(ii) We recall from \mref{E:exponential_est} that on $D'(0,\epsilon \sqrt{x_0-t})$, the exponential term is bounded by the integrable (on $\mathbb{R}^{d-1}$) function $$\exp{\left(-\frac{1}{4}\gamma_{\omega} |x'|^2\right)},$$ and the function $$\displaystyle \sup_{s \in \mathcal{V}_{\omega}, |\xi'|\leq \epsilon}|\nu_s(\cdot,\xi')|$$ is integrable over $(-\infty, \infty)$. Hence, the integral in this lemma is estimated from above by:
\begin{equation*}
\begin{split}
&\int_{-\infty}^{x_0^{\alpha}}
\int_{x_0^{\beta}\leq |x'|<\epsilon\sqrt{x_0}}\exp{\left(-\frac{1}{4}\gamma_{\omega}|x'|^2\right)}
\left|\nu_s\left(t,\frac{x'}{\sqrt{x_0-t}}\right)\right|dx' dt \\
\leq &\int_{-\infty}^{\infty}\sup_{s \in \mathcal{V}_{\omega}, |\xi'|\leq \epsilon}\left|\nu_s(t,\xi')\right|dt\int_{|x'| \geq x_0^{\beta}}\exp{\left(-\frac{1}{4}\gamma_{\omega}|x'|^2\right)}dx'\\
\lesssim &\int_{x_0^{\beta}}^{\infty}\exp{\left(-\frac{1}{4}\gamma_{\omega} r^2\right)}r^{d-2}dr  \\
\lesssim & \int_{x_0^{\beta}}^{\infty}r^{-n/\beta-1-(d-2)}r^{d-2}dr=O(x_0^{-n}).
\end{split}
\end{equation*}
\end{proof}
Fixing a small $\alpha$, we let
\begin{equation*}
\begin{split}
I'_1=\frac{1}{(2\pi)^d}\int_{-\infty}^{x_0^{\alpha}}\frac{1}{(1-\frac{t}{x_0})^{(d-1)/2}}
\int_{|x'|<\epsilon\sqrt{(x_0-t)}}&\exp{\left(-\frac{1}{2}x'\cdot Q_s x'+O\left(\frac{|x'|^3}{\sqrt{x_0-t}}\right)\right)}\\&
\times\left|\nu_s\left(t,\frac{x'}{\sqrt{x_0-t}}\right)\right|dx' dt.
\end{split}
\end{equation*}
Part (i) of Lemma \ref{L:truncation_I1} yields that $I_1=x_0^{-(d-1)/2} I'_1+O(x_0^{-n})$. Thus, we can deal with the truncated integral $x_0^{-(d-1)/2}I'_1$ instead of working directly with $I_1$ in proving \mref{E:I1_2}. So our goal is to estimate the remainder $(I'_1-c)=O(x_0^{\varepsilon-1/2})$. We represent $(I'_1-c)$ as follows:
\begin{equation}
\label{E:R}
\begin{split}
\int_{-\infty}^{\infty}\int_{\mathbb{R}^{d-1}}\exp&{\left(-\frac{1}{2}x'\cdot Q_s x'\right)}
\bigg(\chi_{(-\infty,x_0^{\alpha}]\times D'(0,\epsilon \sqrt{x_0-t})}(t,x')\big(1-\frac{t}{x_0}\big)^{-(d-1)/2}\\
&\times \exp{\left(O\left(\frac{|x'|^3}{\sqrt{x_0-t}}\right)\right)}\nu_s\left(t,\frac{x'}{\sqrt{x_0-t}}\right)-\nu_s(t,0) \bigg) dx' dt.
\end{split}
\end{equation}
The first step is to replace the factor $\displaystyle \big(1-\frac{t}{x_0}\big)^{-(d-1)/2}$ in the integrand by $1$. More explicitly,
\begin{lemma}
If $0<\alpha<\varepsilon$, we obtain
\begin{equation*}
\begin{split}
\int_{-\infty}^{x_0^{\alpha}}\int_{D'(0,\epsilon \sqrt{x_0-t})}\exp&{\left(-\frac{1}{2}x'\cdot Q_s x'+O\left(\frac{|x'|^3}{\sqrt{x_0-t}}\right)\right)}\nu_s\left(t,\frac{x'}{\sqrt{x_0-t}}\right)
\\& \times
\bigg(\big(1-\frac{t}{x_0}\big)^{-(d-1)/2}-1\bigg) dx' dt=o(x_0^{\varepsilon-1/2}).
\end{split}
\end{equation*}
\end{lemma}

\begin{proof}
As we argued in the proof of Lemma \ref{L:truncation_I1} (ii), this integral is majorized by
\begin{equation*}
\begin{split}
\int_{-\infty}^{x_0^{\alpha}}\int_{D'(0,\epsilon \sqrt{x_0-t})}\exp&{\left(-\frac{1}{4}x'\cdot Q_s x'\right)}\left|\nu_s\left(t,\frac{x'}{\sqrt{x_0-t}}\right)\right|\bigg(\big(1-\frac{t}{x_0}\big)^{-(d-1)/2}-1\bigg) dx' dt.
\end{split}
\end{equation*}

We break the external integral into two parts, where the first one is over $(-\infty, -x_0^{\alpha})$ and the second one is over $[-x_0^\alpha, x_0^{\alpha}]$.
Note that the factor $\displaystyle \big(1-\frac{t}{x_0}\big)^{-(d-1)/2}$ is less than $1$ if $t<-x_0^{\alpha}$.
Hence, the rapid decay of $\nu_s$ on $(-\infty, -x_0^{\alpha})$ (as $x_0$ is large enough) implies the same property for the first part. Since $\nu_s$ is uniformly bounded on $\mathbb{R} \times D'(0,\epsilon)$, it suffices to estimate the factor $\displaystyle \big(1-\frac{t}{x_0}\big)^{-(d-1)/2}-1$ for the second part. But this is straightforward, since
\begin{equation*}
\begin{split}
\int_{-x_0^{\alpha}}^{x_0^{\alpha}}\bigg(\big(1-\frac{t}{x_0}\big)^{-(d-1)/2}-1\bigg) &\leq 2x_0^{\alpha}\bigg(\big(1-x_0^{\alpha-1}\big)^{-(d-1)/2}-1\bigg)\\
&=2x_0^{\alpha}\frac{1-(1-x_0^{\alpha-1})^{(d-1)/2}}{(1-x_0^{\alpha-1})^{(d-1)/2}}=O(x_0^{2\alpha-1}).
\end{split}
\end{equation*}
Now the claim follows from the inequality $2\alpha-1<\varepsilon-1/2$.
\end{proof}

This lemma reduces estimates the expression in \mref{E:R} to estimating
\begin{equation*}
\begin{split}
R=\int_{-\infty}^{\infty}\int_{\mathbb{R}^{d-1}}\exp{\bigg(-\frac{1}{2}x'\cdot Q_s x'\bigg)}&
\bigg(\chi_{(-\infty,x_0^{\alpha}]\times D'(0,\epsilon \sqrt{x_0-t})}(t,x')\exp{\left(O\left(\frac{|x'|^3}{\sqrt{x_0-t}}\right)\right)}\\
&\times\nu_s\left(t,\frac{x'}{\sqrt{x_0-t}}\right)-\nu_s(t,0) \bigg) dx' dt.
\end{split}
\end{equation*}
Thanks to part (ii) of Lemma \ref{L:truncation_I1}, it is enough to deal with the following truncation of $R$, which can also be decomposed as the sum of $R^{(1)}$ and  $R^{(2)}$
\begin{equation*}
\begin{split}
\int_{-\infty}^{\infty}\int_{\mathbb{R}^{d-1}}\exp{\bigg(-\frac{1}{2}x'\cdot Q_s x'\bigg)}&
\bigg(\chi_{(-\infty,x_0^{\alpha}]\times D'(0, x_0^{\beta})}(t,x')\exp{\left(O\left(\frac{|x'|^3}{\sqrt{x_0-t}}\right)\right)}\\
&\times\nu_s\left(t,\frac{x'}{\sqrt{x_0-t}}\right)-\nu_s(t,0) \bigg) dx' dt=R^{(1)}+R^{(2)},
\end{split}
\end{equation*}
where
\begin{equation*}
\begin{split}
R^{(1)}=\int_{-\infty}^{x_0^{\alpha}}\int_{|x'|<x_0^{\beta}}&\exp{\bigg(-\frac{1}{2}x'\cdot Q_s x'
+O\bigg(\frac{|x'|^3}{\sqrt{x_0-t}}\bigg)\bigg)}\\
&\hspace{4cm}\times \bigg(\nu_s\bigg(t,\frac{x'}{\sqrt{x_0-t}}\bigg)-\nu_s(t,0) \bigg) dx' dt
\end{split}
\end{equation*}
and
\begin{equation*}
\begin{split}
R^{(2)}=\int_{-\infty}^{\infty}\int_{\mathbb{R}^{d-1}}\exp{\bigg(-\frac{1}{2}x'\cdot Q_s x'\bigg)}\nu_s(t,0)&
\bigg(\chi_{(-\infty,x_0^{\alpha}]\times D'(0, x_0^{\beta})}(t,x')\\
&\times \exp{\left(O\left(\frac{|x'|^3}{\sqrt{x_0-t}}\right)\right)-1}
\bigg) dx' dt.
\end{split}
\end{equation*}
Consider $R^{(1)}$ first. Since $\displaystyle \sup_{(s,t) \in \mathcal{V}_{\omega} \times \mathbb{R}, |x'| \leq \epsilon}|\nabla_{x'}\nu_s(t,x')|<\infty$, the difference factor $\displaystyle \bigg|\nu_s\left(t,\frac{x'}{\sqrt{x_0-t}}\right)-\nu_s(t,0) \bigg|$ is dominated by $O(|x'|/\sqrt{(x_0-t)})=O(x_0^{\beta-1/2})$. For any $n>0$ large enough, $R^{(1)}$ is majorized by
\begin{equation*}
\begin{split}
|R^{(1)}| &\lesssim \int_{|x'|<x_0^{\beta}}\exp{\bigg(-\frac{1}{4}x'\cdot Q_s x'}\bigg)dx' \left(\int_{-x_0^{\alpha}}^{x_0^{\alpha}}x_0^{\beta-1/2}dt+\int_{-\infty}^{-x_0^{\alpha}}|t|^{-n}dt\right)\\
& \lesssim x_0^{\alpha+\beta-1/2}.
\end{split}
\end{equation*}

Now consider $R^{(2)}$. For any $\beta<1/6$, we could first estimate the below difference factor of the integrand
\begin{equation*}
\begin{split}
&\quad \left|\chi_{(-\infty,x_0^{\alpha}]\times D'(0, x_0^{\beta})}(t,x')\exp{\left(O\left(\frac{|x'|^3}{\sqrt{x_0-t}}\right)\right)}-1\right| \\
& \leq \chi_{(-\infty,x_0^{\alpha}]\times D'(0, x_0^{\beta})}(t,x')\left(\exp{\left(O(x_0^{3\beta-1/2})\right)}-1\right)+(1-\chi_{(-\infty,x_0^{\alpha}]\times D'(0, x_0^{\beta})}(t,x'))\\
& \leq \chi_{(-\infty,x_0^{\alpha}]\times D'(0, x_0^{\beta})}(t,x')O(x_0^{3\beta-1/2})+\chi_{(-\infty, \infty) \times \{|x'|\geq x_0^{\beta}\}}(t,x')+\chi_{(x_0^{\alpha}, \infty) \times D'(0, x_0^{\beta})}(t,x').
\end{split}
\end{equation*}
The above inequality leads to the following estimate for $R^{(2)}$:
\begin{equation*}
\begin{split}
|R^{(2)}| &\lesssim x_0^{3\beta-1/2}\left(\int_{-\infty}^{x_0^{\alpha}}|\nu_s(t,0)|dt\right)+\int_{x_0^{\alpha}}^{\infty}|\nu_s(t,0)|dt+\int_{|x'|\geq x_0^{\beta}}\exp{\bigg(-\frac{1}{2}x'\cdot Q_s x'\bigg)}dx'\\
& \lesssim  x_0^{\alpha+3\beta-1/2}+\int_{|t|>x_0^{\alpha}}^{\infty}|t|^{-n}dt+\int_{x_0^{\beta}}^{\infty}r^{-n}dr=O(x_0^{\alpha+3\beta-1/2}).
\end{split}
\end{equation*}
Hence, $R^{(1)}+R^{(2)}=O(x_0^{\alpha+3\beta-1/2})$. By choosing $\beta\in (0,(\varepsilon-\alpha)/3)$, we conclude that $R=o(x_0^{\varepsilon-1/2})$. This proves the claim \mref{E:I1_2}.

%%%%%%%%%%%%%%%%%%%%%%%%%%%%%%%%%%%%%%%

%%%%%%%%%%%%%%%%%%%%%%%%%%%%%%%%%%%%%%%
\section{The full Green's function asymptotics}
Here we show that the full Green's function $G_{s,\lambda}$ has the same asymptotics as the reduced Green's function $G_0$ as $|x-y| \rightarrow \infty$.

Let us fix a point $s$ in $\mathbb{S}^{d-1}$. As we discussed in Section \ref{sec:floquet}, we handle the Schwartz kernel $K_s(x,y)$ of the operator $T_s$ acting on $L^2(\mathbb{R}^d)$ as follows:
$$\mathcal{F}T_s\mathcal{F}^{-1}=\int_{\mathcal{O}}^{\oplus}T_s(k)d'k,$$
where $\mathcal{F}$ is the Floquet transform (see Definition \ref{Floquet transforms}) and
$$T_s(k)=(1-\eta(k))(L_s(k)-\lambda)^{-1}+\eta(k)((L_{s}(k)-\lambda)|_{R(Q_s(k))})^{-1}Q_s(k).$$
We would like to show that $K_s(x,y)$ is continuous away from the diagonal and decays sufficiently fast as $|x-y| \rightarrow \infty$ uniformly with respect to $s$.

Observe that the kernel of each projector $P_s(k)$ (see Subsection 5.3) is the smooth function:
$$\displaystyle \frac{\phi(k+i\beta_s,x)\overline{\phi(k-i\beta_s,y)}}{F(k+i\beta_s)},$$ for each $k$ in the support of $\eta$. Thus, $(1-\eta(k))P_s(k)$ is a finite rank smoothing operator on $\mathbb{T}$. Moreover, we also have $(L_s(k)-\lambda)T_s(k)=T_s(k)(L_s(k)-\lambda)=I-\eta(k)P_s(k)$.
Each $T_s(k)$ is a parametrix (i.e., an inverse modulo a smoothing operator) of the elliptic operator $L_s(k)-\lambda$ when $(s,k) \in \mathbb{S}^{d-1} \times \mathcal{O}$. This suggests to study parametrices of the family of elliptic operators $L_s(k)-\lambda$ simultaneously.
%%%%%%%%%%%%%%%%%%%%%%%%%%%%%%%%%
\subsection{Parameter-dependent periodic pseudodifferential operators}
$ $
\vspace{0.3cm}

First, we briefly recall some basic definitions of \textbf{periodic} (or \textbf{toroidal}) pseudodifferential operators (i.e., $\Psi$DO on the torus $\mathbb{T}$). We also introduce some useful classes of symbols with parameters and describe some of their properties that we will use.
%Then we will prove the existence of a family of parametrics corresponding to the family of elliptic operators $\{L_s(k)-\lambda\}_{(s,k) \in \mathbb{S}^{d-1} \times \mathcal{O}}$.

There are several approaches to defining pseudodifferential operators on the torus. The standard approach based on H\"ormander's symbol classes (see e.g., \cite{Shubin_pseudo}) uses local smooth structure on the torus $\mathbb{T}$ and thus ignores the group structure on $\mathbb{T}$.
An alternative approach uses Fourier series with the difference calculus and avoids using local coordinate charts on $\mathbb{T}$ (the details in Chapter 4 in \cite{Ruzh-Turu})\footnote{A different approach to periodic $\Psi DO$s is introduced by A. Sobolev \cite{Sob}.}.
To make a distinction, Ruzhansky and Turunen in \cite{Ruzh-Turu} refer to the symbols in the first approach as \textbf{Euclidean symbols} and the symbols in the latter one as \textbf{toroidal symbols} (see Section 4.5 in \cite{Ruzh-Turu}).
We recall their definitions for only the Kohn-Nirenberg symbol classes, which we need here:
\begin{defi}
\label{Euclidean_symbols}
Let $m$ be a real number.

(a) The class $S^{m}(\mathbb{T} \times \mathbb{R}^d)$ consists of all smooth functions $\sigma(x,\xi)$ on $\mathbb{T} \times \mathbb{R}^d$ such that for any multi-indices $\alpha, \beta$,
$$|D^{\alpha}_{\xi}D^{\beta}_{x}\sigma(x,\xi)| \leq C_{\alpha \beta}(1+|\xi|)^{m-|\alpha|},$$
for some constant $C_{\alpha, \beta}$ that depends only on $\alpha, \beta$.
Symbols in $S^m(\mathbb{T} \times \mathbb{R}^d)$ are called \textbf{Euclidean symbols} of order $m$ on $\mathbb{T}$.

(b)  The class $S^{m}(\mathbb{T} \times \mathbb{Z}^d)$ consists of all functions $\sigma(x,\xi)$ on $\mathbb{T} \times \mathbb{Z}^d$ such that for each $\xi \in \mathbb{Z}^d$, $\sigma(.,\xi) \in C^{\infty}(\mathbb{T})$ and for any multi-indices $\alpha, \beta$,
$$|\Delta^{\alpha}_{\xi}D^{\beta}_{x}\sigma(x,\xi)| \leq C_{\alpha \beta}(1+|\xi|)^{m-|\alpha|},$$
for some constant $C_{\alpha, \beta}$ that depends only on $\alpha, \beta$. Here we recall the definition of the forward difference operator $\Delta^{\alpha}_{\xi}$  with respect to the variable $\xi$ \cite{Ruzh-Turu}. Let $f$ be a complex-valued function defined on $\displaystyle \mathbb{Z}^d$ and $1 \leq j \leq d$. Then we define
$$\Delta_{j}f(\xi):=f(\xi_1,\dots,\xi_{j-1},\xi_j+1,\xi_{j+1},\dots,\xi_d)-f(\xi),$$
and for any multi-index $\alpha$,
$$\Delta^{\alpha}_{\xi}:=\Delta^{\alpha_1}_{1}\dots\Delta^{\alpha_d}_{d}.$$
Symbols in $S^{m}(\mathbb{T} \times \mathbb{Z}^d)$ are called \textbf{toroidal symbols} of order $m$ on $\mathbb{T}$.

(c) The intersection of all the classes $S^{m}(\mathbb{T} \times \mathbb{R}^d)$ ($S^{m}(\mathbb{T} \times \mathbb{Z}^d)$) is denoted by $S^{-\infty}(\mathbb{T} \times \mathbb{R}^d)$ ($S^{-\infty}(\mathbb{T} \times \mathbb{Z}^d)$), which are also called \textbf{smoothing symbols}.
\end{defi}
Due to Theorem 4.5.3 in \cite{Ruzh-Turu}, a symbol is toroidal of order $m$ if and only if it could be extended in $\xi$ to an Euclidean symbol of the same order $m$. Such an extension is unique modulo a smoothing symbol. Consequently, we will use the notation $S^m(\mathbb{T})$ for both classes $S^m(\mathbb{T} \times \mathbb{R}^d)$ and $S^m(\mathbb{T} \times \mathbb{Z}^d)$. The two approaches are essentially equivalent in defining pseudodifferential operators on $\mathbb{T}$ whenever the symbol is in the class $S^m(\mathbb{T})$.
Following \cite{GUI}, this motivates us to define periodic pseudodifferential operators as follows:
\begin{defi}
\label{periodic_pseudo}
Given a symbol $\sigma(x,\xi) \in S^m(\mathbb{T})$, we denote by $Op(\sigma)$ the corresponding periodic pseudodifferential operator defined by
\begin{equation}
\label{E:def_pseudo}
\left(Op(\sigma)f\right)(x):=\sum_{\xi \in \mathbb{Z}^d}\sigma(x,\xi)\tilde{f}(\xi)e^{2\pi i \xi \cdot x},
\end{equation}
where $\tilde{f}(\xi)$ is the Fourier coefficient of $f$ at $\xi$. The right hand side of \mref{E:def_pseudo} converges absolutely if, for instance, $f \in C^{\infty}(\mathbb{T})$.

We also use the notation $Op(S^m(\mathbb{T}))$ for the set of all periodic pseudodifferential operators $Op(\sigma)$ with $\sigma \in S^m(\mathbb{T})$.
\end{defi}

Since we must deal with parameters $s$ and $k$, we introduce a suitable class of symbols depending on parameters $(s,k) \in \mathbb{S}^{d-1} \times \mathcal{O}$.
\begin{defi}
\label{parameter_class_symbols}
The parameter-dependent class $\tilde{S}^{m}(\mathbb{T})$ consists of symbols $\sigma(s,k;x,\xi)$ satisfying the following conditions:
\begin{itemize}
\item For each $(s,k) \in \mathbb{S}^{d-1} \times \mathcal{O}$, the function $\sigma(s,k;\cdot,\cdot)$ is a symbol in the class $S^m(\mathbb{T})$.

\item Consider any multi-indices $\alpha, \beta, \gamma$. Then for each $s \in\mathbb{S}^{d-1}$, the function $\sigma(s,\cdot;\cdot,\cdot)$ is smooth on $\mathcal{O} \times \mathbb{T} \times \mathbb{R}^d$ , and furthermore,
\begin{equation*}
\label{E:parameter_class_symbols}
\sup_{s \in \mathbb{S}^{d-1}}|D^{\alpha}_{k}D^{\beta}_{\xi}D^{\gamma}_x \sigma(s,k;x,\xi)| \leq C_{\alpha \beta \gamma} (1+|\xi|)^{m-|\alpha|-|\beta|},
\end{equation*}
for some constant $C_{\alpha \beta \gamma}>0$ that is independent of $s,k,x,$ and $\xi$.
\end{itemize}
Thus, taking derivatives of a symbol in $\tilde{S}^{m}(\mathbb{T})$ with respect to $k$ improves decay in $\xi$.
We also denote
$$\tilde{S}^{-\infty}(\mathbb{T}):=\bigcap_{m \in \mathbb{R}}\tilde{S}^{m}(\mathbb{T}).$$
\end{defi}

\begin{defi}
For each $m \in \mathbb{R} \cup \{-\infty\}$, we denote by $Op(\tilde{S}^m(\mathbb{T}))$ the set of all families of periodic pseudodifferential operators $\{Op(\sigma(s,k;\cdot,\cdot))\}_{(s,k) \in \mathbb{S}^{d-1} \times \mathcal{O}}$, where $\sigma$ runs over the class $\tilde{S}^m(\mathbb{T})$.
\end{defi}

\begin{example}
$ $

\begin{itemize}
\item
Suppose that $|\lambda|$ is small enough so that $\max_{s \in \mathbb{S}^{d-1}} |\beta_s|<1$. Then the family of symbols $\{(1+(\xi+k+i\beta_s)^2)^{m/2}\}_{(s,k)}$ belongs to the class $\tilde{S}^m(\mathbb{T})$ for any $m \in \mathbb{R}$.

\item
If $a_{\alpha}(x) \in C^{\infty}(\mathbb{T})$ and $m \geq 0$, then
$$\{\sum_{|\alpha| \leq m} a_{\alpha}(x)(\xi+k+i\beta_s)^{\alpha}\}_{(s,k)} \in \tilde{S}^m(\mathbb{T}).$$

\item
The family of elliptic operators $\{(L_s(k)-\lambda)\}_{(s,k)}$ is in $Op(\tilde{S}^{2}(\mathbb{T}))$.

\item
If $a=\{a(s,k;x,\xi)\}_{(s,k)} \in \tilde{S}^l(\mathbb{T})$ and $b=\{b(s,k;x,\xi)\}_{(s,k)} \in \tilde{S}^m(\mathbb{T})$ then $ab=\{ab(s,k;x,\xi)\}_{(s,k)} \in \tilde{S}^{l+m}(\mathbb{T})$.

\item
$a(s,k;x,\xi) \in \tilde{S}^l(\mathbb{T})$ implies $D^{\alpha}_k D^{\beta}_{\xi} D^{\gamma}_x a(s,k;x,\xi) \in \tilde{S}^{l-|\alpha|-|\beta|}(\mathbb{T})$.
\end{itemize}
\end{example}

The following result will be needed in the next subsection:
\begin{thm}
\label{parametrices}
There exists a family of parametrices $\{A_{s}(k)\}_{(s,k)}$ in the class $Op(\tilde{S}^{-2}(\mathbb{T}))$ for the family of elliptic operators $\{(L_s(k)-\lambda)\}_{(s,k)}$.
\end{thm}
The reader can refer to Section 8 for the proof of this result as well as some other basic properties of parameter-dependent toroidal $\Psi$DOs.
%%%%%%%%%%%%%%%%%%%%%%%%%%%%%%%%%%%%%%%%%%%%%%%%%%%%%%%%%
\subsection{Decay of the Schwartz kernel of $T_s$}
\begin{lemma}
\label{L:identity_symbol}
For all $k$ on a sufficiently small neighborhood of the support of $\eta$, $\lambda$ $(<0)$ is in the resolvent of the operator $L_s(k)Q_s(k)$ acting on $L^2(\mathbb{T})$. Furthermore, for such $k$, we have the following identity:
\begin{equation}
\label{identity_symbol}
((L_s(k)-\lambda)|_{R(Q_s(k))})^{-1}Q_s(k)=\lambda^{-1}P_s(k)+(L_s(k)Q_s(k)-\lambda)^{-1}.
\end{equation}
\begin{proof}
In the block-matrix form, $(L_s(k)Q_s(k)-\lambda)$ is
\begin{equation}
\label{E:block-matrix}
 \left(
    \begin{array}{r@{}c|c@{}l}
  &    \begin{matrix}
        -\lambda P_s(k)
      \end{matrix} & \mbox{0} & \\\hline
  &    \mbox{0} &
       \begin{matrix}
       (L_s(k)-\lambda)|_{R(Q_s(k))}
      \end{matrix}
    \end{array}
\right).
\end{equation}
This gives the first claim of this lemma. The inverse of \mref{E:block-matrix} is
\begin{equation*}
 \left(
    \begin{array}{r@{}c|c@{}l}
  &    \begin{matrix}
         -\lambda^{-1} P_s(k)
      \end{matrix} & \mbox{0} & \\\hline
  &    \mbox{0} &
       \begin{matrix}
       ((L_s(k)-\lambda)|_{R(Q_s(k))})^{-1}
      \end{matrix}
    \end{array}
\right),
\end{equation*}
which proves the identity \mref{identity_symbol}.
\end{proof}
\end{lemma}
The identity \mref{identity_symbol} implies that for each $(s,k)$, the operator $$\eta(k)((L_s(k)-\lambda)|_{R(Q_s(k))})^{-1}Q_s(k)$$ is a periodic pseudodifferential operator in $S^{-2}(\mathbb{T})$. Thus, each of the operators $T_s(k)$ is also in $S^{-2}(\mathbb{T})$ and its symbol is smooth in $(s,k)$ since $P_s(k)$ and $Q_s(k)$ are smooth in $(s,k)$. Actually, more information about the family of operators $\{T_s(k)\}_{(s,k)}$ and their Schwartz kernels can be obtained.

At first, we want to introduce a class of family of operators whose kernels behave nicely.
\begin{defi}
\label{class_S}
We denote by $\mathcal{S}$ the set consisting of families of smoothing operators $\{U_{s}(k)\}_{(s,k)}$ acting on $\mathbb{T}$ so that the following properties hold:
\begin{itemize}
\item For any $ m_1, m_2 \in \mathbb{R}$, the operator $U_s(k)$ is smooth in $k$ as a $\displaystyle B(H^{m_1}(\mathbb{T}), H^{m_2}(\mathbb{T}))$-valued function\footnote{We remind the reader that $B(E,F)$ denotes the space of all bounded linear operators from the Banach space $E$ to $F$.}.

\item The following uniform condition holds for any multi-index $\alpha$:
$$\sup_{s,k}\|D^{\alpha}_k U_s(k)\|_{B(H^{m_1}(\mathbb{T}), H^{m_2}(\mathbb{T}))}<\infty.$$
\end{itemize}
%(iii) For any multi-index $\alpha$ such that $|\alpha|>d-2$, we have $D^{\alpha}_k K_{U_{s}(k)}(x,y)$ is continuous on $\mathbb{T} \times \mathbb{T}$. Moreover, the following uniform condition also holds:
%$$\sup_{s,k}\|D^{\alpha}_k K_{U_s(k)}(\cdot,\cdot)\|_{C(\mathbb{T} \times \mathbb{T})}<\infty.$$
\end{defi}

We remark that if the family of smoothing operators $\{U_s(k)\}_{(s,k)}$ is in $Op(\tilde{S}^{-\infty}(\mathbb{T}))$, then this family also belongs to $\mathcal{S}$.

In order to obtain information on Schwartz kernels of a family of operators in $\mathcal{S}$, we need to use the following standard lemma on Schwartz kernels of integral operators acting on $\mathbb{T}$.
\begin{lemma}
\label{L:kernel_estimate_agmon}
Let $A$ be a bounded operator in $L^2(\mathbb{T})$. Suppose that the range of $A$ is contained in $H^m(\mathbb{T})$, where $m>d/2$ and in addition,
$$\|Af\|_{H^m(\mathbb{T})} \leq C\|f\|_{H^{-m}(\mathbb{T})}$$
for all $f \in L^2(\mathbb{T})$.

Then $A$ is an integral operator whose kernel $K_A(x,y)$ is bounded and uniformly continuous on $\mathbb{T} \times \mathbb{T}$ and the following estimate holds:
\begin{equation}
\label{E:kernel_estimate_agmon}
|K_A(x,y)| \leq \gamma_0 C,
\end{equation}
where $\gamma_0$ is a constant depending only on $d$ and $m$.
\end{lemma}
The fact can be found in Lemma 2.2 in \cite{Ag}.

Now we can state a useful property of Schwartz kernels of a family of operators in $\mathcal{S}$.
\begin{cor}
\label{kernel_in_S}
If $\{U_s(k)\}_{(s,k)}$ is a family of smoothing operators in $\mathcal{S}$, then the Schwartz kernel $K_{U_s}(k,x,y)$ of the operator $U_s(k)$ satisfies
$$\sup_{s,k,x,y}|D^{\alpha}_k K_{U_s}(k,x,y)|<\infty,$$
for any multi-index $\alpha$.
\end{cor}
\begin{proof}
We pick any $m>d/2$. Then by Definition \ref{class_S}, we have
$$\sup_{s,k}\|D^{\alpha}_k U_s(k)f\|_{H^m(\mathbb{T})} \leq C_{\alpha}\|f\|_{H^{-m}(\mathbb{T})}.$$
Applying Lemma \ref{L:kernel_estimate_agmon}, the estimates \mref{E:kernel_estimate_agmon} hold for kernels $D^{\alpha}_k K_{U_s}(k,x,y)$ of the operators $D^{\alpha}_k U_s(k)$ uniformly in $(s,k)$.
\end{proof}

We now go back to the family of operators $T_s(k)$.
\begin{prop}
\label{P:pseudo_ops}
There is a family of periodic pseudodifferential operators $\{B_{s}(k)\}_{(s,k)}$ in $Op(\tilde{S}^{-2}(\mathbb{T}))$ such that the family of operators $\{T_s(k)-B_s(k)\}_{(s,k)}$ belongs to $\mathcal{S}$.
\end{prop}

\begin{proof}
Due to Theorem \ref{parametrices}, there is a family of operators $\{A_s(k)\}_{(s,k)}$ in $Op(\tilde{S}^{-2}(\mathbb{T}))$ and a family of operators $\{R_s(k)\}_{(s,k)}$ in $Op(\tilde{S}^{-\infty}(\mathbb{T}))$ such that
$$(L_s(k)-\lambda)A_s(k)=I-R_s(k).$$
Since $T_s(k)(L_s(k)-\lambda)=I-\eta(k)P_s(k)$, we deduce that
\begin{equation}
\label{E:eqn_T(k)}
T_s(k)=A_s(k)-\eta(k)P_s(k)A_s(k)+T_s(k)R_s(k).
\end{equation}

One can check that the symbols of operators $\eta(k)P_s(k)A_s(k)$ are in the class $\tilde{S}^{-\infty}(\mathbb{T})$ due to the composition formula. Thus if we let $B_s(k):=A_s(k)-\eta(k)P_s(k)A_s(k)$, then $\{B_s(k)\}_{(s,k)} \in Op(\tilde{S}^{-2}(\mathbb{T}))$.

Now it remains to show that the family of smoothing operators $\{T_s(k)R_s(k)\}_{(s,k)}$ is in $\mathcal{S}$. Let us fix any two real numbers $m_1$, $m_2$ and a multi-index $\alpha$. Notice that $(L_s(k)-\lambda)$ is analytic in $k$ as a $\displaystyle B(H^{m_2}(\mathbb{T}), H^{m_2-2}(\mathbb{T}))$-valued function and also,
$$\sup_{s,k}\|D^{\alpha}_k (L_s(k)-\lambda)\|_{B(H^{m_2}(\mathbb{T}), H^{m_2-2}(\mathbb{T}))}<\infty.$$
Due to Lemma \ref{L:identity_symbol},
$$T_s(k)=(1-\eta(k))(L_s(k)-\lambda)^{-1}+\eta(k)\lambda^{-1}P_s(k)+\eta(k)(L_s(k)Q_s(k)-\lambda)^{-1}.$$
Thus, $T_s(k)$ is smooth in $k$ as a $B(H^{m_2-2}(\mathbb{T}), H^{m_2}(\mathbb{T}))$-valued function and moreover,
\begin{equation}
\label{uniform_Tsk}
\sup_{s,k}\|D^{\alpha}_k T_s(k)\|_{B(H^{m_2-2}(\mathbb{T}), H^{m_2}(\mathbb{T}))}<\infty.
\end{equation}
Since $\{R_{s}(k)\}$ is in $Op(\tilde{S}^{-\infty}(\mathbb{T}))$, $R_s(k)$ is smooth in $k$ as a $B(H^{m_1}(\mathbb{T}), H^{m_2-2}(\mathbb{T}))$-valued function and furthermore,
\begin{equation}
\label{uniform_Rsk}
\sup_{s,k}\|D^{\alpha}_k R_s(k)\|_{B(H^{m_1}(\mathbb{T}), H^{m_2-2}(\mathbb{T}))}<\infty.
\end{equation}
By \mref{uniform_Tsk}, \mref{uniform_Rsk} and Leibnitz's rule, we deduce that $T_s(k)R_s(k)$ is smooth in $k$ as a  $B(H^{m_1}(\mathbb{T}), H^{m_2}(\mathbb{T}))$-valued function and the corresponding uniform estimate also holds. We conclude that the family $\{T_s(k)-B_s(k)\}_{(s,k)}$ is in $\mathcal{S}$.
\end{proof}

We need the following important estimate of Schwartz kernels of operators $T_s(k)$:
\begin{cor}
\label{kernel_estimate}
Let $K_s(k,x,y)$ be the Schwartz kernel of the operator $T_s(k)$. Let $N>d-2$. If $\alpha$ is a multi-index such that $|\alpha|=N$, then each $D^{\alpha}_k K_s(k,x,y)$ is a continuous function on $\mathbb{T} \times \mathbb{T}$ and the following estimate also holds uniformly with respect to $(x,y)$:
\begin{equation*}
\sup_{(s,k) \in \mathbb{S}^{d-1} \times \mathcal{O}}|D_{k}^{\alpha}K_s(k,x,y)|< \infty.
\end{equation*}
\end{cor}
\begin{proof}
Due to Proposition \ref{P:pseudo_ops}, the operator $T_s(k)$ is a sum of operators $B_s(k)$ and $U_s(k)$ such that $\{B_s(k)\}_{(s,k)} \in Op(\tilde{S}^{-2}(\mathbb{T}))$ and $\{U_s(k)\}_{(s,k)} \in \mathcal{S}$. In particular,
$$K_s(k,x,y)=K_{B_s}(k,x,y)+K_{U_s}(k,x,y).$$

Recall that in the distributional sense, the Schwartz kernel $K_{B_s}(k,x,y)$ of the periodic pseudodifferential operator $B_s(k)$ is given by
$$\sum_{\xi \in \mathbb{Z}^d}\sigma(s,k;x,\xi)e^{2\pi i\xi\cdot (x-y)},$$
where $\sigma(s,k;x,\xi)$ is the symbol of the operator $B_s(k)$.

Since $\{\sigma(s,k;x,\xi)\}_{(s,k)}$ is in $\tilde{S}^{-2}(\mathbb{T})$,
$$|e^{2\pi i\xi\cdot (x-y)}D^{\alpha}_k \sigma(s,k;x,\xi)| \lesssim (1+|\xi|)^{-2-N}.$$
Since $-(2+N)<-d$, the sum
$$\sum_{\xi \in \mathbb{Z}^d} D^{\alpha}_k\sigma(s,k;x,\xi)e^{2\pi i\xi\cdot (x-y)}$$
converges absolutely and moreover,
\begin{equation*}
\sup_{(s,k,x,y) \in \mathbb{S}^{d-1} \times \mathcal{O} \times \mathbb{T} \times \mathbb{T}}|D_{k}^{\alpha}K_{B_s}(k,x,y)| \lesssim \sum_{\xi \in \mathbb{Z}^d}(1+|\xi|)^{-(d+1)}<\infty.
\end{equation*}
Combining this with Corollary \ref{kernel_in_S}, we complete the proof.
\end{proof}

\begin{notation}
Let $\psi$ be a function on $\mathbb{R}^d$ and $\gamma$ be a vector in $\mathbb{R}^d$, then $\tau_{\gamma}\psi$ is the $\gamma$-shifted version of $\psi$. Namely, it is defined as follows: $$\tau_{\gamma}\psi(\cdot)=\psi(\cdot+\gamma).$$

We denote by $\mathcal{P}$ the subset of $C^{\infty}_{0}(\mathbb{R}^d)$ consisting of all functions $\psi$ such that its support is connected, and if $\gamma$ is a \textbf{non-zero vector in $\mathbb{Z}^d$}, then the support of $\tau_{\gamma}\psi$ does not intersect with the support of $\psi$.
\end{notation}

\begin{defi}
Since $\mathbb{R}^d$ is the universal covering space of $\mathbb{T}$, we can consider the covering map $$\pi: \mathbb{R}^d \rightarrow \mathbb{R}^d/\mathbb{Z}^d=\mathbb{T}.$$
In particular, $\pi(x+\gamma)=\pi(x)$ for any $x \in \mathbb{R}^d$ and $\gamma \in \mathbb{Z}^d$.

A \textbf{standard fundamental domain} (with respect to the covering map $\pi$) is of the form $[0,1]^d+\gamma$ for some vector $\gamma$ in $\mathbb{R}^d$. Thus, a standard fundamental domain is a fundamental domain of $\mathbb{R}^d$ with respect to the lattice $\mathbb{Z}^d$.
\end{defi}

Using Definition \ref{Floquet transforms} of the Floquet transform $\mathcal{F}$, we can obtain the following formula:
\begin{lemma}
\label{kernel_formula}
Let $\phi$ and $\theta$ be any two smooth functions in $\mathcal{P}$. Then the Schwartz kernel $K_{s,\phi, \theta}$ of the operator $\phi T_s\theta$ satisfies the following identity for any $(x,y) \in \mathbb{R}^d \times \mathbb{R}^d$:
\begin{equation*}
K_{s,\phi,\theta}(x,y)=\frac{1}{(2\pi)^d}\int_{\mathcal{O}} e^{ik\cdot (x-y)}\phi(x)K_s(k,\pi(x),\pi(y))\theta(y)dk.
\end{equation*}
\end{lemma}
\begin{proof}
Since both $\phi, \theta \in \mathcal{P}$, there are standard fundamental domains $W_{\phi}$ and  $W_{\theta} \subset \mathbb{R}^d$ so that
$$\supp(\phi) \subset \mathring{W_{\phi}}, \quad \supp(\theta) \subset \mathring{W_{\theta}}.$$
Then, it suffices to show that $\langle \phi T_{s}\theta f,g \rangle$ equals
$$\frac{1}{(2\pi)^d}\int_{W_{\phi}}\int_{W_{\theta}}\int_{\mathcal{O}} e^{ik\cdot (x-y)}(\phi \overline{g})(x)K_s(k, \pi(x), \pi(y))(\theta f)(y)dkdydx,$$
for any $f,g$ in $C^{\infty}(\mathbb{R}^d)$.

We observe that
\begin{equation*}
\begin{split}
&\left\langle \phi T_{s}\theta f,g \right\rangle=\left\langle \mathcal{F}\phi T_{s}\theta f,\mathcal{F}g \right\rangle\\=&\frac{1}{(2\pi)^d}\left\langle (\mathcal{F}\phi \mathcal{F}^{-1}) \left(\int_{\mathcal{O}}^{\oplus}T_s(k)dk\right) \mathcal{F}(\theta f)
,\mathcal{F}g \right\rangle
\\=&\frac{1}{(2\pi)^d}\left\langle \left(\int_{\mathcal{O}}^{\oplus}T_s(k)dk\right) \mathcal{F}(\theta f)
, \mathcal{F}(\overline{\phi}g) \right\rangle.
\end{split}
\end{equation*}

Since $\theta \in \mathcal{P}$, for any $y$ in $W_{\theta}$, we have
$$\mathcal{F}(\theta f)(k,\pi(y))=(\theta f)(y)e^{-ik\cdot y}.$$

Similarly,
$$\mathcal{F}(\overline{\phi}g)(k,\pi(x))=(\overline{\phi}g)(x)e^{-ik\cdot x}, \quad \forall x \in W_{\phi}.$$

We also have
\begin{equation*}
%\begin{split}
\left(\int_{\mathcal{O}}^{\oplus}T_s(k)dk\right) (\mathcal{F}(\theta f))(k,\pi(x))
=T_s(k)(\mathcal{F}(\theta f)(k,\cdot))(\pi(x)).
%\end{split}
\end{equation*}

Consequently,
\begin{equation*}
\begin{split}
&\left\langle \left(\int_{\mathcal{O}}^{\oplus}T_s(k)dk\right) \mathcal{F}(\theta f)
, \mathcal{F}(\overline{\phi}g)\right\rangle\\
=&\int_{\mathcal{O}}\int_{W_{\phi}}T_s(k)(\mathcal{F}(\theta f)(k,\cdot))(\pi(x))\overline{(\overline{\phi} g)(x)e^{-ik\cdot x}}dxdk\\
=&\int_{\mathcal{O}}\int_{W_{\phi}}\int_{W_{\theta}}K_s(k,\pi(x),\pi(y))\mathcal{F}(\theta f)(k,\pi(y))(\phi \overline{g})(x)e^{ik\cdot x}dydxdk\\
=&\int_{\mathcal{O}}\int_{W_{\phi}}\int_{W_{\theta}}e^{ik\cdot (x-y)}K_s(k,\pi(x),\pi(y))(\theta f)(y)(\phi \overline{g})(x)dydxdk.
\end{split}
\end{equation*}
Using Fubini's theorem to rewrite the above integral, we have the desired identity.
\end{proof}

\begin{prop}
\label{cutoffkernel}
Consider any two smooth compactly supported functions $\phi$ and $\theta$ on $\mathbb{R}^d$
such that their supports are disjoint. Then the kernel $K_{s, \phi, \theta}(x,y)$ is continuous on $\mathbb{R}^d \times \mathbb{R}^d$ and moreover, it satisfies the following decay:
$$\sup_{s}|K_{s, \phi, \theta}(x,y)| \leq C_N |\phi(x)\theta(y)|\cdot |x-y|^{-N},$$
for any $N>d-2$. Here, the constant $C_N$ is independent of $\phi$ and $\theta$.
\end{prop}

\begin{proof}
By using partitions of unity, any smooth compactly supported function can be written as a finite sum of smooth functions in the set $\mathcal{P}$. Thus, we can assume without loss of generality that both $\phi$ and $\theta$ belong to $\mathcal{P}$.

First, observe that for any $(k,n) \in \mathcal{O} \times \mathbb{Z}^d$,
$$T_s(k+2\pi n)=\mathcal{M}_{n}^{-1}T_s(k) \mathcal{M}_{n},$$
where $\mathcal{M}_n$ is the multiplication operator on $L^2(\mathbb{T})$ by the exponential function $e^{2\pi in \cdot x}.$

Hence,
$$\nabla_{k}^{\alpha}K_s(k+2\pi n,\pi(x),\pi(y))=e^{-2\pi i n\cdot \pi(x)}\nabla_{k}^{\alpha}K_s(k,\pi(x),\pi(y))e^{2\pi i n\cdot \pi(y)},$$
for any multi-index $\alpha$.
Since $e^{2\pi i n \cdot x}=e^{2\pi i n \cdot \pi(x)}$ for any $x \in \mathbb{R}^d$, we obtain
\begin{equation}
\label{perboundary}
e^{i(k+2\pi n)\cdot(x-y)}\nabla_{k}^{\alpha}K_s(k+2\pi n,\pi(x),\pi(y))=e^{ ik\cdot(x-y)}\nabla_{k}^{\alpha}K_s(k,\pi(x),\pi(y)).
\end{equation}

Applying Lemma \ref{kernel_formula}, we then use integration by parts (all boundary terms vanish when applying integration by parts due to \mref{perboundary}) to derive that for any $|\alpha|=N$,
\begin{equation*}
(2\pi)^d(i(x-y))^{\alpha}K_{s,\phi,\theta}(x,y)=\phi(x)\theta(y)\int_{\mathcal{O}} e^{ik\cdot (x-y)}\nabla^{\alpha}_{k} K_s(k,\pi(x), \pi(y))dk.
\end{equation*}

Suppose $N>d-2$. Then by applying Corollary \ref{kernel_estimate}, the above integral is absolutely convergent and it is also uniformly bounded in $(s,x,y)$. Consequently, the kernel $K_{s,\phi, \theta}(x,y)$ is continuous. Furthermore,
$$\sup_{s}|K_{s,\phi, \theta}(x,y)| \lesssim  |\phi(x)\theta(y)|\cdot\min_{|\alpha|=N}|(x-y)^{\alpha}|^{-1} \lesssim  |\phi(x)\theta(y)|\cdot |x-y|^{-N}.$$

\end{proof}
We now have enough tools to approach our goal:
\begin{thm}
\label{rapid_decay_kernel}
The Schwartz kernel $K_s(x,y)$ is continuous away from the diagonal and furthermore, as $|x-y| \rightarrow \infty$, we have
$$\sup_{s}|K_s(x,y)|=O(|x-y|^{-N}), \quad \forall N>0.$$

\end{thm}

\begin{proof}
Let us fix a point $(s,x)$ in $\mathbb{S}^{d-1} \times \mathbb{R}^d$. Now we consider a point $y=x+st$, where $t$ is a real number. When $|t|>0$, we can choose two cut-off functions $\phi$ and $\theta$ such that $\phi$ and $\theta$ equal $1$ on some neighborhoods of $x$ and $y$, respectively, and also, the supports of these two functions are disjoint.
%We can assume further that the supports of $\phi$ and $\theta$ are small enough so that both $\phi$ and $\theta$ belong to $\mathcal{P}$.
Then, Proposition \ref{cutoffkernel} implies that the kernel $K_s(x,y)$ is continuous at $(x,y)$ since it coincides with $K_{s,\phi,\theta}$ on a neighborhood of $(x,y)$. This yields the first claim.
Again, by Proposition \ref{cutoffkernel}, we obtain
$$\sup_{s}|K_s(x,y)|=\sup_{s}|K_{s, \phi, \theta}(x,y)| \leq C_N |x-y|^{-N},$$
which proves the last claim.
\end{proof}

%%%%%%%%%%%%%%%%%%%%%%%%%%%%%%%%%%%
\section{Some results on parameter-dependent toroidal \textbf{$\Psi$DOs}}
The aim in this section is to provide some results needed to complete the proof of Theorem \ref{parametrices}. We adopt the approach of \cite{GUI} to periodic elliptic differential operators.

The next two theorems are straightforward modifications of the proofs for non-parameter toroidal $\Psi$DOs:
\begin{thm}
\textbf{(The asymptotic summation theorem)}
Given families of symbols $b_{l} \in \tilde{S}^{m-l}(\mathbb{T})$, where each family $b_l=\{b_l(s,k)\}_{(s,k)}$ for $l=0,1,\dots$, there exists a family of symbols $b$ in $\tilde{S}^m(\mathbb{T})$ such that
\begin{equation}
\label{E:asymptotic_summation}
\{b(s,k)-\sum_{i<l}b_{i}(s,k)\}_{(s,k)} \in \tilde{S}^{m-l}(\mathbb{T}).
\end{equation}
We will write $\displaystyle b \sim \sum_{l}b_l$ if $b$ satisfies  \mref{E:asymptotic_summation}.
\end{thm}
\begin{proof}
\textit{Step 1}.
Let $n=m+\epsilon$ for some $\epsilon>0$. Then
$$|b_l(s,k;x,\xi)| \leq C_l(1+|\xi|)^{m-l}=\frac{C_l (1+|\xi|)^{n-l}}{(1+|\xi|)^{\epsilon}}.$$
Thus, there is a sequence $\{\eta_l\}_{l \geq 1}$ such that $\eta_l \rightarrow +\infty$ and
$$|b_l(s,k;x,\xi)|<\frac{1}{2^l}(1+|\xi|)^{n-l}$$
for $|\xi|>\eta_l$. Let $\rho \in C^{\infty}(\mathbb{R})$ satisfy that $0 \leq \rho \leq 1$, $\rho(t)=0$ whenever $|t|<1$ and $\rho(t)=1$ whenever $|t|>2$. We define:
\begin{equation*}
b(s,k;x,\xi)=\sum_{l} \rho \left(\frac{|\xi|}{\eta_l}\right)b_l(s,k;x,\xi).
\end{equation*}
Since only a finite number of summands are non-zero on any compact subset of $\mathbb{T} \times \mathbb{R}^d$, $b(s,\cdot;\cdot,\cdot) \in C^{\infty}(\mathcal{O} \times \mathbb{T} \times \mathbb{R}^d)$. Moreover, $b(s,k)-\sum_{r<l} b_r(s,k)$ is equal to:

\begin{equation*}
\sum_{r<l}\left(\rho \left( \frac{|\xi|}{\eta_r}\right)-1\right)b_r(s,k)+b_l(s,k)+\sum_{r>l}\rho\left(\frac{|\xi|}{\eta_r}\right)b_r(s,k).
\end{equation*}

The first summand is compactly supported while the second summand is in $S^{m-l}(\mathbb{T})$. Now let $\epsilon<1$. Then, the third summand is bounded from above by
$$\sum_{r>l}\frac{1}{2^r}(1+|\xi|)^{n-r} \leq (1+|\xi|)^{n-l-1} \leq (1+|\xi|)^{m-l}.$$
%which is less than $(1+|\xi|)^{n-l-1}$, and so for $\epsilon<1$, this is less than $(1+|\xi|)^{m-l}$.
Consequently,
$$\sup_{s \in \mathbb{S}^{d-1}}\left|b(s,k)-\sum_{r<l} b_r(s,k)\right| \leq C(1+|\xi|)^{m-l}.$$

\textit{Step 2}.
For $|\alpha|+|\beta|+|\gamma| \leq N$, one can choose $\eta_l$ such that
$$\sup_{s \in \mathbb{S}^{d-1}}|D^{\alpha}_k D^{\beta}_{\xi} D^{\gamma}_x b_l(s,k;x,\xi)| \leq \frac{1}{2^l}(1+|\xi|)^{n-l-|\alpha|-|\beta|}$$
for $\eta_l<|\xi|$. The same argument as in Step 1 implies that
\begin{equation}
\label{step2asymptotic}
\sup_{s \in \mathbb{S}^{d-1}}|D^{\alpha}_k D^{\beta}_{\xi} D^{\gamma}_x (b(s,k)-\sum_{r<l}b_r(s,k))| \leq C_N (1+|\xi|)^{m-l-|\alpha|-|\beta|}.
\end{equation}

\textit{Step 3}.
The sequence of $\eta_l$'s in Step 2 depends on $N$. We denote this sequence by $\eta_{l,N}$ to indicate this dependence on $N$. By induction, we can assume that for all $l$, $\eta_{l,N} \leq \eta_{l,N+1}$. Applying the Cantor diagonal process to this family of sequences, i.e., let $\eta_l=\eta_{l,l}$ then $b$ has the property \mref{step2asymptotic} for every $N$.
\end{proof}

\begin{thm}
\textbf{(The composition formula)}
\label{composition_formula}
Let $a=\{a(s,k)\}$ be a family of symbols in $\tilde{S}^l(\mathbb{T})$ and $P(x,D)=\sum_{|\alpha| \leq m}a_{\alpha}(x)D^{\alpha}$ be a differential operators of order $m \geq 0$ with smooth coefficients $a_{\alpha}(x)$. Then the family of periodic pseudodifferential operators $\{P(x,\xi+k+i\beta_s)Op(a(s,k))\}_{(s,k)} \in Op(\tilde{S}^{l+m}(\mathbb{T}))$. Indeed, we have:
$$P_s(k) Op(a(s,k))=Op(P\circ a) (s,k)),$$
where
\begin{equation}
\label{E:composition_formula}
(P\circ a) (s,k;x,\xi)=\sum_{|\beta| \leq m}\frac{1}{\beta !}D^{\beta}_{\xi}P_s(k)(x,\xi)D^{\beta}_x a(s,k;x,\xi)
\end{equation}
and
$$P_s(k)(x,\xi)=P(x,\xi+k+i\beta_s).$$
\end{thm}

\begin{proof}
The composition formula \mref{E:composition_formula} is obtained for each $(s,k)$ is standard in pseudodifferential operator theory (see e.g., \cite{GUI, Ruzh-Turu, Shubin_pseudo}). We only need to check that the family of symbols $\{(P \circ a) (s,k;x,\xi)\}_{(s,k)}$  is in $\tilde{S}^{l+m}(\mathbb{T})$.
But this fact follows easily from \mref{E:composition_formula} and Leibnitz's formula.
\end{proof}
We now finish the proof of Theorem \ref{parametrices}.
\begin{thm}
\label{Inversion}
\textbf{(The inversion formula)}
There exists a family of symbols $a=\{a(s,k)\}_{(s,k)}$ in $\tilde{S}^{-2}(\mathbb{T})$ and a family of symbols $r=\{r(s,k)\}_{(s,k)}$ in $\tilde{S}^{-\infty}(\mathbb{T})$ such that
$$(L_s(k)-\lambda)Op(a(s,k))=I-Op(r(s,k)).$$
\end{thm}
\begin{proof}
Let
$$L_0(s,k;x,\xi):=\sum_{|\alpha|=2}a_{\alpha}(x)(\xi+k+i\beta_s)^{\alpha},$$
$$\|a\|_{\infty}:=\sum_{|\alpha|=2}\|a_{\alpha}(\cdot)\|_{L^{\infty}(\mathbb{T})},$$
and
$$M:=\max_{(s,k) \in \mathbb{S}^{d-1} \times \mathcal{O}}\left(|k|^2+\theta^{-1}\|a\|_{\infty}|\beta_s|^2+\theta^{-1}\right),$$
where $\theta$ is the ellipticity constant in \mref{E:ellipticity}.
Whenever $|\xi|>(2M)^{1/2}$,
\begin{align*}
|L_0(s,k;x,\xi)| \geq \Re(L_0(s,k;x,\xi)) &\geq \theta|\xi+k|^2-\sum_{|\alpha|=2}a_{\alpha}(x)(\beta_s)^{\alpha}
\\ & \geq \theta \left(\frac{|\xi|^2}{2}-|k|^2 \right)-\|a\|_{\infty}|\beta_s|^2>1.
\end{align*}
Let $\rho \in C^{\infty}(\mathbb{R})$ be a function satisfying $\rho(t)=0$ when $|t|<(2M)^{1/2}$ and $\rho(t)=1$ when $|t|>2M^{1/2}$. We define the function
\begin{equation}
\label{a_0}
a_0(s,k)(x,\xi)=\rho(|\xi|)\frac{1}{L_0(s,k;x,\xi)}.
\end{equation}
Then $a_0:=\{a_0(s,k)\}_{(s,k)}$ is well-defined and belongs to $\tilde{S}^{-2}(\mathbb{T})$. The next lemma is the final piece we need to complete the proof of the theorem.
\begin{lemma}
\label{L:inversion}

(i) If $b=\{b(s,k)\}_{(s,k)} \in \tilde{S}^l(\mathbb{T})$ then $b- (L-\lambda) \circ (a_0 b) \in \tilde{S}^{l-1}(\mathbb{T})$.

(ii) There exists a sequence of families of symbols $a_l=\{a_l(s,k)\}_{(s,k)}$ in $\tilde{S}^{-2-l}(\mathbb{T}), l=0,1,\dots$ and a sequence of families of symbols $r_l=\{r_l(s,k)\}_{(s,k)}$ in $\tilde{S}^{-l}(\mathbb{T}), l=0,1,\dots$ such that $a_0$ is the family of symbols in \mref{a_0}, $r_0(s,k)=1$ for every $(s,k)$ and for all $l$,
$$(L-\lambda)\circ a_l=r_l -r_{l+1}.$$
\end{lemma}

\begin{proof}
(i) Let $p(s,k)=(L(s,k)-\lambda)(x,\xi)-L_0(s,k;x,\xi)$ so that $p=\{p(s,k)\}_{(s,k)} \in \tilde{S}^{1}(\mathbb{T})$ and hence, $p \circ (a_0 b)$ is in $\tilde{S}^{l-1}(\mathbb{T})$ due to Theorem \ref{composition_formula}.
Moreover, $b-L_0a_0b=(1-\rho(|\xi|))b$ is a family of symbols whose $\xi$-supports are compact and thus it is in $\tilde{S}^{-\infty}(\mathbb{T})$. We can now derive again from the composition formula \mref{E:composition_formula} when $P:=L_0$ that
$$(L-\lambda) \circ (a_0 b)=L_0 \circ (a_0 b)+p \circ (a_0 b)=L_0 a_0b+\dots=b+\dots,$$
where the dots are the terms in $\tilde{S}^{l-1}(\mathbb{T})$.

(ii) Recursively, let $a_l=a_0 r_l$ and $r_{l+1}=r_l - (L-\lambda) \circ a_l$. By part (i), $r_{l+1} \in \tilde{S}^{-(l+1)}(\mathbb{T})$.
\end{proof}

Now let $a$ be the asymptotic sum of the families of symbols $a_l$, i.e., $a \sim \sum_{l} a_l$. Then
$$(L-\lambda)\circ a \sim \sum_{l} (L-\lambda) \circ a_l=\sum_l r_l - r_{l+1}=r_0=1,$$
which implies that $1-(L-\lambda) \circ a \sim 0$. In other words, this means that $r:=1-(L-\lambda) \circ a \in \tilde{S}^{-\infty}(\mathbb{T})$. Hence, there exists a family of symbols $a$ in $\tilde{S}^{-2}(\mathbb{T})$ and a family of symbols $r$ in $\tilde{S}^{-\infty}(\mathbb{T})$ satisfying $(L-\lambda) \circ a=1-r$.
Finally, an application of Theorem \ref{composition_formula} completes the proof of Theorem \ref{parametrices}.
\end{proof}
%%%%%%%%%%%%%%%%%%%%%%%%%%%%%%%%%%%%%%%%%%%%%%%
\section{Some auxiliary statements}
\subsection{A lemma on the principle of non-stationary phase}
\begin{lemma}
\label{L:uniform}
Let $M$ be a compact manifold (with or without boundary) and $a:\mathbb{R} \times M \rightarrow  \mathbb{C}$ be a smooth function with compact support.
Then for any $N>0$, there exists a constant $C_N>0$ so that the following estimate holds for any non-zero $t \in \mathbb{R}$:
\begin{equation}
\label{E:non-stationary}
\sup_{x \in M}\left|\int_{-\infty}^{\infty} e^{ity}a(y,x)dy\right| \leq C_N |t|^{-N}.
\end{equation}
Here $C_N$ depends only on $N$, the diameter $R$ of the $y$-support of $a$ and $\displaystyle \sup_{x,y}|\partial_{y}^{N} a|$.
\end{lemma}
\begin{proof}
Let $t \neq 0$. Applying integration by parts repeatedly ($N$-times), it follows that
\begin{equation*}
\left|\int_{-\infty}^{\infty} e^{ity}a(y,x)dy\right|=|t|^{-N}\left|\int_{-\infty}^{\infty} e^{ity}\nabla_{y}^{N}a(y,x)dy\right|
\leq R \sup_{x,y}|\partial_{y}^{N} a|\cdot|t|^{-N}.
\end{equation*}
\end{proof}

\subsection{The Weierstrass preparation theorem}
\begin{thm}
\label{Weierstrass}
%(Weierstrass preparation theorem)
Let $f(t,z)$ be an analytic function of $(t,z) \in \mathbb{C}^{1+n}$ in a neighborhood of $(0,0)$ such that $(0,0)$ is a simple zero of $f$, i.e.:
\begin{equation*}
f(0,0)=0, \quad \frac{\partial f}{\partial t}(0,0)\neq 0.
\end{equation*}
Then there is a unique factorization
\begin{equation*}
f(t,z)=(t-A(z))B(t,z),
\end{equation*}
where $A,B$ are analytic in a neighborhood of $0$ and $(0,0)$ respectively. Moreover, $B(0,0) \neq 0$ and $A(0)=0$.
\end{thm}
The proof of a more general version of this theorem could be found in Theorem 7.5.1 in \cite{Hormander}.
\subsection{Proofs of Proposition 4.1 and Lemma 5.3}\label{subsec:Proposition proofs}
\begin{remark}
\label{R:CR}
Consider a domain $\mathcal{D}$ of $\mathbb{C}^d$ and let $\displaystyle f: \mathcal{D} \rightarrow \mathbb{C}$ be a holomorphic function. For $z \in \mathbb{C}^d$, write $z=x+iy$ where $x,y \in \mathbb{R}^d$. Now we fix a vector $\beta$ in $\mathbb{R}^d$ and denote $\mathcal{D}_{\beta}=(\mathcal{D}-i\beta) \cap \mathbb{R}^d$. If this intersection is non-empty, we may consider the restriction $k \rightarrow f(k+i\beta)$ as a real analytic function defined on a subdomain $\mathcal{D}_{\beta}$ of $\mathbb{R}^d$. Thanks to Cauchy-Riemann equations of $f$, we do not need to make any distinction between derivatives of $f$ with respect to $x$ (when $f$ is viewed as a real analytic one) or $z$ (when $f$ is considered as a complex analytic one) at every point in $\mathcal{D}_{\beta}$ since
\begin{equation*}
\frac{\partial f}{\partial x_l}(k+i\beta)=\frac{\partial f}{\partial z_l}(k+i\beta)=-i\frac{\partial f}{\partial y_l}(k+i\beta), \quad 1 \leq l \leq d.
\end{equation*}
For higher order derivatives, we use induction and the above identity to obtain
\begin{equation*}
\partial_x^{\alpha}f(k+i\beta)=\partial_z^{\alpha}f(k+i\beta)=(-i)^{|\alpha|}\partial_y^{\alpha}f(k+i\beta),
\end{equation*}
for any multi-index $\alpha$. We  use these facts implicitly for the function $\lambda_j$.
When dealing with the analytic function $f=\lambda_j$ in this part, denote $\partial^{\alpha}\lambda_j$ to indicate either its $x$ or $z$-derivatives.
%Recall the function $E$, which is defined via $E(\beta)=\lambda_j(i\beta+k_0)$ for any $\beta \in \mathbb{R}^d$.

We also want to mention this simple relation between derivatives of $\lambda_j$ and $E$:
\begin{equation*}
\partial^{\alpha}E(\beta)=\partial_y^{\alpha}\lambda_j(k_0+i\beta)=i^{|\alpha|}\partial^{\alpha}\lambda_j(k_0+i\beta).
\end{equation*}
\end{remark}

\begin{prmainprop}
We recall from Section 2 that $V$ is an open neighborhood of $k_0$ in $\mathbb{C}^d$ such that the properties \textbf{(P1)-(P6)} are satisfied. Note that $V$ depends only on the local structure at $k_0$ of the dispersion branch $\lambda_j$ of $L$. Denote $\mathcal{O}_{s}=\{k + it\beta_{s}: k \in \mathcal{O}, t \in [0,1]\}$ for each $s \in \mathbb{S}^{d-1}$. For $C>0$
(which we define it later), set $M_{s,C}=\mathcal{O}_{s}\cap \{z \in \mathbb{C}^d: |\Re(z)-k_0|<C\}$ and $N_{s,C}=\mathcal{O}_{s}\setminus M_{s,C}$.
For $C$ and $|\lambda|$ small enough, we can suppose $M_{s,C} \Subset V$ since $\beta_s$ is small too. We also assume that $|\lambda|\leq \epsilon_0$.\footnote{Recall the definition of $\epsilon_0$ from \textbf{(P3)}.}

%\vspace{-1cm}  %without the indent command%
%\vspace{-0.25cm}
\begin{figure}[h]
\begin{tikzpicture}
%\shade[top color = gray!50, bottom color=gray!50]
   %   (-0.5,0) rectangle (0.5,2) |- (0,0);
%\draw[fill, color = gray!50] (-1,0) rectangle (1,2) |- (0,0);
\draw[fill, color = gray!50] (-3.5,0) rectangle (-1,2);
\draw[fill, color = gray!50] (1,0) rectangle (3.5,2);
\draw[->] (-5,0) -- (5,0) node [below]{$\mathbb{R}^d$};
\draw[->] (0,2) -- (0,4) node [left]{$i\mathbb{R}^d$};
\draw[thick] (-3.5,0) rectangle (3.5, 2);
\draw[dashed] (1, 0) -- (1,2);
\draw[dashed] (-1, 0) -- (-1,2);
\draw[fill] (0,2) circle (2pt);
\draw (0,2.3) node [right]{$i \beta_s$};
\draw[thick, decorate,decoration={brace,amplitude=6pt, mirror},xshift=0pt,yshift=-3pt] (-3.5,0) -- (3.5,0);
\draw (0,-0.2) node [below]{$[-\pi, \pi]^d$};
\draw (0,1) node{$M_{s,C}$};
\draw (-2.25,1) node{$N_{s, C}$};
\draw (2.25,1) node{$N_{s, C}$};
\draw (0,-0.2) node [below]{$[-\pi, \pi]^d$};
%\foreach \x/\xtext in {0/{0}}
%    \draw[shift={(\x,0)}]  (0pt,2pt) -- (0pt,-2pt) node[above] {$\xtext$};
\draw[fill] (0,0) circle (2pt);
\draw (0,0.3) node {0};
\end{tikzpicture}
\caption{An illustration of the regions $M_{s,C}$ and $N_{s,C}$ when $k_0\equiv 0$.}
\end{figure}
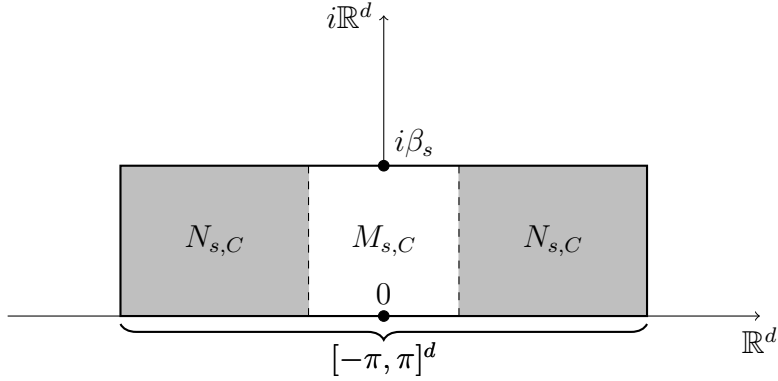
%\vspace{-0.25cm}
%\vspace{-0.8cm} %without the indent command%
For any point $z=k+it\beta_s \in M_{s,C}$, we want to show if $\lambda \in \sigma(L(z))$, it forces $z=k_0+i\beta_s$. By \textbf{(P3)}, this is the same as showing the equation $\lambda_{j}(z)=\lambda$ has no solution $z$ in $M_{s,C}$ except for the trivial solution $z=k_0+i\beta_s$. Suppose for contradiction $\lambda_{j}(k+it\beta_s)=\lambda=\lambda(\beta_s)$ for some $t \in [0,1]$ and $k$ in $\{k \in \mathcal{O} \mid 0<|k-k_0|<C\}$. By Taylor expanding around $k_0+it\beta_s$, there is some $\gamma \in (0,1)$ such that
\begin{equation}
\label{E:taylor_singularity}
\begin{split}
&\lambda-\lambda_{j}(k_0+it\beta_s)=\left((k-k_0)\cdot \nabla \lambda_{j}(k_0+it\beta_s)+\sum_{|\alpha|=3}\frac{(k-k_0)^{\alpha}}{\alpha !}\partial^{\alpha}\lambda_j (k_0+it\beta_s)\right)\\
&+\left(\frac{1}{2}(k-k_0)\cdot \Hess{(\lambda_j)}(k_0+it\beta_s)(k-k_0)
+\sum_{|\alpha|=4}\frac{(k-k_0)^{\alpha}}{\alpha !}\partial^{\alpha}\lambda_j(\gamma (k-k_0)+k_0+it\beta_s)\right).
\end{split}
\end{equation}
If $|\alpha|$ is odd, then by Remark \ref{R:CR} and the fact that $E$ is real, we have
\begin{equation*}
\partial^{\alpha}\lambda_j(k_0+it\beta_s)=\frac{1}{i^{|\alpha|}}\partial^{\alpha}E(t\beta_s) \in i\mathbb{R}.
\end{equation*}
Taking the real part of the equation \mref{E:taylor_singularity} to get
\begin{equation*}
\begin{split}
E(\beta_s)-E(t\beta_s)=-\frac{1}{2}(k-k_0)\cdot \Hess&(E)(t\beta_s)(k-k_0)+\sum_{|\alpha|=4}\frac{(k-k_0)^{\alpha}}{\alpha !}\times\\
&\times \Re(\partial^{\alpha}\lambda_j(\gamma (k-k_0)+k_0+it\beta_s)).
\end{split}
\end{equation*}
The left-hand side is bounded above by $(1-t)\lambda \leq 0$ because of the concavity of $E$ ($E(t\beta_s) \geq tE(\beta_s)=t\lambda$). On the other hand, by \textbf{(P5)},
\begin{equation*}
-\frac{1}{2}(k-k_0)\cdot \Hess{(E)}(t\beta_s)(k-k_0) \geq \frac{1}{4}|k-k_0|^2 \min\sigma(\Hess{(\lambda_{j})}(0))
\end{equation*}
\begin{equation*}
\left|\sum_{|\alpha|=4}\frac{(k-k_0)^{\alpha}}{\alpha !}\Re(\partial^{\alpha}\lambda_j(\gamma (k-k_0)+k_0+it\beta_s))\right| \leq C(d)|k-k_0|^{4}\max_{z \in \overline{V}, |\alpha|=4}|\partial^{\alpha}\lambda_{j}(z)|.
\end{equation*}
We simply choose $\displaystyle C^2<\frac{\min\sigma(\Hess{(\lambda_{j})}(0))}{C(d)\max_{z \in \overline{V}, |\alpha|=4}|\partial^{\alpha}\lambda_{j}(z)|}$ to get a contradiction if $k \neq k_0$. \\

For the remaining part, we just need to treat points $k+it\beta_s$ in $N_{s,C}$. We have $\lambda \in \rho(L(k)), \forall k \in \mathbb{R}^d$. The idea is to adapt the upper-semicontinuity of the spectrum of an analytic family of type $\mathcal{A}$ on $\mathbb{C}^d$, following \cite{Ka}. For any $k \in \mathcal{O}$ and $z \in \mathbb{C}^d$, the composed operators $(L(k+z)-L(k))(L(k)-\lambda)^{-1}$ are closed and defined on $L^2(\mathbb{T})$ and by closed graph theorem, these are bounded operators. Clearly,
$$\displaystyle L(k+z)-\lambda=(1+(L(k+z)-L(k))(L(k)-\lambda)^{-1})(L(k)-\lambda).$$
Thus, $\lambda$ is in the resolvent of $L(k+z)$ if the operator  $\displaystyle 1+(L(k+z)-L(k))(L(k)-\lambda)^{-1}$ is invertible.
Hence, it is enough to show that there is some positive constant $\tau$ such that for any $k \in \mathcal{O}$ and $|z|<\tau$,
\begin{equation}
\label{E:semicontinuity}
\|(L(k+z)-L(k))(L(k)-\lambda)^{-1}\|_{op}<1/2, \quad |k-k_0|\geq C,
\end{equation}
where the operator norm on $L^2(\mathbb{T})$ is denoted by $\|\cdot\|_{op}$.
Indeed, if $|\lambda|$ is small enough so that we have $\displaystyle \max_{s \in \mathbb{S}^{d-1}}|\beta_s|<\tau$ and then \mref{E:semicontinuity} implies that $\lambda \in \rho(L(k+it\beta_s)$ for any $t \in [0,1]$.

Finally, we will use some energy estimates of linear elliptic equations and spectral theory to obtain \mref{E:semicontinuity}. Observe that,
\begin{equation*}
L(k+z)-L(k)=z\cdot A(x)(D+k)+(D+k)\cdot A(x)z+z\cdot A(x)z.
\end{equation*}
For $v \in H^1(\mathbb{T})$ and $|z|<1$, there is some constant $C_1>0$ (independent of $z$) such that
\begin{equation}
\label{E:energy1}
\|(z\cdot A(x)(D+k)+(D+k)\cdot A(x)z+z\cdot A(x)z)v\|_{L^2(\mathbb{T})} \leq C_1|z|\cdot\|v\|_{H^{1}(\mathbb{T})}.
\end{equation}
Set $v:=(L(k)-\lambda)^{-1}u$ for $u \in L^2(\mathbb{T})$. Ellipticity of $L(k)$ yields $v \in H^{2}(\mathbb{T})$ and in particular, we obtain \mref{E:energy1} for such $v$. Testing the equation $(L(k)-\lambda)v=u$ with the function $v$, we derive the standard energy estimate
\begin{equation}
\label{E:energy2}
\|Dv\|_{L^2(\mathbb{T})} \leq C_2(\|v\|_{L^2(\mathbb{T})}+\|u\|_{L^2(\mathbb{T})}).
\end{equation}
Note that both $C_1$ and $C_2$ in \mref{E:energy1} and \mref{E:energy2} are independent of $k$ and $\lambda$ since we take $k$ in the bounded set $\mathcal{O}$ and consider $|\lambda|$ to be small enough.
%since
%$k$ is in the bounded set $\mathcal{O}$ and
%we take only small values of $\lambda$.

Suppose that $|\lambda|$ is less than one-half of the length of the gap between the dispersion branches $\lambda_{j}$ and $\lambda_{j-1}$.
Due to functional calculus of the self-adjoint operator $L(k)$, we get
$$\|(L(k)-\lambda)^{-1}\|_{op}=dist(\lambda, \sigma(L(k)))^{-1}=\min\{(\lambda_{j}(k)-\lambda), (\lambda- \lambda_{j-1}(k))\}^{-1}.$$

Now let $\delta_{1}=-\frac{1}{2}\max \lambda_{j-1}(k)>0$ and $\displaystyle \delta_2=\min_{k \in \mathcal{O}, |k-k_0| \geq C}\lambda_{j}(k)$. Then due to \textbf{A3}, $\delta_2>0$. Moreover,
$$\lambda - \lambda_{j-1}(k)>\lambda - \max_{k \in \mathcal{O}}\lambda_{j-1}(k)>\delta_{1},$$
and
$$\lambda_{j}(k)-\lambda \geq \min_{k \in \mathcal{O}, |k-k_0| \geq C}\lambda_{j}(k)-\lambda>\delta_2.$$
Hence,
\begin{equation}
\label{E:resolvent_bound}
\|(L(k)-\lambda)^{-1}\|_{op}<\delta:=\min\{\delta_1,\delta_2\}^{-1}.
\end{equation}

In other words, $\|v\|_{L^2(\mathbb{T})} \leq \delta \|u\|_{L^2(\mathbb{T})}$. Applying this fact together with \mref{E:energy1} and \mref{E:energy2}, we have
\begin{equation*}
\begin{split}
\|(L(k+z)-L(k))(L(k)-\lambda)^{-1}u\|_{L^2(\mathbb{T})} &\leq |z|C_1\|v\|_{H^{1}(\mathbb{T})}\\
&\leq |z|C_1C_2(\|v\|_{L^2(\mathbb{T})}+\|u\|_{L^2(\mathbb{T})})\\
&\leq |z|C_1C_2(1+\delta)\|u\|_{L^2(\mathbb{T})}.
\end{split}
 \end{equation*}
Now \mref{E:semicontinuity} is a consequence of the above estimate if we let
$$\displaystyle \tau\leq \min{\left(\frac{1}{2C_1C_2(1+\delta)},1\right)}.$$
\end{prmainprop}
%%%%%%%%%%%%%%%%%%%%%%%%%%%%%%%%%%%%%%%%%%%%%%%%%%%%%%%%%%

\begin{prmainlem}
From Lemma \ref{L:Bloch_variety}, the complex Bloch variety $\Sigma:=B_{L}$ of the operator $L$ is an analytic subset of codimension one in $\mathbb{C}^{d+1}$. By \cite{K, Wilcox}, there exist an entire scalar function $h(k,\mu)$ and an entire operator-valued function $I(k,\mu)$ on $\mathbb{C}^{d+1}$ such that

(1) $h$ vanishes only on $\Sigma$ and has simple zeros at all points of $\Sigma$.

(2) In $\mathbb{C}^{d+1} \setminus \Sigma$, $(L(k)-\mu)^{-1}=h(k,\mu)^{-1}I(k,\mu)$.

In particular, $(L_{t,s}(k)-\lambda)^{-1}=h(k+it\beta_s,\lambda)^{-1}I(k+it\beta_s,\lambda)$ for $k \in \mathbb{R}^{d}$ and $t \in [0,1)$ by Proposition \ref{singularity}. Due to \textbf{Assumption A} and \textbf{(P2)}, if $k_0+it\beta_s \in V$, the $k$-variable function $h(k,\lambda)^{-1}$
is equal (up to a non-vanishing analytic factor) to $(\lambda_{j}(k+it\beta_s)-\lambda)^{-1}$ on an open disc $D(k_0,2\varepsilon)\subseteq V$ in $\mathbb{C}^d$ for some $\varepsilon>0$. Hence, we can write the sesquilinear form for such values of $k$ as
\begin{equation*}
(R_{t,s,\lambda}f,\varphi)=R_1+R_2,
\end{equation*}
where
\begin{equation*}
R_1=(2\pi)^{-d}\int_{\mathcal{O}\cap D(k_0,\varepsilon)}\frac{\left(M(k,\lambda)\widehat{f}(k), \widehat{\varphi}(k)\right)}{\lambda_{j}(k+it\beta_s)-\lambda}dk
\end{equation*}
and
\begin{equation*}
R_2=(2\pi)^{-d}\int_{\mathcal{O}\setminus D(k_0,\varepsilon)}\left(L(k+it\beta_s)-\lambda)^{-1}\widehat{f}(k), \widehat{\varphi}(k)\right)dk.
\end{equation*}
Here $M(k,\lambda)$ is a $L^{2}(\mathbb{T})$-valued analytic function on $D(k_0,\varepsilon)$ when $|\lambda|$ is small. Since $f$ and $ \varphi$ have compact supports, their Floquet transforms $\widehat{f}(k), \widehat{\varphi}(k)$ are analytic with respect to $k$.
To prove the equality \mref{E:limit}, we apply the Lebesgue Dominated Convergence Theorem. For $R_1$, it suffices to show that the denominator in the integrand when $t \rightarrow 1^{-}$ is integrable over $D(k_0,\varepsilon)$ for $d\geq 2$. Indeed,
\begin{equation*}
\label{E:denominator_integrable1}
%\begin{split}
|\lambda_{j}(k+i\beta_{s})-\lambda|
%=\left|i\nabla E(\beta_s)\cdot (k-k_0)-\frac{1}{2}(k-k_0)\cdot Hess(E)(\beta_s)(k-k_0)+O(|k-k_0|^3)\right|\\
 \geq  \delta \left| i\nabla{E}(\beta_s)\cdot(k-k_0)-\frac{1}{2}(k-k_0)\cdot \Hess{(E)}(\beta_s)(k-k_0)\right |,
%\end{split}
\end{equation*}
for some $\delta>0$ if $\varepsilon$ is chosen small enough so that in the Taylor expansion of $\lambda_j$ at $k_0+i\beta_s$, the remainder term $O(|k-k_0|^3)$ is dominated by the quadratic term $|k-k_0|^2$. Furthermore,
\begin{equation*}
\label{E:denominator_integrable2}
\left| i\nabla{E}(\beta_s)\cdot |k-k_0|-\frac{1}{2}(k-k_0)\cdot \Hess{(E)}(\beta_s)(k-k_0)\right |^2\geq C(|k-k_0|^2+|k-k_0|^4),
\end{equation*}
for some constant $C>0$ (independent of $k$).
%the negative definite quadratic form $Hess(E)(\beta_s)$ and $|\nabla E(\beta_s)|$.
Now let $v:=(k-k_0)$ and so the right hand side of the above estimate is just $|v|^2+|v|^4$  (up to a constant factor). One can apply H\"older's inequality
%with $(p,q)=(4/3,4)$
to obtain
\begin{equation*}
\label{E:denominator_integrable3}
|v|^2+|v|^4 \geq |v_1|^{2}+|v'|^4 \geq C|v_1|^{3/2}|v'|,
\end{equation*}
where $v=(v_1,v')\in \mathbb{R} \times \mathbb{R}^{d-1}$.
Thus,
%from \mref{E:denominator_integrable1}, \mref{E:denominator_integrable2} and \mref{E:denominator_integrable3}
we deduce
\begin{equation}
\label{E:denominator_integrable4}
|\lambda_{j}(k+i\beta_{s})-\lambda|^{-1}\leq C|v_1|^{-3/4}|v'|^{-1/2}.
\end{equation}
Since the function $|x|^{-n}$ is integrable near 0 in $\mathbb{R}^d$ if and only if $n>d$, $|v'|^{-1/2}$ and $|v_1|^{-3/4}$ are integrable near 0 in $\mathbb{R}^{d-1}$ and $\mathbb{R}$ respectively. Therefore, the function in the right hand side of \mref{E:denominator_integrable4} is integrable near 0.

The integrability of $R_2$ as $t \rightarrow 1^{-}$ follows from the estimation \mref{E:semicontinuity} in the proof of Proposition \ref{singularity}. Indeed,
\begin{equation}
\label{E:R2_integrable1}
\begin{split}
\|(L(k+it\beta_s)-\lambda)^{-1}\|_{op}&=\|(1-(L(k+it\beta_s)-L(k))(\lambda-L(k))^{-1})^{-1}(\lambda-L(k))^{-1}\|_{op}\\
&\leq \frac{\|(L(k)-\lambda)^{-1}\|_{op}}{1-\|(L(k+it\beta_s)-L(k))(\lambda-L(k))^{-1}\|_{op}}.
\end{split}
\end{equation}
By decreasing $|\lambda|$, if necessary, and repeating the arguments when showing \mref{E:semicontinuity} and \mref{E:resolvent_bound} we derive:
\begin{equation}
\label{E:R2_integrable2}
1-\|(L(k+it\beta_s)-L(k))(\lambda-L(k))^{-1}\|_{op}\geq 1/2, \quad \forall k \in \mathcal{O}\setminus D(k_0,\varepsilon)
\end{equation}
and
\begin{equation}
\label{E:R2_integrable3}
\sup_{k \in\mathcal{O}\setminus D(k_0,\varepsilon)}\|(L(k)-\lambda)^{-1}\|_{op}<\infty.
\end{equation}
Thanks to \mref{E:R2_integrable1}, \mref{E:R2_integrable2}, \mref{E:R2_integrable3}, Cauchy-Schwarz inequality and Lemma \ref{L:floquet}, we have:
\begin{equation*}
\begin{split}
\sup_{t \in [0,1]}\left|\left(L(k+it\beta_s)-\lambda)^{-1}\widehat{f}(k), \widehat{\varphi}(k)\right)\right| &\leq 2\|(L(k)-\lambda)^{-1}\|_{op}\cdot\|\widehat{f}(k)\|_{L^2(\mathbb{T})}\|\widehat{\varphi}(k)\|_{L^2(\mathbb{T})}
\\ & \lesssim\|\widehat{f}(k)\|_{L^2(\mathbb{T})}\|\widehat{\varphi}(k)\|_{L^2(\mathbb{T})} \quad, \forall k \in \mathcal{O} \setminus D(k_0,\varepsilon)
\end{split}
\end{equation*}
and
\begin{equation*}
\int_{\mathcal{O} \setminus D(k_0,\varepsilon)}\|\widehat{f}(k)\|_{L^2(\mathbb{T})}\|\widehat{\varphi}(k)\|_{L^2(\mathbb{T})}dk \leq \|f\|_{L^2(\mathbb{R}^d)}\|\varphi\|_{L^2(\mathbb{R}^d)}<\infty.
\end{equation*}
This completes the proof of our lemma.
\end{prmainlem}
%%%%%%%%%%%%%%%%%%%%%%%%%%%%%%%%%%%%%%%%%%%%%%%%%%%
\subsection{Regularity of eigenfunctions $\phi(z,x)$}
$ $
\vspace{0.3cm}

In this subsection, we study the regularity of the eigenfunctions $\phi(z,x)$ of the operator $L(z)$ with corresponding eigenvalue $\lambda_j(z)$ (see \textbf{(P4)}).
It is known that for each $z \in V$, the eigenfunction $\phi(z,x)$ is smooth in $x$. We will claim that these eigenfunctions are smooth in $(z,x)$ when $z$ is near to $k_0.$ The idea is that initially, $\phi(z,\cdot)$ is an analytic section of the Hilbert bundle $V \times H^2(\mathbb{T})$ and then by ellipticity, it is also an analytic section of the bundle $V \times H^m(\mathbb{T})$ for any $m>0$ (for statements related to Fredholm morphisms between analytic Banach bundles, see e.g., \cite{ZKKP}) and hence smoothness will follow.

For the sake of completeness, we provide the proof of the above claim by applying standard bootstrap arguments in the theory of elliptic differential equations.
\begin{lemma}
\label{L:joint_continuity}
The function $\partial_{x}^{\alpha}\phi(z,x)$ is jointly continuous on $V \times \mathbb{R}^d$ for any multi-index $\alpha$.
\end{lemma}
\begin{proof}
By periodicity, it suffices to restrict $x$ to $\mathbb{T}$. Let $K:=\overline{V}$.
Due to \textbf{(P4)}, the function $z \mapsto \phi(z, \cdot)$ is a $H^2(\mathbb{T})$-valued analytic on some neighborhood of $K$. Thus,
\begin{equation*}
\label{E:initial}
\sup_{z \in K}\|\phi(z, \cdot)\|_{H^2(\mathbb{T})}<\infty.
\end{equation*}
Then, we can apply bootstrap arguments for the equation $$L(z)\phi(z, \cdot)=\lambda_j(z)\phi(z, \cdot)$$
to see that $M_m:=\sup_{z \in K}\|\phi(z, \cdot)\|_{H^m(\mathbb{T})}$ is finite for any nonnegative integer $m$.

Now we consider $z$ and $z'$ in $K$. Let $\phi_{z,z'}(x):=\phi(z,x)-\phi(z',x)$. Then, $\phi_{z,z'}$ is a (classical) solution of the equation
$$L(z)\phi_{z,z'}=f_{z,z'},$$
where $f_{z,z'}:=(\lambda_j(z)\phi(z, \cdot)-\lambda_j(z')\phi(z', \cdot))+(L(z')-L(z))\phi(z', \cdot)$.

By induction, we will show that for any $m \geq 0$,
\begin{equation}
\label{induction_lemm}
\|\phi_{z,z'}\|_{H^m(\mathbb{T})}\lesssim |z-z'|.
\end{equation}

The case $m=0$ is clear because \textbf{(P4)} implies that $z \mapsto \|\phi(z,\cdot)\|_{L^2(\mathbb{T})}$ is Lipschitz continuous.

Next, we assume that the estimate \mref{induction_lemm} holds for $m$.
As in \mref{E:energy1},
\begin{equation}
\label{E:energy1_lemm}
\|(L(z)-L(z'))\phi(z', \cdot)\|_{H^m(\mathbb{T})} \lesssim |z-z'| \cdot \|\phi(z', \cdot)\|_{H^{m+1}(\mathbb{T})}\lesssim M_{m+1}|z-z'|.
\end{equation}
Using triangle inequalities, the estimates \mref{induction_lemm}, \mref{E:energy1_lemm} and analyticity of $\lambda_j$, we get
\begin{equation}
\label{E:energy2_lemm}
\begin{split}
\|f_{z,z'}\|_{H^{m}(\mathbb{T})}& \lesssim \|\lambda_j(z)\phi(z, \cdot)-\lambda_j(z')\phi(z', \cdot)\|_{H^{m}(\mathbb{T})}+\|(L(z)-L(z'))\phi(z', \cdot)\|_{H^{m}(\mathbb{T})} \\
&\lesssim |\lambda_j(z)|\cdot \|\phi_{z,z'}\|_{H^{m}(\mathbb{T})}+M_m|\lambda_j(z)-\lambda_j(z')|+M_{m+1}|z-z'|
\\&\lesssim |z-z'|.
\end{split}
\end{equation}
Notice that for any $m \geq 0$, the following standard energy estimate holds (see e.g., \cite{Ev,GT,LU}):
\begin{equation}
\label{E:energy3_lemm}
\|\phi_{z,z'}\|_{H^{m+2}(\mathbb{T})}\lesssim \|f_{z,z'}\|_{H^m(\mathbb{T})}+\|\phi_{z,z'}\|_{L^2(\mathbb{T})}.
\end{equation}
Combining \mref{E:energy2_lemm} and \mref{E:energy3_lemm}, we deduce that $\|\phi_{z,z'}\|_{H^{m+2}(\mathbb{T})} \lesssim |z-z'|$. Hence, \mref{induction_lemm} holds for $m+2$. This finishes our induction.

Applying the Sobolev embedding theorem, $\|\phi_{z,z'}\|_{C^m(\mathbb{T})}\lesssim |z-z'|$ for any $m\geq 0$. In other words, $\phi \in C(K,C^m(\mathbb{T}))$ for any $m$. Since $C(K \times \mathbb{T})=C(K,C(\mathbb{T}))$, this completes the proof.
\end{proof}

\begin{notation}
Consider a $z$-parameter family of linear partial differential operators $\{L(z)\}$ where $z \in \mathbb{R}^d$. Suppose $L(x,\xi,z)$ is the symbol of $L(z)$. Whenever it makes sense, the differential operator $\displaystyle \frac{\partial L(z)}{\partial z_l}$ is the one whose symbol is $\displaystyle \frac{\partial L}{\partial z_l}(x,\xi,z)$ for any $l \in \{1,2,\dots,d\}$.
\end{notation}

\begin{prop}
\label{P:joint_continuity}
Assume $D$ is an open disc centered at $k_0$ in $\mathbb{R}^d$ such that $D\pm i\beta_s \Subset V$ for any $s \in \mathbb{S}^{d-1}$. Then all eigenfunctions $\displaystyle \phi(k\pm i\beta_s,x)$ are smooth on a neighborhood of $\overline{D} \times \mathbb{R}^d$. Furthermore, all derivatives of $\phi(k \pm i\beta_s,x)$ are bounded on $\overline{D} \times \mathbb{R}^d$ uniformly in $s$, i.e., for any multi-indices $\alpha, \beta$:
$$\sup_{(s,k,x) \in \mathbb{S}^{d-1}\times \overline{D} \times \mathbb{R}^d}|\partial^{\alpha}_{k}\partial^{\beta}_{x}\phi(k \pm i\beta_s,x)|<\infty.$$
\end{prop}
\begin{proof}
Pick any open disc $D'$ in $\mathbb{R}^d$ so that $\overline{D} \pm i\beta_s \subset D'\pm i\beta_s \subseteq V$. We will prove that all eigenfunctions are smooth on the domain $D' \times \mathbb{R}^d$. Also, it is enough to consider the function $\displaystyle \phi(k+i\beta_s)$ since the other one is treated similarly.

First, we show that $\displaystyle \frac{\partial \phi}{\partial k_l}(k+i\beta_s,x)$ is continuous for any $1 \leq l \leq d$. By Lemma \ref{L:joint_continuity}, the function  $(k,x) \mapsto \phi(k+i\beta_s,x)$ is continuous on $D' \times \mathbb{T}$. We consider any two complex-valued test functions $\varphi \in C^{\infty}_{c}(D')$ and $\psi \in C^{\infty}(\mathbb{T})$. Testing the equation of the eigenfunction $\phi(k+i\beta_s,x)$ with $\psi$ and $\displaystyle \frac{\partial \varphi}{\partial k_l}$, we derive
\begin{equation*}
\int_{D'}\int_{\mathbb{T}}(L(k+i\beta_s)-\lambda_j(k+i\beta_s))\phi(k+i\beta_s,x)\overline{\psi(x)}\frac{\partial \varphi}{\partial k_l}(k)dxdk=0.
\end{equation*}
Observe that $L(k+i\beta_s)^{*}=L(k-i\beta_s)$ and $\displaystyle \left(\frac{\partial L(k-i\beta_s)}{\partial k_l}\right)^*=\frac{\partial L(k+i\beta_s)}{\partial k_l}$. We integrate by parts to derive
\begin{equation}
\label{E:test_eqn}
\begin{split}
0&=\int_{D'}\left(\left(L(k+i\beta_s)-\lambda_j(k+i\beta_s)\right)\phi(k+i\beta_s,x),\psi(x)\right)_{L^2(\mathbb{T})}\frac{\partial \varphi}{\partial k_l}(k)dk\\
&=\int_{D'}\left(\phi(k+i\beta_s,x),\left(L(k-i\beta_s)-\overline{\lambda_j}(k+i\beta_s)\right)\psi(x)\right)_{L^2(\mathbb{T})}\frac{\partial \varphi}{\partial k_l}(k)dk\\
&=\int_{D'}\left(-\frac{\partial\phi}{\partial k_l}(k+i\beta_s,x),\left(L(k-i\beta_s)-\overline{\lambda_j}(k+i\beta_s)\right)\psi(x)\right)_{L^2(\mathbb{T})}\varphi(k)dk\\
&-\int_{D'}\left(\phi(k+i\beta_s,x),\frac{\partial L(k-i\beta_s)}{\partial k_l}\psi(x)-\frac{\partial \overline{\lambda_j}}{\partial k_l}(k+i\beta_s)\psi(x)\right)_{L^2(\mathbb{T})}\varphi(k)dk\\
&=\int_{D'}\left(\left(L(k+i\beta_s)-\lambda_j(k+i\beta_s)\right)\frac{\partial\phi}{\partial k_l}(k+i\beta_s,x),\psi(x)\right)_{L^2(\mathbb{T})}\varphi(k)dk\\
&-\int_{D'}\left(\left(\frac{\partial L(k+i\beta_s)}{\partial k_l}-\frac{\partial \lambda_j}{\partial k_l}(k+i\beta_s)\right)\phi(k+i\beta_s,x),\psi(x)\right)_{L^2(\mathbb{T})}\varphi(k)dk.
\end{split}
\end{equation}
We introduce
\begin{equation*}
\begin{split}
\phi_l(k,x)&:=\frac{\partial \phi}{\partial k_l}(k+i\beta_s,x),\\
G(k)&:=L(k+i\beta_s)-\lambda_j(k+i\beta_s),\\
H(k,x)&:=\left(\frac{\partial L(k+i\beta_s)}{\partial k_l}-\frac{\partial \lambda_j}{\partial k_l}(k+i\beta_s)\right)\phi(k+i\beta_s,x).
\end{split}
\end{equation*}
By invoking the previous lemma, the Lipschitz continuity of the $C^{m+2}(\mathbb{T})$-valued function $\phi(k+i\beta_s,\cdot)$ implies that the mapping $k \mapsto H(k,\cdot)$ must be Lipschitz as a $C^m(\mathbb{T})$-valued function on $\overline{D'}$ for any $m\geq 0$. On the other hand, the $H^2(\mathbb{T})$-valued function $\phi_l(k,\cdot)$ is also Lipschitz on $D'$ due to \textbf{(P4)}. Hence, both $(G(k)\phi_l(k,\cdot), \psi)_{L^2(\mathbb{T})}$ and $(H(k,\cdot),\psi)_{L^2(\mathbb{T})}$ are continuous on $D'$ for any test function $\psi$.
The continuity let us conclude from \mref{E:test_eqn} that for every $k \in D'$, $\phi_l(k,\cdot)$ is a weak solution of the equation
\begin{equation}
\label{E:eqn}
G(k)\phi_l(k,x)=H(k,x).
\end{equation}
We interpret \mref{E:eqn} in the classical sense since all the coefficients of this equation are smooth. Consider any $k_1, k_2$ in $D'$ and subtract the equation corresponding to $k_1$ from the one corresponding to $k_2$ to obtain the equation for the oscillation function $\phi_l(k_1,\cdot)-\phi_l(k_2,\cdot)$:
\begin{equation*}
G(k_1)(\phi_l(k_1,x)-\phi_l(k_2,x))=(G(k_2)-G(k_1))\phi_l(k_2,x)+(H(k_1,x)-H(k_2,x)).
\end{equation*}
Note that due to regularities of $\lambda_j$, $H$ and the fact that the differential operator $G(k)$ depends analytically on $k$, we get
$$\|H(k_1,\cdot)-H(k_2,\cdot)\|_{H^m(\mathbb{T})}+\|(G(k_1)-G(k_2))\phi_l(k_2,\cdot)\|_{H^m(\mathbb{T})}=O(|k_1-k_2|), \quad \forall m \in \mathbb{N}.$$
Combining this with the uniform boundedness in $k$ of the supremum norms of all coefficients of the differential operator $G(k_1)$, we obtain
$$\|\phi_l(k_1,\cdot)-\phi_l(k_2,\cdot)\|_{H^m(\mathbb{T})}=O(|k_1-k_2|),$$ by using energy estimates as in the proof of Lemma \ref{L:joint_continuity}. An application of the Sobolev embedding theorem
shows that $\partial_{x}^{\beta}\phi_l(k,x)$ is continuous on $D' \times \mathbb{T}$ for any multi-index $\beta$.

To deduce continuity of higher derivatives $\partial_{x}^{\beta} \partial_{k}^{\alpha}\phi(k+i\beta_s)$ ($|\alpha|>1$, $|\beta| \geq 0$), we induct on $|\alpha|$ and
repeat the arguments of the $|\alpha|=1$ case.

Finally, the last statement of this proposition also follows since all of our estimates hold uniformly in $s$.
\end{proof}
\begin{obs}
1. The property \textbf{(P4)} is crucial in order to bootstrap regularities of eigenfunctions $\phi(k\pm i\beta_s)$.

2. If one just requires $\phi(k \pm i\beta_s) \in C^m(\overline{D}\times \mathbb{R}^d)$ for certain $m>0$ then the smoothness on coefficients of $L$ could be relaxed significantly (see \cite{GT,LU}).
\end{obs}

%%%%%%%%%%%%%%%%%%%%%%%
%%%%%%%%%%%%%%%%%%%%%%%
\section{Concluding remarks}
\begin{enumerate}
\item The condition that the potentials $A$, $V$ are infinitely differentiable is an overkill. The Fredholm property of the corresponding Floquet operators is essential, which can be obtained under much weaker assumptions.

\item The main result of this article assumes the central symmetry (evenness) of the relevant branch of the dispersion curve $\lambda_j(k)$, which does not hold for instance for operators with periodic magnetic potentials \cite{Shterenberg,FKTassym}. Note that the result of \cite{KR} at the spectral edge does not require such a symmetry. It seems that in the inside-the-gap situation one also should not need such a symmetry. However, the authors have not been able to do so, and thus were limited to the case of high symmetry points of the Brillouin zone.

\item In the case when $\lambda$ is below the whole spectrum, the result of this paper implies the Theorem 1.1 in \cite{MT} for self-adjoint operators.

\end{enumerate}

%%%%%%%%%%%%%%%%%
\section{Acknowledgements}
%%%%%%%%%%%%%%%%%%

The work of all authors was partially supported by NSF DMS grants. The authors express their gratitude to NSF for the support. P.K. would like to thank the Isaac Newton Institute for Mathematical Sciences, Cambridge, for support and hospitality during the programme Periodic and Ergodic Problems, where work on this paper was undertaken.

\end{document}